\documentclass[sigplan,screen]{acmart}

\newif\ifextended
\extendedtrue

\ifextended
\fi

\setcopyright{rightsretained}
\acmPrice{}
\acmDOI{10.1145/3385412.3386009}
\acmYear{2020}
\copyrightyear{2020}
\acmSubmissionID{pldi20main-p439-p}
\acmISBN{978-1-4503-7613-6/20/06}
\acmConference[PLDI '20]{Proceedings of the 41st ACM SIGPLAN International Conference on Programming Language Design and Implementation}{June 15--20, 2020}{London, UK}
\acmBooktitle{Proceedings of the 41st ACM SIGPLAN International Conference on Programming Language Design and Implementation (PLDI '20), June 15--20, 2020, London, UK}

\usepackage{booktabs}
\usepackage{rotating}
\usepackage{subcaption}
\usepackage[utf8]{inputenc}
\usepackage[T1]{fontenc}
\usepackage{amsthm}
\usepackage{wrapfig}
\usepackage{xspace}
\usepackage{cleveref}
\usepackage{bbold}
\usepackage{tikz}
\usetikzlibrary{positioning}
\usetikzlibrary{calc}
\usepackage{pgfplots}
\usepackage{pgfplotstable}
\pgfplotsset{compat=newest}
\usepackage{pifont}
\usepackage{xargs}
\usepackage{enumitem}
\usepackage{balance}

\theoremstyle{definition}
\newtheorem*{theorem*}{Theorem}

\usepackage{xcolor}
\definecolor{pfcolor}{RGB}{110,180,227}
\newcommand{\pfbullet}{{\color{pfcolor}{\ensuremath{\blacksquare}}}\xspace}
\definecolor{bdscolor}{RGB}{229,158,12}
\newcommand{\bdsbullet}{{\color{bdscolor}{\ensuremath{\bullet}}}\xspace}
\definecolor{pmdscolor}{RGB}{147,3,211}
\newcommand{\pmdsbullet}{{\color{pmdscolor}{\ensuremath{\blacktriangledown}}}\xspace}
\definecolor{ndscolor}{RGB}{69,155,118}
\newcommand{\ndsbullet}{{\color{ndscolor}{\ensuremath{\blacktriangle}}}\xspace}

\usepackage{listings}
\usepackage{letltxmacro}
\newcommand*{\SavedLstInline}{}
\LetLtxMacro\SavedLstInline\lstinline
\DeclareRobustCommand*{\lstinline}{\ifmmode
    \let\SavedBGroup\bgroup
    \def\bgroup{\let\bgroup\SavedBGroup
      \hbox\bgroup
    }\fi
  \SavedLstInline
}

\lstdefinelanguage{zelus}
   {morekeywords={
	let,in,rec,where,end,if,then,else,do,done,run,
	open,
	fun,node,hybrid,proba,
	match,with,automaton,emit,
	pre,when,whenot,fby,merge,on,clock,
	or,and,not,as,mod,
	unless,until,continue,reset,every,await,
	init,der,type,last,
	period,local,present,sample, observe, factor, infer, eval
    },
    sensitive=true,
    morecomment=[n]{(*}{*)},
    morestring=[b]",
escapechar=\%,
    columns=fullflexible,
    keepspaces=true,
    basicstyle=\ttfamily,
    mathescape=true,
    }

\def\zl{\lstinline[basicstyle=\small\ttfamily]}
\def\zlm{\lstinline[basicstyle=\small\ttfamily]}
\def\zlmf{\lstinline[basicstyle=\footnotesize\ttfamily]}
\def\zlmm{\lstinline[basicstyle=\scriptsize\ttfamily]} 
\newcommand{\lustre}{Lustre\xspace}
\newcommand{\zelus}{Zelus\xspace}
\newcommand{\esterel}{Esterel\xspace}
\newcommand{\scade}{SCADE\xspace}
\newcommand{\simulink}{Simulink\xspace}

\newcommand{\gcdsacro}{SDS\xspace}

\newcommand{\muF}{$\mu F$\xspace}

\tikzstyle{initialized}=[circle, draw=black, minimum width=1ex]
\tikzstyle{marginalized}=[circle, draw=black, fill=gray!30, minimum width=0.2cm]
\tikzstyle{realized}=[circle, draw=black, fill=gray!80, minimum width=0.2cm]
\newcommandx{\link}[4][1={-stealth}, 4={}]{\draw[#1] (#2) to[bend left=25] node[pos=0.5] {#4} (#3)}
\newcommand{\xmark}{\text{\ding{55}}}
\newcommand{\rmark}{{\Large{\color{red!90!black}\text{\ding{55}}}}}

\renewcommand{\paragraph}[1]{\noindentparagraph{#1}}

\begin{document}

\title{Reactive Probabilistic Programming}

\author{Guillaume Baudart}
\affiliation{
  \institution{MIT-IBM Watson AI Lab,\\IBM Research}
  \country{USA}}

\author{Louis Mandel}
\affiliation{
  \institution{MIT-IBM Watson AI Lab,\\IBM Research}
  \country{USA}}

\author{Eric Atkinson}
\affiliation{
  \institution{MIT}
  \country{USA}}

\author{Benjamin Sherman}
\affiliation{
  \institution{MIT}
  \country{USA}}

\author{Marc Pouzet}
\affiliation{
  \institution{École Normale Supérieure,\\PSL Research University}
  \country{France}}

\author{Michael Carbin}
\affiliation{
  \institution{MIT}
  \country{USA}}

\renewcommand{\shortauthors}{}

\newcommand{\ProbZelus}{ProbZelus\xspace}

\begin{abstract}

Synchronous modeling is at the heart of programming languages like
\lustre, \esterel, or \scade{} used routinely for implementing safety
critical control software, e.g., fly-by-wire and engine control in planes.
However, to date these languages have had limited modern support for
modeling uncertainty --- probabilistic aspects of the software's
environment or behavior --- even though modeling uncertainty is a
primary activity when designing a control system.

In this paper we present \ProbZelus the first {\em synchronous
  probabilistic programming language}. 
\ProbZelus conservatively provides the
facilities of a synchronous language to write control software, with
probabilistic constructs to model uncertainties and perform {\em
  inference-in-the-loop}.

We present the design and implementation of the language. We propose a measure-theoretic semantics of probabilistic stream functions and a simple type discipline to
separate deterministic and probabilistic expressions. We demonstrate a
semantics-preserving compilation into a first-order functional language
that lends itself to a simple presentation of inference algorithms for streaming models.
We also redesign the delayed sampling inference algorithm to provide efficient \emph{streaming} inference.
Together with an evaluation on several reactive applications, our results
demonstrate that \ProbZelus enables the design of reactive probabilistic applications and efficient, bounded memory inference.

\end{abstract}



\maketitle

\section{Introduction}
\label{sec:introduction}

\newcommand{\bN}{\mathbb{N}}
\newcommand{\xo}{\mathit{xo}}
\newcommand{\Pre}[1]{\zl{pre}\,{#1}}
\newcommand{\Arrow}[2]{{#1}\; \zl{->}\; {#2}}
\newcommand{\Nil}{\mathit{nil}}

Synchronous languages~\cite{synchronous-twelve-years-later} were
introduced thirty years ago for designing and implementing real-time
control software.
They are founded on the synchronous abstraction~\cite{esterel:ifip89}
where a system is modeled ideally, as if communications and
computations were instantaneous and paced on a global clock. This
abstraction is simple but powerful: input, output and local signals
are streams that advance synchronously and a system is a stream
function. It is at the heart of the data-flow languages
\lustre~\cite{lustre:ieee91} and \scade~\cite{lucy:tase17};
it is also the underlying model
behind the discrete-time subset of \simulink.

The data-flow programming style is very well adapted to the direct
expression of the classic control blocks of control engineering (e.g., relays, filters, PID controllers, control logic), and a
discrete time model of the environment, with the feedback between the
two. For example, consider a backward Euler integration method
defined by the following stream equations and its corresponding
implementation in \zelus~\cite{lucy:hscc13}, a
language reminiscent of Lustre:
$$\begin{array}{@{}l}
  x_0 = \xo_0
  \qquad 
  x_n = x_{n-1} + x'_n \times h \quad \forall n \in \bN, n > 0
\end{array}$$

\begin{lstlisting}[aboveskip=0.25em]
    let node integr (xo, x') = x where 
      rec x = xo -> (pre x + x' * h)
\end{lstlisting}

\smallskip

\noindent
The \emph{node} \zl{integr} is a function from input streams \zl{xo} and
\zl{x'} to output stream \zl{x}. The \emph{initialization} operator \zl{->} returns its left-hand side value at the first time step and its right-hand side expression on every time step thereafter. The \emph{unit-delay} operator \zl{pre} returns the value of its expression at the previous time step.
The following table presents a sample \emph{timeline} showing the sequences of values taken by the streams defined in the program (where \zl{h} is set to \zl{0.1}).

\begin{center}
\begin{tabular}{l|rrrrrrrr}
\zlmf|xo|      &\zlmf|0|  &\zlmf|0|&\zlmf|0|&\zlmf|0|&\zlmf|0| &\zlmf|0| &\zlmf|0| & \dots \\
\zlmf|x'|      &\zlmf|1|  &\zlmf|2|&\zlmf|1|&\zlmf|0|&\zlmf|-1|&\zlmf|-1|&\zlmf|1| & \dots \\
\zlmf|x' * h|  &\zlmf|0.1|&\zlmf|0.2|&\zlmf|0.1|&\zlmf|0|  &\zlmf|-0.1|&\zlmf|-0.1|&\zlmf|0.1| & \dots \\
\zlmf|pre x|   & $\bot$  &\zlmf|0|&\zlmf|0.2|&\zlmf|0.3|&\zlmf|0.3| &\zlmf|0.2| &\zlmf|0.1| & \dots \\
\zlmf|x|       &\zlmf|0|  &\zlmf|0.2|&\zlmf|0.3|&\zlmf|0.3|&\zlmf|0.2| &\zlmf|0.1| &\zlmf|0.2| & \dots
\end{tabular}
\end{center}

\medskip

The node \zl{integr} can be used to define other stream functions, e.g., a PID controller,
which can be called in control structures like
hierarchical automata, e.g., to express a system that switches between automatic
and manual control.
Compared to a general
purpose functional language (or an embedded DSL), the expressiveness of
a synchronous language
is purposely constrained to modularily ensure safety properties that
are critical for the targeted applications: determinacy, deadlock
freedom (reactivity), the generation of statically scheduled code that
runs in bounded time and~space.

However, to date, these languages have had limited
support for modeling uncertainty (e.g., a noisy sensor or channel, a
variable delay), to simulate the interaction of a software controller
and a partially unknown environment, or to infer parameters from
noisy observations. 
Moreover, uncertainty is a first-order design concern for
a controller that operates under the assumption of a probabilistic model
of their environment (e.g., object tracking).
Using a probabilistic
environment model and data gathered from observing the environment,
a controller can infer
a distribution of likely environments given the observations.
Existing approaches consist in hand-coding stochastic controllers that have a known solution (e.g., Kalman filters) which can be tedious and
error-prone, or to simply perform
off-line statistical testing on the generated code of a controller.
Alternatively, in recent years, {\em probabilistic
  programming} has developed as an approach to endow general purpose
languages with the ability to automate inference.

\paragraph{Probabilistic Programming.}

Probabilistic programming languages are used to describe probabilistic models and automatically infer distributions of latent~(i.e., unobserved) parameters from observations.

A popular approach~\cite{goodman_stuhlmuller_2014,TranHSB0B17,BinghamCJOPKSSH19,MurrayS18,tolpin_et_al_2016,wu2016SwiftCI} consists in extending a general-purpose programming language with three constructs:
(1)~\zl{x = sample(d)} introduces a random variable \zl{x} of distribution \zl{d},
(2)~\zl{observe(d, y)} conditions on the fact that a sample from distribution \zl{d} produced the observation \zl{y}, and
(3)~\zl{infer m obs} computes the distribution of the output values of a program or \emph{model}~\zl{m} w.r.t. the observation of the input data~\zl{obs}.

Probabilistic programming languages offer a variety of automatic inference techniques ranging from exact inference~\cite{psi} to approximate inference~\cite{hakaru} and include hybrid approaches that combine exact and approximate techniques when part of the program is analytically tractable~\cite{murray18delayed_sampling}.
However, a standing challenge for these programming languages is that none of them meet the design goals of synchronous reactive languages by being immediately amenable to techniques to ensure that for example a program with an indefinite execution time runs in bounded memory.

\paragraph{Inference in the Loop.}
In this paper we extend \zelus\footnote{Language distribution and manual available at \url{http://zelus.di.ens.fr}.} to provide a synchronous probabilistic programming language, \ProbZelus. \ProbZelus enables one to combine deterministic reactive data-flow programs, such as \zl{integr} (above), with probabilistic programming constructs to produce reactive probabilistic programs.

Compared to other probabilistic languages (e.g. WebPPL, Church, Stan) where inference is executed on terminating pure functions, our probabilistic models are stateful stream processors.
Inference on probabilistic models runs in parallel with deterministic processes that interact with the environment.
The distributions computed by \zl{infer} at each step can thus be used by deterministic processes to compute new inputs for the next inference step.
We term this capability {\em inference-in-the-loop}.

\paragraph{Streaming Inference.}
\ProbZelus provides multiple inference algorithms, most notably the \emph{delayed sampling} inference algorithm~\cite{murray18delayed_sampling}.
This hybrid strategy combines the approximate inference technique of particle filtering~\cite{particlefilter} with exact inference when it is possible to symbolically determine the exact distribution for some or all of the latent variables of the program~\cite{rbpf}.

However, the memory consumption of delayed sampling strictly increases with the number of random variables which is not practical for reactive applications that operate on infinite streams.
We propose a novel \emph{streaming} implementation of delayed sampling that can operate over infinite streams in constant memory for a large class of models.
\ProbZelus therefore provides an expressive language for reactive probabilistic programming with appropriate memory consumption properties.

\paragraph{Contributions.} We present the following contributions:

\begin{description}[wide, labelwidth=!, labelindent=0pt, topsep=0.25em, itemsep=0.25em]
\item[Design, Semantics, Compilation.]
We present \ProbZelus, the first synchronous probabilistic programming language, combining language constructs for streams (reactivity) with those for probabilistic programming thus enabling \emph{inference-in-the-loop}.
We give a measure-based co-iterative semantics for \ProbZelus that forms the basis of a compiler and demonstrate a semantics-preserving compilation strategy to a first-order functional language~\muF.

\item[Streaming Inference.]
We define the semantics of multiple inference algorithms on \muF including particle filtering and delayed sampling.
We then present a novel \emph{streaming delayed sampling} implementation which enables partial exact inference over infinite streams in bounded memory for a large class of models.

\item[Evaluation.]
We evaluate the performance of \ProbZelus on a set of benchmarks that illustrate multiple aspects of the language.
We demonstrate that streaming delayed sampling drastically reduces the number of particles required to achieve better accuracy compared to a particle filter.
\end{description}

The result is \ProbZelus, a synchronous probabilistic language that enables us to write, in the very same source, a deterministic model for the control software and a probabilistic model with complex interactions between the two.
On one hand, a deterministic model of a controller can rely on predictions computed by a probabilistic model.
On the other hand, a probabilistic model can be programmed in an expressive reactive language.
\ProbZelus is open source (\url{https://github.com/IBM/probzelus}).
\ifextended
This paper is a version with appendices of the paper published at PLDI 2020~\cite{rppl-short}.
\else
An extended version with appendices of the paper is also available~\cite{rppl-extended}.
\fi

 \section{Example}\label{sec:example}

In this section, we demonstrate how \ProbZelus  provides probabilistic modeling, inference-in-the-loop, and bounded-memory inference for a robot navigation system.
The results of the inference are continuously used by a controller to correct the robot trajectory.

\subsection{Inference in the Loop.}

\begin{figure}
\begin{lstlisting}[basicstyle=\footnotesize\ttfamily,aboveskip=0em]
let proba kalman (xo, u, acc, gps) = x where
  rec mu = xo -> (a *@ pre x) +@ (b *@ u)
  and x = sample (mv_gaussian (mu, noise))
  and () = observe (gaussian (vec_get (x, 2), 1.0), acc)
  and () = present gps(pos) ->
             observe (gaussian (vec_get (x, 0), 0.01), pos)
           else ()

let node robot (xo, uo, acc, gps) = u where
  rec x_dist = infer 1000 kalman (xo, u, acc, gps)
  and u = uo -> lqr a b (mean (pre x_dist))
\end{lstlisting}
\caption{Robot controller with inference-in-the-loop. \zlm{+@} and \zlm{*@} are matrix operations, \zlm{vec\_get x i} is the $i$th projection.}
\label{fig:lqr_src}
\end{figure}

We consider a robot equipped with an accelerometer and a GPS.
We assume that the motion of the robot can be described as: ${x_{t+1} = Ax_t + Bu_t}$ where~$x_t$ denotes the state of the robot~(position, velocity, and acceleration) at a given time step~$t$, and~$u_t$ denotes the command sent to the robot.
$A$~and~$B$ are constant matrices.
In addition, the robot receives at each step noisy \emph{observations} from the accelerometer~$a_t$, and sporadically an estimation of the position from a GPS~$p_t$.

\Cref{fig:lqr_src} presents a controller, \zl{robot}, that given an initial state~\zl{xo}, an initial command~\zl{uo}, and inputs from the accelerometer~\zl{acc} and the GPS~\zl{gps} computes a stream~\zl{u} of commands that drives the robot to a given target.
The body of \zl{robot} is the parallel composition of (1)~the inference of a probabilistic process \zl{kalman} that estimates \zl{x_dist} the stream of distributions over the robot's state, and (2)~a deterministic process that computes the stream \zl{u} of commands.
It is written as two mutually recursive equations that define~\zl{x_dist} from~\zl{u} and~\zl{u} from the previous value of~\zl{x_dist}.

The command \zl{u} is set to the initial command \zl{uo} at the first time step, and is then computed by a \emph{Linear-Quadratic Regulator} (LQR)~\cite{sontag13control} --- a stable and optimal controller for such dynamic systems --- given the estimation of the state at the previous step. 
Because LQR controllers depend only on mean posterior state, the example in \Cref{fig:lqr_src} uses the \zl{mean} function to compute the mean of \zl{x_dist} before invoking the LQR controller.

\paragraph{Inference.}
The stream~\zl{x_dist} of distributions of state is inferred from the model defined by the probabilistic node \zl{kalman} given the initial state~\zl{xo}, the command~\zl{u}, and the observations~\zl{acc} and~\zl{gps}.
The keyword \zl{proba} indicates a probabilistic model.

In this example, the model is a \emph{Kalman Filter} illustrated in \Cref{fig:hmm}.
A Kalman filter is a time-dependent probabilistic model used to describe inference problems such as tracking, in which a tracker estimates the true position of an object given noisy, sensed observations.
The robot's state~$x_t$ is a \emph{latent} random variable in that the tracker is not able to directly observe it.
Each arrow connecting two random variables denotes a dependence of the variable at the head of the arrow on the variable at the tail.
In this case, the observations at each time step depends on the current state, and the robot's state at a given time step depends only on its states at the previous time step.

\tikzstyle{latent}=[circle, draw=black, minimum width=1cm]
\tikzstyle{observed}=[circle, draw=black, fill=gray!70, minimum width=1cm]
\begin{figure}[t]
  \centering
  \scalebox{0.7}{
    \begin{tikzpicture}[node distance=0.5cm and 0.8cm, every node/.style={font = \normalsize}]
    \node (past) {{\normalsize\dots}};

    \node[right=of past, latent] (pos1) {$x_{t-1}$};
    \node[observed, below=of pos1] (obs1) {$a_{t-1}$};
    \draw[-latex] (pos1) -- (obs1);
    \draw[-latex] ([xshift=0.2cm]past.east) -- (pos1);

    \node[latent, right=of pos1] (pos2) {$x_{t}$};
    \node[observed, below=of pos2] (obs2) {$a_{t}$};
    \draw[-latex] (pos2) -- (obs2);
    \draw[-latex] (pos1) -- (pos2);

    \node[latent, right=of pos2, xshift=1.25em] (pos3) {$x_{t+1}$};
    \node[observed, below=of pos3, xshift=-1.75em] (obs3) {$a_{t+1}$};
    \node[observed, below=of pos3, xshift=1.75em] (obs4) {$p_{t+1}$};
    \draw[-latex] (pos3) -- (obs3);
    \draw[-latex] (pos3) -- (obs4);
    \draw[-latex] (pos2) -- (pos3);

    \node[latent, right=of pos3, xshift=1.25em] (pos4) {$x_{t+2}$};
    \node[observed, below=of pos4] (obs4) {$a_{t+2}$};
    \draw[-latex] (pos4) -- (obs4);
    \draw[-latex] (pos3) -- (pos4);

    \node[right=of pos4] (future) {{\normalsize\dots}};
    \draw[-latex] (pos4) -- ([xshift=-0.2cm]future.west);
  \end{tikzpicture}}
  \caption{Kalman filter for the robot example. Variables are either latent (white, e.g., state~$x$) or observed (gray, e.g., acceleration~$a$). The position~$p$ is only sporadically observed.}
  \label{fig:hmm}
\end{figure}
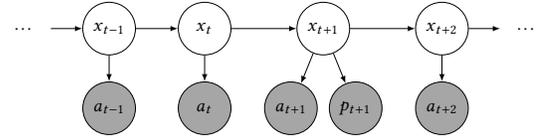

\begin{figure*}
\centering
{\small\sf
\begin{minipage}{0.4\textwidth}
\centering
\begingroup
  \makeatletter
  \providecommand\color[2][]{\GenericError{(gnuplot) \space\space\space\@spaces}{Package color not loaded in conjunction with
      terminal option `colourtext'}{See the gnuplot documentation for explanation.}{Either use 'blacktext' in gnuplot or load the package
      color.sty in LaTeX.}\renewcommand\color[2][]{}}\providecommand\includegraphics[2][]{\GenericError{(gnuplot) \space\space\space\@spaces}{Package graphicx or graphics not loaded}{See the gnuplot documentation for explanation.}{The gnuplot epslatex terminal needs graphicx.sty or graphics.sty.}\renewcommand\includegraphics[2][]{}}\providecommand\rotatebox[2]{#2}\@ifundefined{ifGPcolor}{\newif\ifGPcolor
    \GPcolorfalse
  }{}\@ifundefined{ifGPblacktext}{\newif\ifGPblacktext
    \GPblacktexttrue
  }{}\let\gplgaddtomacro\g@addto@macro
\gdef\gplbacktext{}\gdef\gplfronttext{}\makeatother
  \ifGPblacktext
\def\colorrgb#1{}\def\colorgray#1{}\else
\ifGPcolor
      \def\colorrgb#1{\color[rgb]{#1}}\def\colorgray#1{\color[gray]{#1}}\expandafter\def\csname LTw\endcsname{\color{white}}\expandafter\def\csname LTb\endcsname{\color{black}}\expandafter\def\csname LTa\endcsname{\color{black}}\expandafter\def\csname LT0\endcsname{\color[rgb]{1,0,0}}\expandafter\def\csname LT1\endcsname{\color[rgb]{0,1,0}}\expandafter\def\csname LT2\endcsname{\color[rgb]{0,0,1}}\expandafter\def\csname LT3\endcsname{\color[rgb]{1,0,1}}\expandafter\def\csname LT4\endcsname{\color[rgb]{0,1,1}}\expandafter\def\csname LT5\endcsname{\color[rgb]{1,1,0}}\expandafter\def\csname LT6\endcsname{\color[rgb]{0,0,0}}\expandafter\def\csname LT7\endcsname{\color[rgb]{1,0.3,0}}\expandafter\def\csname LT8\endcsname{\color[rgb]{0.5,0.5,0.5}}\else
\def\colorrgb#1{\color{black}}\def\colorgray#1{\color[gray]{#1}}\expandafter\def\csname LTw\endcsname{\color{white}}\expandafter\def\csname LTb\endcsname{\color{black}}\expandafter\def\csname LTa\endcsname{\color{black}}\expandafter\def\csname LT0\endcsname{\color{black}}\expandafter\def\csname LT1\endcsname{\color{black}}\expandafter\def\csname LT2\endcsname{\color{black}}\expandafter\def\csname LT3\endcsname{\color{black}}\expandafter\def\csname LT4\endcsname{\color{black}}\expandafter\def\csname LT5\endcsname{\color{black}}\expandafter\def\csname LT6\endcsname{\color{black}}\expandafter\def\csname LT7\endcsname{\color{black}}\expandafter\def\csname LT8\endcsname{\color{black}}\fi
  \fi
    \setlength{\unitlength}{0.0500bp}\ifx\gptboxheight\undefined \newlength{\gptboxheight}\newlength{\gptboxwidth}\newsavebox{\gptboxtext}\fi \setlength{\fboxrule}{0.5pt}\setlength{\fboxsep}{1pt}\begin{picture}(3400.00,2380.00)\gplgaddtomacro\gplbacktext{\csname LTb\endcsname \put(-132,594){\makebox(0,0)[r]{\strut{}$10^{1}$}}\put(-132,1042){\makebox(0,0)[r]{\strut{}$10^{2}$}}\put(-132,1491){\makebox(0,0)[r]{\strut{}$10^{3}$}}\put(-132,1939){\makebox(0,0)[r]{\strut{}$10^{4}$}}\put(78,374){\makebox(0,0){\strut{}$1$}}\put(1185,374){\makebox(0,0){\strut{}$10$}}\put(2292,374){\makebox(0,0){\strut{}$100$}}\put(3399,374){\makebox(0,0){\strut{}$1000$}}}\gplgaddtomacro\gplfronttext{\csname LTb\endcsname \put(-605,1266){\rotatebox{-270}{\makebox(0,0){\strut{}LQR loss}}}\put(1699,154){\makebox(0,0){\strut{}Particles}}\put(1699,2049){\makebox(0,0){\strut{}Robot Accuracy}}\csname LTb\endcsname \put(3267,1766){\makebox(0,0)[r]{\strut{}PF}}\csname LTb\endcsname \put(3267,1546){\makebox(0,0)[r]{\strut{}SDS}}}\gplbacktext
    \put(0,0){\includegraphics{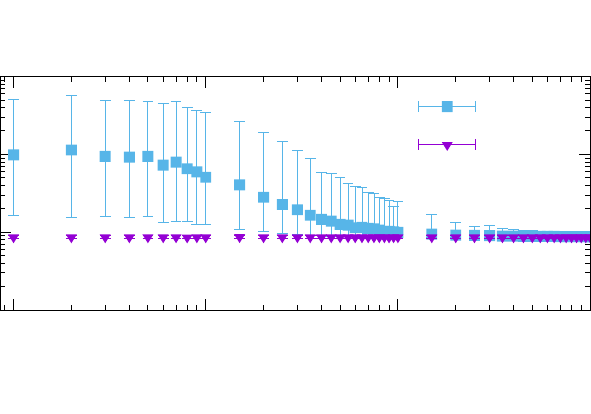}}\gplfronttext
  \end{picture}\endgroup
 \end{minipage}
\hspace{0.1\textwidth}
\begin{minipage}{0.4\textwidth}
\centering
\begingroup
  \makeatletter
  \providecommand\color[2][]{\GenericError{(gnuplot) \space\space\space\@spaces}{Package color not loaded in conjunction with
      terminal option `colourtext'}{See the gnuplot documentation for explanation.}{Either use 'blacktext' in gnuplot or load the package
      color.sty in LaTeX.}\renewcommand\color[2][]{}}\providecommand\includegraphics[2][]{\GenericError{(gnuplot) \space\space\space\@spaces}{Package graphicx or graphics not loaded}{See the gnuplot documentation for explanation.}{The gnuplot epslatex terminal needs graphicx.sty or graphics.sty.}\renewcommand\includegraphics[2][]{}}\providecommand\rotatebox[2]{#2}\@ifundefined{ifGPcolor}{\newif\ifGPcolor
    \GPcolorfalse
  }{}\@ifundefined{ifGPblacktext}{\newif\ifGPblacktext
    \GPblacktexttrue
  }{}\let\gplgaddtomacro\g@addto@macro
\gdef\gplbacktext{}\gdef\gplfronttext{}\makeatother
  \ifGPblacktext
\def\colorrgb#1{}\def\colorgray#1{}\else
\ifGPcolor
      \def\colorrgb#1{\color[rgb]{#1}}\def\colorgray#1{\color[gray]{#1}}\expandafter\def\csname LTw\endcsname{\color{white}}\expandafter\def\csname LTb\endcsname{\color{black}}\expandafter\def\csname LTa\endcsname{\color{black}}\expandafter\def\csname LT0\endcsname{\color[rgb]{1,0,0}}\expandafter\def\csname LT1\endcsname{\color[rgb]{0,1,0}}\expandafter\def\csname LT2\endcsname{\color[rgb]{0,0,1}}\expandafter\def\csname LT3\endcsname{\color[rgb]{1,0,1}}\expandafter\def\csname LT4\endcsname{\color[rgb]{0,1,1}}\expandafter\def\csname LT5\endcsname{\color[rgb]{1,1,0}}\expandafter\def\csname LT6\endcsname{\color[rgb]{0,0,0}}\expandafter\def\csname LT7\endcsname{\color[rgb]{1,0.3,0}}\expandafter\def\csname LT8\endcsname{\color[rgb]{0.5,0.5,0.5}}\else
\def\colorrgb#1{\color{black}}\def\colorgray#1{\color[gray]{#1}}\expandafter\def\csname LTw\endcsname{\color{white}}\expandafter\def\csname LTb\endcsname{\color{black}}\expandafter\def\csname LTa\endcsname{\color{black}}\expandafter\def\csname LT0\endcsname{\color{black}}\expandafter\def\csname LT1\endcsname{\color{black}}\expandafter\def\csname LT2\endcsname{\color{black}}\expandafter\def\csname LT3\endcsname{\color{black}}\expandafter\def\csname LT4\endcsname{\color{black}}\expandafter\def\csname LT5\endcsname{\color{black}}\expandafter\def\csname LT6\endcsname{\color{black}}\expandafter\def\csname LT7\endcsname{\color{black}}\expandafter\def\csname LT8\endcsname{\color{black}}\fi
  \fi
    \setlength{\unitlength}{0.0500bp}\ifx\gptboxheight\undefined \newlength{\gptboxheight}\newlength{\gptboxwidth}\newsavebox{\gptboxtext}\fi \setlength{\fboxrule}{0.5pt}\setlength{\fboxsep}{1pt}\begin{picture}(3400.00,2380.00)\gplgaddtomacro\gplbacktext{\csname LTb\endcsname \put(-132,594){\makebox(0,0)[r]{\strut{}$10^{1}$}}\put(-132,930){\makebox(0,0)[r]{\strut{}$10^{2}$}}\put(-132,1267){\makebox(0,0)[r]{\strut{}$10^{3}$}}\put(-132,1603){\makebox(0,0)[r]{\strut{}$10^{4}$}}\put(-132,1939){\makebox(0,0)[r]{\strut{}$10^{5}$}}\put(78,374){\makebox(0,0){\strut{}$1$}}\put(1185,374){\makebox(0,0){\strut{}$10$}}\put(2292,374){\makebox(0,0){\strut{}$100$}}\put(3399,374){\makebox(0,0){\strut{}$1000$}}}\gplgaddtomacro\gplfronttext{\csname LTb\endcsname \put(-605,1266){\rotatebox{-270}{\makebox(0,0){\strut{}Speed (ms)}}}\put(1699,154){\makebox(0,0){\strut{}Particles}}\put(1699,2049){\makebox(0,0){\strut{}Robot Latency}}\csname LTb\endcsname \put(1119,1766){\makebox(0,0)[r]{\strut{}PF}}\csname LTb\endcsname \put(1119,1546){\makebox(0,0)[r]{\strut{}SDS}}}\gplbacktext
    \put(0,0){\includegraphics{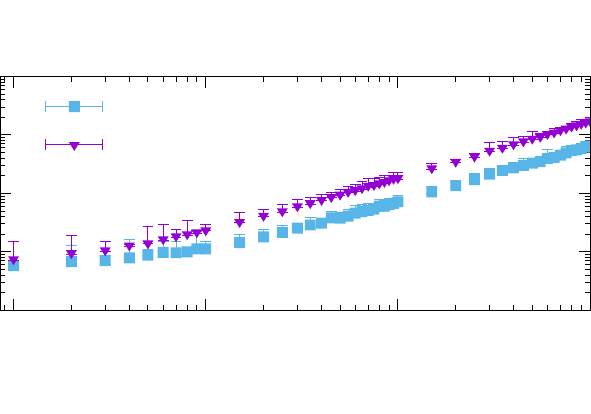}}\gplfronttext
  \end{picture}\endgroup
 \end{minipage}}
\vspace{-1em}
\caption{Particle filter (PF) and streaming delayed sampling (\gcdsacro) performances for the robot example of \Cref{fig:lqr_src}. Accuracy is measured using the loss function of the LQR. Speed corresponds to the execution time of 500 steps.}
\label{fig:kalman-perf}
\end{figure*}

\paragraph{Sampling.} Inside the \zl{kalman} node, the \zl{sample} operator samples a value from a probability distribution.
In this case, the expression samples the current state~\zl{x} from a multivariate Gaussian with mean obtained by applying the motion model to the previous state and the command.
This code models the trajectory of a robot where at each time step, the state is Gaussian-distributed around an estimation computed from the motion model.

\paragraph{Observations.}
The expression \zl{observe} conditions the execution on observed data.
Its first parameter denotes a distribution that models the observation and its second parameter denotes the observed value itself.
In this case, at each step, the first \zl{observe} statement models a Gaussian-distributed observation of the current acceleration \zl{vec_get x 2} given by~\zl{acc}.
The input \zl{gps} is a \emph{signal} that is only emitted when the GPS computes a new position.
When a value \zl{pos} is emitted on \zl{gps}, the \zl{present} construct executes its left branch, further conditioning the model by adding a Gaussian-distributed observation of the current position \zl{vec_get x 0} given by~\zl{pos}.

\subsection{Streaming Inference}
\label{sec:example-inference}

A classic operational interpretation of a probabilistic model is an \emph{importance sampler} that generates random samples from the model together with an importance weight measuring the quality of the sample.
In this model, each execution of a \zl{sample} operator samples a value from the operator's corresponding distribution.
Each execution of an \zl{observe} evaluates the likelihood of the provided observation and multiplies the current importance weight by this value.
Then, each execution step of \zl{infer} yields a distribution represented as a set of pairs (output, weight) or \emph{particles}. The particles can be re-sampled at each step to build a \emph{particle filter}~\cite{doucet-smc-2006}.

The integer parameter to \zl{infer} determines how many particles to use: the more particles the user specifies, the more accurate the estimate of the distribution becomes.
The PF points in \Cref{fig:kalman-perf} present this improvement in accuracy as a function of increasing the number of particles for the robot example.
However, as the latency results presents, the more particles the user specifies, the more computation is required for each step because each particle requires a full, independent execution of each time step of the model.

\paragraph{Streaming Delayed Sampling.}
Delayed sampling~\cite{murray18delayed_sampling} can reduce the number of particles required to achieve a given desired quality of inference.
Specifically, delayed sampling exploits the opportunity to symbolically
reason about the relationships between random variables to compute closed-form distributions whenever possible.
To capture relationships between random variables, delayed sampling maintains a graph: a \emph{Bayesian network} that can be used to compute closed-form distributions involving subsets of random variables.
For instance, this inference scheme is able to compute the exact posterior distribution for our robot example.
The \gcdsacro dots in \Cref{fig:kalman-perf} show that the accuracy is independent of the number of particles since each particle computes the exact solution.

\newcommandx{\dep}[4][1={-stealth}, 4={}]{\draw[#1] (#2) to node[pos=0.4] {#4} (#3)}

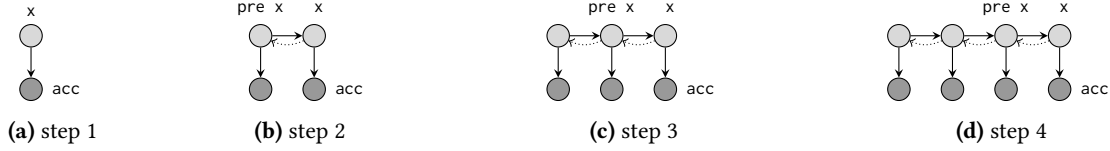
\begin{figure*}
\centering
\begin{small}
\begin{tabular}{cccc}
\begin{subfigure}[b]{.11\textwidth}
\centering
\begin{tikzpicture}[node distance=0.4cm and 0.4cm,every node/.style={font=\scriptsize, align=center}]
    \node[marginalized, label={\texttt{x}}] (x1) {};
    \node[label={[label distance=0.1ex]0:\texttt{acc}}, realized, below=of x1] (o1) {};
    \draw[-stealth] (x1) -- (o1);
\end{tikzpicture}
\caption{step $1$}
\end{subfigure}
& \begin{subfigure}[b]{.22\textwidth}
\centering
\begin{tikzpicture}[node distance=0.4cm and 0.4cm,every node/.style={font=\scriptsize, align=center}]
    \node[marginalized, label={\texttt{pre x}}] (x1) {};
    \node[realized, below=of x1] (o1) {};
    \node[marginalized, label={\texttt{\phantom{p}x}}, right=of x1] (x2) {};
    \draw[-stealth] (x1) -- (x2);
    \link[->, densely dotted]{x2}{x1};
    \node[realized, label={[label distance=0.1ex]0:\texttt{acc}}, below=of x2] (o2) {};
    \draw[-stealth] (x1) -- (o1);
    \draw[-stealth] (x2) -- (o2);
\end{tikzpicture}
\caption{step 2}
\label{fig:ds_kalman_addx}
\end{subfigure}
& \begin{subfigure}[b]{.24\textwidth}
\centering
\begin{tikzpicture}[node distance=0.4cm and 0.4cm,every node/.style={font=\scriptsize, align=center}]
    \node[marginalized] (x1) {};
    \node[realized, below=of x1] (o1) {};
    \node[marginalized, label={\texttt{pre x}}, right=of x1] (x2) {};
    \draw[-stealth] (x1) -- (x2);
    \link[->, densely dotted]{x2}{x1};
    \node[realized, below=of x2] (o2) {};
    \node[marginalized, label={\texttt{\phantom{p}x}}, right=of x2] (x3) {};
    \draw[-stealth] (x2) -- (x3);
    \link[->, densely dotted]{x3}{x2};
    \node[realized, label={[label distance=0.1ex]0:\texttt{acc}}, below=of x3] (o3) {};
    \draw[-stealth] (x1) -- (o1);
    \draw[-stealth] (x2) -- (o2);
    \draw[-stealth] (x3) -- (o3);
\end{tikzpicture}
\caption{step 3}
\label{fig:ds_kalman_addo}
\end{subfigure}
& \begin{subfigure}[b]{.27\textwidth}
\centering
\begin{tikzpicture}[node distance=0.4cm and 0.4cm,every node/.style={font=\scriptsize, align=center}]
    \node[marginalized] (x1) {};
    \node[realized, below=of x1] (o1) {};
    \node[marginalized, right=of x1] (x2) {};
    \draw[-stealth] (x1) -- (x2);
    \link[->, densely dotted]{x2}{x1};
    \node[realized, below=of x2] (o2) {};
    \node[marginalized, label={\texttt{pre x}}, right=of x2] (x3) {};
    \draw[-stealth] (x2) -- (x3);
    \link[->, densely dotted]{x3}{x2};
    \node[realized, below=of x3] (o3) {};
    \node[marginalized, label={\texttt{\phantom{p}x}}, right=of x3] (x4) {};
    \draw[-stealth] (x3) -- (x4);
    \link[->, densely dotted]{x4}{x3};
    \node[realized, label={[label distance=0.1ex]0:\texttt{acc}}, below=of x4] (o4) {};
    \draw[-stealth] (x1) -- (o1);
    \draw[-stealth] (x2) -- (o2);
    \draw[-stealth] (x3) -- (o3);
    \draw[-stealth] (x4) -- (o4);
\end{tikzpicture}
\caption{step 4}
\label{fig:ds_kalman_margx}
\end{subfigure}
\end{tabular}
\end{small}
\caption{Evolution of the delayed sampling graph for the model of \Cref{fig:lqr_src}. Each node denotes either a value (dark gray) or a distribution (light gray). Plain arrows represent dependencies in the underlying Bayesian network. The dotted arrow represents the pointers in the original data-structure implementing the graph. Labels indicate the program variables.}
\label{fig:ds_kalman_naive}
\end{figure*}

\Cref{fig:ds_kalman_naive} illustrates the evolution of the delayed sampling graph as it proceeds through the first four time steps of the robot example (for simplicity we assume that there is no GPS activation in these four steps).
A notable challenge with the traditional delayed sampling algorithm is that the graph grows linearly in the number of samples.
This property is not tractable in our reactive context because we would like to deploy our programs under the model that they run indefinitely, thus requiring that they execute with bounded resources.
To address this problem, we propose a novel \emph{streaming delayed sampling}~(SDS) implementation of the delayed sampling algorithm.
Specifically, in \Cref{fig:ds_kalman_naive} the node denoting the marginal posterior for \zl{x} at step~1 can be eliminated from the graph at step~3 because the distributions for \zl{pre x} and \zl{x} have fully incorporated its effect on
their values and, moreover, the program no longer maintains a reference to the node.

\begin{figure}[t]
\centering
{\small\sf
\begingroup
  \makeatletter
  \providecommand\color[2][]{\GenericError{(gnuplot) \space\space\space\@spaces}{Package color not loaded in conjunction with
      terminal option `colourtext'}{See the gnuplot documentation for explanation.}{Either use 'blacktext' in gnuplot or load the package
      color.sty in LaTeX.}\renewcommand\color[2][]{}}\providecommand\includegraphics[2][]{\GenericError{(gnuplot) \space\space\space\@spaces}{Package graphicx or graphics not loaded}{See the gnuplot documentation for explanation.}{The gnuplot epslatex terminal needs graphicx.sty or graphics.sty.}\renewcommand\includegraphics[2][]{}}\providecommand\rotatebox[2]{#2}\@ifundefined{ifGPcolor}{\newif\ifGPcolor
    \GPcolorfalse
  }{}\@ifundefined{ifGPblacktext}{\newif\ifGPblacktext
    \GPblacktexttrue
  }{}\let\gplgaddtomacro\g@addto@macro
\gdef\gplbacktext{}\gdef\gplfronttext{}\makeatother
  \ifGPblacktext
\def\colorrgb#1{}\def\colorgray#1{}\else
\ifGPcolor
      \def\colorrgb#1{\color[rgb]{#1}}\def\colorgray#1{\color[gray]{#1}}\expandafter\def\csname LTw\endcsname{\color{white}}\expandafter\def\csname LTb\endcsname{\color{black}}\expandafter\def\csname LTa\endcsname{\color{black}}\expandafter\def\csname LT0\endcsname{\color[rgb]{1,0,0}}\expandafter\def\csname LT1\endcsname{\color[rgb]{0,1,0}}\expandafter\def\csname LT2\endcsname{\color[rgb]{0,0,1}}\expandafter\def\csname LT3\endcsname{\color[rgb]{1,0,1}}\expandafter\def\csname LT4\endcsname{\color[rgb]{0,1,1}}\expandafter\def\csname LT5\endcsname{\color[rgb]{1,1,0}}\expandafter\def\csname LT6\endcsname{\color[rgb]{0,0,0}}\expandafter\def\csname LT7\endcsname{\color[rgb]{1,0.3,0}}\expandafter\def\csname LT8\endcsname{\color[rgb]{0.5,0.5,0.5}}\else
\def\colorrgb#1{\color{black}}\def\colorgray#1{\color[gray]{#1}}\expandafter\def\csname LTw\endcsname{\color{white}}\expandafter\def\csname LTb\endcsname{\color{black}}\expandafter\def\csname LTa\endcsname{\color{black}}\expandafter\def\csname LT0\endcsname{\color{black}}\expandafter\def\csname LT1\endcsname{\color{black}}\expandafter\def\csname LT2\endcsname{\color{black}}\expandafter\def\csname LT3\endcsname{\color{black}}\expandafter\def\csname LT4\endcsname{\color{black}}\expandafter\def\csname LT5\endcsname{\color{black}}\expandafter\def\csname LT6\endcsname{\color{black}}\expandafter\def\csname LT7\endcsname{\color{black}}\expandafter\def\csname LT8\endcsname{\color{black}}\fi
  \fi
    \setlength{\unitlength}{0.0500bp}\ifx\gptboxheight\undefined \newlength{\gptboxheight}\newlength{\gptboxwidth}\newsavebox{\gptboxtext}\fi \setlength{\fboxrule}{0.5pt}\setlength{\fboxsep}{1pt}\begin{picture}(3400.00,2380.00)\gplgaddtomacro\gplbacktext{\csname LTb\endcsname \put(-132,594){\makebox(0,0)[r]{\strut{}$10^{0}$}}\put(-132,1267){\makebox(0,0)[r]{\strut{}$10^{3}$}}\put(-132,1939){\makebox(0,0)[r]{\strut{}$10^{3}$}}\put(0,374){\makebox(0,0){\strut{}$0$}}\put(340,374){\makebox(0,0){\strut{}$50$}}\put(680,374){\makebox(0,0){\strut{}$100$}}\put(1020,374){\makebox(0,0){\strut{}$150$}}\put(1360,374){\makebox(0,0){\strut{}$200$}}\put(1700,374){\makebox(0,0){\strut{}$250$}}\put(2039,374){\makebox(0,0){\strut{}$300$}}\put(2379,374){\makebox(0,0){\strut{}$350$}}\put(2719,374){\makebox(0,0){\strut{}$400$}}\put(3059,374){\makebox(0,0){\strut{}$450$}}\put(3399,374){\makebox(0,0){\strut{}$500$}}}\gplgaddtomacro\gplfronttext{\csname LTb\endcsname \put(-605,1266){\rotatebox{-270}{\makebox(0,0){\strut{}thousands of words}}}\put(1699,154){\makebox(0,0){\strut{}Step}}\put(1699,2049){\makebox(0,0){\strut{}Robot Memory}}\csname LTb\endcsname \put(921,1766){\makebox(0,0)[r]{\strut{}SDS}}\csname LTb\endcsname \put(921,1546){\makebox(0,0)[r]{\strut{}DS}}}\gplbacktext
    \put(0,0){\includegraphics{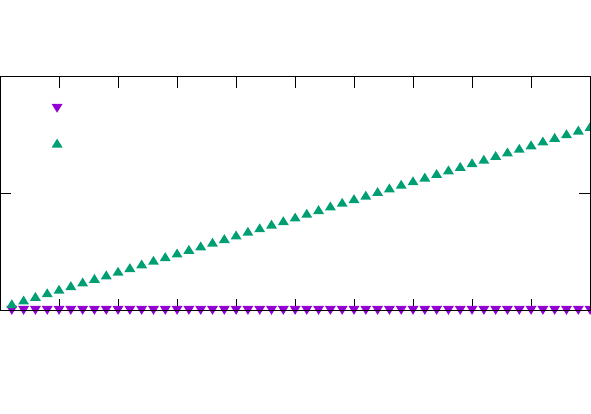}}\gplfronttext
  \end{picture}\endgroup
 }
\caption{Delayed sampling (DS) and streaming delayed sampling (SDS) memory consumption in thousands of live words in the heap per steps for the robot example.}
\label{fig:kalman-memory}
\end{figure}

While the standard delayed sampling algorithm will keep this node alive through the edge pointers
it maintains, SDS builds a {\em pointer-minimal} graph representation with a minimal number of edges that 1) ensure that the graph has sufficient connectivity to support operations in the traditional delayed sampling algorithm and 2) only maintain the reachability of nodes that can effect the distribution of future nodes in the graph.
The result is  that the memory consumption of SDS is constant across the number of steps while the memory consumption of the original delayed sampling implementation~DS increases linearly in the number of steps~(\Cref{fig:kalman-memory}).

 \newcommand{\muZ}{\ProbZelus\xspace}
\newcommand{\compile}{\mathcal{C}}
\newcommand{\defs}{\mathcal{D}}
\newcommand{\inits}{\mathcal{I}}
\newcommand{\alloc}{\mathcal{A}}
\newcommand{\inferds}{\mathcal{I}_\mathrm{ds}}
\newcommand{\sem}[1]{\llbracket #1 \rrbracket}
\newcommand{\psem}[1]{\{\mkern-3.8mu[ #1 ]\mkern-3.8mu\}}
\newcommand{\indicator}{\mathbb{1}}
\newcommand{\letin}[1]{\ensuremath{\mathit{let} \; #1 \; \textit{in} \;}}
\newcommand{\fun}[1]{\ensuremath{\lambda #1. \;}}

\section{Language: Syntax, Typing, Semantics}
\label{sec:language}

\ProbZelus is a reactive probabilistic language with inference-in-the-loop which enables interaction between probabilistic models and deterministic processes.
This capability introduces two design requirements.
First, a probabilistic model must be able to receive inputs from an evolving environment.
Second, instead of awaiting the final result of the inference,
deterministic processes running in parallel need access to
intermediate results. The resulting inference-in-the-loop enables
feedback loops between inferred distributions from probabilistic models and deterministic processes,
which our design controls by enforcing a separation between the semantics of probabilistic and deterministic
execution.

In this section, we formalize the syntax of \ProbZelus, introduce a type system that imposes a clear separation between deterministic and probabilistic expressions, and define the semantics of the language in a co-iteration framework where the semantics of probabilistic processes is adapted from Staton's measure-theoretic semantics for probabilistic programs~\cite{staton17}.
The co-iterative semantics forms the basis of a compiler that is described in \Cref{sec:compilation}.

\subsection{Syntax}

We focus on the following kernel of \ProbZelus. The missing constructs (e.g., \zl{pre} and \zl{->}) can be compiled into this kernel via a source-to-source transformation.

\smallskip
  
\begin{lstlisting}
$d ::=$ let node $f$ $x$ = $e$  $|$  let proba $f$ $x$ = $e$  $|$  $d$ $d$
$e ::=$ $c$  $|$  $x$  $|$  ($e$,$e$)  $|$  $\mathit{op}$($e$)  $|$  $f$($e$)  $|$ last $x$  
   $|$ $e$ where rec $E$
   $|$ present $e$ -> $e$ else $e$  $|$  reset $e$ every $e$
   $|$ sample($e$)  $|$  observe($e$, $e$)  $|$  infer($e$)
$E ::=$ $x$ = $e$  $|$  init $x$ = $c$  $|$  $E$ and $E$
\end{lstlisting}

\medskip
    
A program is a sequence of \emph{declarations}~$d$ of stream functions (\zl{node}) and probabilistic stream functions (\zl{proba}).
An \emph{expression}~$e$ is either a constant~($c$), a variable~($x$), a pair, an external operator application~($\textit{op}$), a function application~(\zl{$f$($e$)}), a delay (\zl{last $x$}) that returns a value ($x$) from the previous step, or a set of locally recursive equations~(\zl{$e$ where rec $E$}).
A set of \emph{equations}~$E$ is either an equation \zl{$x$ = $e$} that define~\zl{x} with the stream~$e$, the initialization of a variable with a constant \zl{init $x$ = c}, or parallel composition of sets of equations.

Operators~($\textit{op}$) include boolean and arithmetic operators.
In addition, \ProbZelus offers a library of dedicated operators to analyze distributions, such as \zl{mean} and \zl{variance}, that can be used in any context (probabilistic or deterministic), e.g., on the result of the inference.

The control structure \zl{present $e$ -> $e_1$ else $e_2$} is an \emph{activation condition} that executes the expression~$e_1$ only when the value of~$e$ is \zl{true} and executes~$e_2$ otherwise.
It differs from \zl{if $e$ then $e_1$ else $e_2$}, where both $e_1$ and $e_2$ are computed at each step~(making their internal states evolve) and the returned value is chosen based on the value of~$e$.\footnote{The \zlm+if+ construct can thus be considered as an external operator.}
In the following example, \zl{o1} and \zl{o2} are different streams:
\smallskip
\begin{lstlisting}[belowskip=0em]
let node cpt () = o where rec o = 0 -> pre o + 1
\end{lstlisting}
\begin{lstlisting}[belowskip=0em]
let node present_vs_if (b) = (o1, o2) where
  rec o1 = present (b) -> cpt () else 0
  and o2 = if b then cpt  () else 0
\end{lstlisting}
\begin{center}
  \begin{tabular}{l|lllllllll}
    \zlm|b|  & \zlm|true| & \zlm|true| & \zlm|false| & \zlm|true| & \zlm|false| & \zlm|false| & \zlm|true| & ... \\
    \zlm|o1| & \zlm|0| & \zlm|1| & \zlm|0| & \zlm|2| & \zlm|0| & \zlm|0| & \zlm|3| & ... \\
    \zlm|o2| & \zlm|0| & \zlm|1| & \zlm|0| & \zlm|3| & \zlm|0| & \zlm|0| & \zlm|6| & ...
  \end{tabular}
\end{center}

\medskip

The \zl{reset $e_1$ every $e$} construct re-initializes the values of the \zl{init} equations and the corresponding \zl{last} expressions in $e_1$ each time $e$ is true.

The language is extended with the classic probabilistic expressions: \zl{sample} to draw from a distribution, \zl{observe} to condition on observations, and \zl{infer} to compute the distribution described by a model.

\paragraph{Scheduling.}
In the expression \zl{$e$ where rec $E$}, $E$ is a set of mutually recursive equations.
In practice, a scheduler reorders the equations according to their dependencies.
Initializations \zl{init $x_j$ = $c_j$} are grouped at the beginning, and an equation $x_j = e_j$ must be scheduled after the equation $x_i = e_i$ if the expression $e_j$ uses $x_i$ outside a \zl{last} construct.
A program satisfying this partial order is said to be \emph{scheduled}.
The compiler can also introduce additional equations to relax the scheduling constraints and rejects programs that cannot be statically scheduled~\cite{clock-lctes08}.
After scheduling, the expression \zl{$e$ where rec $E$} has the following form.
\begin{lstlisting}
$e$ where rec init $x_1$ = $c_1$ ... and init $x_k$ = $c_k$
        and $y_1$ = $e_1$ ... and $y_n$ = $e_n$
\end{lstlisting}

\noindent
For simplicity, we also assume that every initialized variable is defined in a subsequent equation, i.e., $\{x_i\}_{1..k} \cap \{y_j\}_{1..n} = \{x_i\}_{1..k}$.
If it is not the case, in this kernel we can always add additional equations of the form \zl!$x_i$ = last $x_i$!.

\paragraph{Kernel.}
All \ProbZelus programs can be encoded in this kernel language.
For instance, the program \zl{integr} of \Cref{sec:introduction}
can be rewritten in the kernel as follows:
\begin{lstlisting}
let node integr (xo, x') = x where
  rec init first = true and init x = 0.
  and first = false
  and x = if last first then xo else last x + (x' * h)
\end{lstlisting}
A stream \zl{first} is defined such that \zl{last first} is only true at the first step. 
The \zl{->} operator is then compiled into an \zl{if} statement. 
The \zl{pre} operator is compiled into a \zl{last} operator. 
The initialization value is arbitrary, the compiler's initialization analysis guaranties that this value is never used~\cite{pouzet-initialization}.
Similarly, other constructs like hierarchical automata can be re-written using \zl{present} and \zl{reset}~\cite{automaton-emsoft06}.

\subsection{Typing: Deterministic vs. Probabilistic}

\newcommand{\typeD}{\ensuremath{\texttt{D}}\xspace}
\newcommand{\typeP}{\ensuremath{\texttt{P}}\xspace}
\newcommand{\type}[4]{\ensuremath{{#1} \vdash^{#4} \zlm!#2! : {#3}}}

\newcommand{\typebool}{\zlm!bool!}
\newcommand{\typefloat}{\zlm!float!}
\newcommand{\typeunit}{\zlm!unit!}
\newcommand{\typetimes}[2]{{#1} \times {#2}}
\newcommand{\typefun}[3]{{#1} \rightarrow^{#2} {#3}}
\newcommand{\typedist}[1]{{#1}\ \zlm!dist!^*}
\newcommand{\typesampler}[1]{{#1}\ \zlm!dist!}

\newcommand{\kindof}[1]{\ensuremath{\mathit{kindOf}(#1)}}
\newcommand{\typeof}[1]{\ensuremath{\mathit{typeOf}(#1)}}

\newcommand{\infrule}[2]{
  {\begin{array}[b]{cc}
      {#1} \vspace{-.2cm}\\
      \hrulefill\\
      {#2}
   \end{array}}}
\newcommand{\axiom}[1]{
    {\begin{array}[b]{cc}
        \hrulefill\\
        {#1}
     \end{array}}}

Deterministic and probabilistic expressions have distinct interpretations.
A dedicated type system discriminates between the two kinds of expressions,
assigning one of two \emph{kinds} to each expression:
\typeD for deterministic, or \typeP for probabilistic.
The typing judgment $\type{G}{$e$}{T}{k}$ states that in the environment~$G$ which maps variable names to their type, the expression~$e$ has kind~$k$ and type~$T$.
Function types $\typefun{T}{k}{T'}$ are extended with the kind~$k$ of the body and we introduce a new datatype \zl{$T$ dist} for the probability distribution over values of type~$T$.

The expressions \zl{sample}, and \zl{observe} are probabilistic but their arguments must be deterministic.
Any deterministic expression can be lifted to a probabilistic expression using a sub-typing rule.
The transition from probabilistic to deterministic is realized via \zl{infer}
(the complete set of typing rules is presented in
\ifextended
\Cref{fig:typing} of \Cref{app:typing}\else
Figure~12 of Appendix A.1\fi).

\medskip

\begin{small}
$$
\begin{array}{c}
\infrule
{\type{G}{$e$}{\typesampler{T}}{\typeD}}
{\type{G}{sample($e$)}{T}{\typeP}}
\qquad
\infrule
{\type{G}{$e_1$}{\typedist{T}}{\typeD} \quad \type{G}{$e_2$}{T}{\typeD}}
{\type{G}{observe($e_1$,\ $e_2$)}{\typeunit}{\typeP}}
\\[0.7em]
\infrule
{\type{G}{$e$}{T}{\typeD}}
{\type{G}{$e$}{T}{\typeP}}
\quad
\infrule
{\type{G}{$e$}{T}{\typeP}}
{\type{G}{infer($e$)}{\typesampler{T}}{\typeD}}
\quad
\infrule
{\type{G}{$e$}{\typedist{T}}{\typeD}}
{\type{G}{$e$}{\typesampler{T}}{\typeD}}
\end{array}
$$
\end{small}

\medskip

The type $\typesampler{T}$ represents distributions over values of type~$T$.
Distributions can be sampled (\zl{sample} statement) and analyzed with external operators such as \zl{mean} and \zl{variance}.
The type $\typedist{T}$ is a subtype of $\typesampler{T}$ that represents distributions known to have a density, i.e., discrete distributions~(w.r.t. the counting measure) and a subset of continuous distributions~(w.r.t. the Lebesgue measure).
In \ProbZelus, to simplify the semantics of the language, the observe statement requires a value of type $\typedist{T}$.
In \ifextended\Cref{app:probzelus}\else Appendix~A\fi, we extend the language with the \zl{factor} statement for arbitrary conditioning.

\subsection{Co-Iterative Semantics}
\label{sec:co-iteration}

We now give the semantics of \ProbZelus in a co-iteration framework~\cite{pouzet-cmcs98}.
In this framework, a \emph{deterministic stream} of type~$T$
is defined by an initial state of type~$S$ and a transition function of type $S \rightarrow T \times S$.

\begin{small}
$$
\mathit{CoStream}(T, S) = S \times (S \rightarrow T \times S)
$$
\end{small}
Repeatedly executing the transition function from the initial state yields a stream of values of type~$T$.

The semantics of a \emph{deterministic expression} ($\kindof{e} = {\typeD}$) is defined using two auxiliary functions.
If $\gamma$ is an environment mapping variable names to values, $\sem{e}^i_\gamma$ denotes the initial state, and $\sem{e}^s_\gamma$ denotes the transition function:

\begin{small}
$$
\sem{e}_\gamma: \mathit{CoStream}(T, S) = \sem{e}^i_\gamma, \sem{e}^s_\gamma
$$
\end{small}

The deterministic semantics of \ProbZelus presented in \Cref{fig:sem_deterministic}  is an extension of~\cite{pouzet-cmcs98} with the control structures \zl{present} and~\zl{reset} (see also
\ifextended
\Cref{fig:sem_deterministic_full} of \Cref{app:semantics}\else
Figure~13 of Appendix~A.2\fi).

\begin{figure}
\begin{small}
$$
\begin{array}{@{}l@{}}
  \begin{array}{@{}lcl}
  \sem{x}^i_\gamma &=& ()\\
  \sem{x}^s_\gamma &=& \fun{s} (\gamma(x),s)
  \end{array}\\[1.5em]
\begin{array}{@{}lcl}
  \sem{\zlm!present\ $e$ ->\ $e_1$ else\ $e_2$!}^i_\gamma &=& 
    (\sem{e}^i_\gamma, \sem{e_1}^i_\gamma, \sem{e_2}^i_\gamma)\\
  \sem{\zlm!present\ $e$ ->\ $e_1$ else\ $e_2$!}^s_\gamma &=& \\
  \multicolumn{3}{l}{\quad
    \begin{array}[t]{@{}l@{}}
    \fun{(s, s_1, s_2)} \letin{v,s' = \sem{e}^s_\gamma(s)}\\
      \quad \mathit{if} \; v \; \begin{array}[t]{@{}l@{}} 
      \mathit{then} \; 
        \letin{v_1,s_1' = \sem{e_1}^s_\gamma(s_1)}
        (v_1, (s', s_1', s_2))\\
      \mathit{else} \;\;
        \letin{v_2,s_2' = \sem{e_2}^s_\gamma(s_2)}
        (v_2, (s', s_1, s_2'))
      \end{array}
    \end{array}
  }
  \end{array}\\[3.5em]
\begin{array}{@{}l@{}}
 \left \llbracket
\begin{minipage}[t][1em]{22.1em}
\begin{lstlisting}[aboveskip=-1.5em, xleftmargin=0pt, xrightmargin=0pt]
$e$ where rec init $x_1$ = $c_1$ ... and init $x_k$ = $c_k$
        and $y_1$ = $e_1$ ... and $y_n$ = $e_n$
\end{lstlisting}
\end{minipage}
\right \rrbracket^i_\gamma = \\[1.1em]
\qquad (
  (c_1, \dots, c_k),\;
  (\sem{e_1}^i_\gamma, \dots, \sem{e_n}^i_\gamma),\;
  \sem{e}^i_\gamma)\\[1.3em]
\left \llbracket
\begin{minipage}[t][1em]{22.1em}
\begin{lstlisting}[aboveskip=-1.5em, xleftmargin=0pt, xrightmargin=0pt]
$e$ where rec init $x_1$ = $c_1$ ... and init $x_k$ = $c_k$
        and $y_1$ = $e_1$ ... and $y_n$ = $e_n$
\end{lstlisting}
\end{minipage}
\right \rrbracket^s_\gamma
= \\[1.1em]
\quad\begin{array}{@{}l@{}}
\fun{((m_1, \dots, m_k), (s_1, \dots, s_n), s)}\\
\quad \begin{array}[t]{@{}l@{}}
\letin{\gamma_1 = \gamma[m_1/\zlm!$x_1$_last!]} \dots
\letin{\gamma_k = \gamma_{k-1}[m_k/\zlm!$x_k$_last!]}\\
\letin{v_1, s_1' = \sem{e_1}^s_{\gamma_k}(s_1)}
\letin{\gamma'_1 = \gamma_k[v_1/y_1]} \dots\\
\letin{v_n, s_n' = \sem{e_n}^s_{\gamma'_{n-1}}(s_n)}
\letin{\gamma'_n = \gamma'_{n-1}[v_n/y_n]}\\
\letin{v, s' = \sem{e}^s_{\gamma'_n}(s)}\\
v, ((\gamma'_n[x_1], \dots, \gamma'_n[x_k]), (s_1', \dots, s_n'), s')
\end{array}
\end{array}
\end{array}
\end{array}
$$
\end{small}
\caption{Semantics of deterministic expressions.}
\label{fig:sem_deterministic}
\end{figure}

The transition function of a variable always returns the corresponding value stored in the environment~$\gamma$.

The \zl{present $e$ -> $e_1$ else $e_2$} construct introduced in \Cref{sec:example} returns the value of~$e_1$ when~$e$ is true and the value of~$e_2$ otherwise.
The state $(s, s_1, s_2)$ stores the state of the three sub-expressions.
The transition function lazily executes the expression $e_1$ or $e_2$ depending on the value of~$e$ and returns the updated state.

The state of a set of scheduled locally recursive definitions \zl{$e$ where rec $E$} comprises three parts: the value of the local variables at the previous step which can be accessed via the \zl{last} operator, the state of the defining expressions, and the state of expression~$e$.
The initialization stores the initial values introduced by \zl{init} and the initial states of all sub-expressions.
The transition function incrementally builds the local environment defined by~$E$. 
First the environment is populated with a set of fresh variables \zl{$x_i$_last} initialized with the values stored in the state that can then be accessed via the \zl{last} operator.  
Then the environment is extended with the definition of all the variables~$y_i$ by executing all the defining expressions (where $\{x_i\}_{1..k} \cap \{y_j\}_{1..n} = \{x_i\}_{1..k}$).
Finally, the expression~$e$ is executed in the final environment.
The updated state contains the value of the initialized variables defined in~$E$ that will the be used to start the next step, and the updated state of the sub-expressions.

\paragraph{Probabilistic extension.}
The semantics of a \emph{probabilistic expression} ($\kindof{e} = {\typeP}$) follows the same scheme, but the transition function returns a \emph{measure} over the set of possible pairs (result, state):
\begin{center}
\begin{small}
$
\mathit{CoPStream}(T, S) = S \times (S \rightarrow (\Sigma_{T \times S} \rightarrow [0, \infty]))
$
\end{small}
\end{center}
A measure~$\mu$ associates a positive number to each measurable set~$U \in \Sigma_{T \times S}$ where $\Sigma_{T \times S}$ denotes the $\Sigma$-algebra of~$T \times S$, i.e., the set of measurable sets over pairs (result, state).
We use the following notation for the semantics of a probabilistic expression~$e$:
\begin{center}
\begin{small}
$
\psem{e}_\gamma : \mathit{CoPStream}(T, S) = \psem{e}^i_\gamma, \psem{e}^s_\gamma
$
\end{small}
\end{center}

\begin{figure}
  \begin{small}
    $$
    \begin{array}{@{}l}
    \begin{array}{@{}lclr}
    \psem{e}^i_\gamma &=& \sem{e}^i_\gamma
        & \mathit{if}\ \kindof{e} = \typeD \\
    \psem{e}^s_\gamma &=& \fun{s} \fun{U} \delta_{\sem{e}^s_\gamma(s)}(U) &
    \mathit{if}\ \kindof{e} = \typeD
    \end{array}\\[1.5em]

    \begin{array}{@{}lcl}
    \psem{\zlm!sample($e$)!}^i_\gamma &=& \sem{e}^i_\gamma\\
    \psem{\zlm!sample($e$)!}^s_\gamma &=& 
    \begin{array}[t]{@{}l@{}}
      \fun{s}\fun{U}
      \letin{\mu, s' = \sem{e}^s_\gamma(s)}
      \int_T \mu(dv) \; \delta_{v, s'}(U)
    \end{array}
  \end{array}\\[1.5em]
  
  \begin{array}{@{}lcl}
  \psem{\zlm!observe($e_1$,\ $e_2$)!}^i_\gamma &=&
    \zlm!($\sem{e_1}^i_\gamma$,$\sem{e_2}^i_\gamma$)!
  \\
  \psem{\zlm!observe($e_1$,\ $e_2$)!}^s_\gamma &=&\\
    \multicolumn{3}{l}{
      \quad\begin{array}[t]{@{}l@{}}
      \fun{(s_1, s_2)}\fun{U}\\ 
      \quad \begin{array}[t]{@{}l@{}} 
      \letin{\mu,s_1' = \sem{e_1}^s_\gamma(s_1)}\\
      \letin{v,s_2' = \sem{e_2}^s_\gamma(s_2)}
      \mu_{\rm{pdf}}(v) * \delta_{\mathtt{()}, (s_1', s_2)}(U)
      \end{array}
    \end{array}}
  \end{array}\\[3.5em]

  \begin{array}{@{}l@{}}
    \left \{ \mkern-5mu \left [
    \begin{minipage}[t][1em]{22.1em}
    \begin{lstlisting}[aboveskip=-1.5em, xleftmargin=0pt, xrightmargin=0pt]
$e$ where rec init $x_1$ = $c_1$ ... and init $x_k$ = $c_k$
            and $y_1$ = $e_1$ ... and $y_n$ = $e_n$
  \end{lstlisting}
  \end{minipage}
  \right ] \mkern-5mu \right \}^s_\gamma
  =\\[1.1em]
  \quad\begin{array}{@{}l}
  \fun{((m_1, \dots, m_k), (s_1, \dots, s_n), s)} \fun{U}\\
  \quad \begin{array}[t]{@{}l@{}}
  \letin{\gamma_1 = \gamma[m_1/\zlm!$x_1$_last!]}
  \dots
  \letin{\gamma_k = \gamma_{k-1}[m_k/\zlm!$x_k$_last!]}\\
  \letin{\mu_1 = \psem{e_1}^s_{\gamma_k}(s_1)}\\
  \int \mu_1(dv_1, ds_1') \letin{\gamma'_1 = \gamma_k[v_1/  y_1]} \dots \\
  \quad\begin{array}[t]{@{}l@{}}
  \letin{\mu_n = \sem{e_n}^s_{\gamma'_{n-1}}(s_n)}\\
  \int \mu_n(dv_n, ds'_n)  \letin{\gamma'_n = \gamma'_  {n-1}[v_n/y_n]}\\
  \quad \begin{array}[t]{@{}l@{}}
  \letin{\mu = \psem{e}^s_{\gamma'_n}(s)}\\
  \int \mu(dv, ds') \; \delta_{v, ((\gamma'_n[x_1], \dots, \gamma'_n[x_k]), (s_1',  \dots, s_n'), s')}(U)
  \end{array}
  \end{array}
  \end{array}
  \end{array}
  \end{array}
  
\end{array}
  $$
\end{small}
\caption{Semantics of probabilistic expressions.}
\label{fig:sem-probabilistic}
\end{figure}

The semantics of probabilistic expressions is presented in \Cref{fig:sem-probabilistic}~(the complete semantics is in
\ifextended
\Cref{fig:sem-probabilistic-full} of \Cref{app:semantics}\else
Figure~14 of Appendix~A.2\fi).
This measure-based semantics is adapted from~\cite{staton17} to explicitly handle the state of the transition functions.

First, any deterministic expression can be lifted as a probabilistic expression.
The transition function returns the Dirac delta measure~({$\delta_x(U) = 1$ if $x \in U$, $0$ otherwise}) on the pair returned by the deterministic transition function applied on the current state: $\sem{e}^s_\gamma(s) : T \times S$.

The probabilistic operator \zl{sample($e$)} evaluates~$e$ which returns a distribution~$\mu : \zl{$T$ dist}$ and a new state~$s': S$, and returns a measure over the pair (result, state) where the state is fixed to the value~$s'$.
\zl{observe($e_1, e_2$)} evaluates~$e_1$ and~$e_2$ into a distribution with density~$\mu : \typedist{T}$ and a value~$v : T$, and weights execution paths using the likelihood of~$v$ w.r.t.~$\mu$~($\mu_{\rm{pdf}}$ denotes the density function of the distribution~$\mu$).

The state of a set of locally recursive definitions is the same as in \Cref{fig:sem_deterministic} and contains the previous value of the initialized variables and the states of the sub-expressions.
The transition starts by adding the variables \zl!$x_i$_last! to the environment.
We note $\int \mu(dv, ds) f(v, s)$ the integral of~$f$ w.r.t. the measure~$\mu$ where variables~$v$ and~$s$ are the integration variables. 
The integration measure appears on the right of the integral to maintain the expression order of the source code and we allow local definitions (e.g., $\letin{x = v} \dots$) inside the integral to simplify the presentation.
Local definitions are interpreted by successively integrating the measure on pairs (value, state) returned by the defining expressions.
In other words, we integrate over all possible executions.
Integrals need to be nested to capture the eventual dependencies in the successive expressions.
The returned value is a measure on pairs (value, state) where the state captures the value of the initialized variables and the state of the sub-expressions.

\paragraph{Inference in the loop.}
The \zl{infer} operator is the boundary between the deterministic and the probabilistic expressions. 
Given a probabilistic model defined by an expression~$e$, at each step the inference computes a distribution of results and a distribution of possible next states.
Expression~$e$ can contain free variables thus capturing inputs from deterministic processes.
The distribution of results can be used by deterministic processes to produce new inputs for the next inference step.

\begin{small}
$$
\begin{array}{@{}l}
\sem{\zl!infer($e$)!}^i_\gamma =
  \fun{U} \delta_{\sem{e}^i_\gamma}(U)\\
\sem{\zl!infer($e$)!}^s_\gamma = \fun{\sigma}
  \begin{array}[t]{@{}l}
  \letin{\mu = \fun{U} \int_{S} \sigma(ds) \psem{e}^s_\gamma(s)(U)}\\
  \letin{\nu = \fun{U} \mu(U) / \mu(\top)}\\
  (\pi_{1*}(\nu), \pi_{2*}(\nu))
  \end{array}
\end{array}
$$
\end{small}

\noindent
The state of \zl{infer($e$)} is a distribution over the possible states for~$e$. 
The initial state is the Dirac delta measure on the initial state of~$e$. 
The transition function integrates the measure defined by~$e$ over all the possible states and normalize the result~$\mu$ to produce a distribution~$\nu: \typesampler{T \times S}$ ($\top$ denotes the entire space).
This distribution is then split into a pair of marginal distributions using the pushforward of~$\mu$ across the projections $\pi_1$ and $\pi_2$. 

\paragraph{Remark.}
In \ProbZelus deterministic processes that interact with the environment cannot rollback on actions based on past estimations~(e.g., the command of the robot controller).
However, the inferred distribution of a random variable captured in the state may evolve at each time step based on subsequent observations (e.g., the initial position of the robot).
These two properties follow from the separation of the distribution of results and the distribution of states in the semantics of \zl{infer}. 

Alternatively, we could define a fix-point semantics of streams based on a Scott order, as simple lazy streams in Haskell, where the value of a stream can depend on future computations~(as illustrated
\ifextended
\Cref{app:sem-alt}\else
Appendix~A.3\fi).
However, this approach is not practical in a reactive context where processes interact with the environment~\cite{gilles1974semantics,caspi92}.

 \section{Compilation}
\label{sec:compilation}

Following the semantics described in \Cref{sec:co-iteration} each expression is compiled into a transition function that can be written in a simple functional first-order language extended with probabilistic operators we call~\muF.
Importantly, the compilation process is the same for deterministic and probabilistic expressions.
We can then give a classic interpretation to deterministic terms, and a measure-based semantics to probabilistic terms following~\cite{staton17}.
We then show that the semantics of the compiled code coincides with the co-iterative semantics described in \Cref{sec:co-iteration}.

\subsection{A First-Order Functional Probabilistic Language}

The syntax of \muF is the following:

\noindent
\begin{lstlisting}
$d ::=$ let $f$ = $e$ $|$ $d$ $d$
$e ::=$ $c$ $|$ $x$ $|$ $(e,e)$ $|$ $\mathit{op}$($e$) $|$ $e$($e$)
   $|$ if $e$ then $e$ else $e$ $|$ let $p$ = $e$ in $e$ $|$ fun $p$ -> $e$
   $|$ sample($e$) $|$ observe($e$, $e$) $|$ infer((fun $x$ -> $e$), $e$)
$p ::=$ $x$ $|$ ($p$, $p$)
\end{lstlisting}

\smallskip

\noindent
A program is a set of definitions.
An expression is either a constant, a variable, a pair, an operator, a function call, a conditional, a local definition, an anonymous function, or one of the probabilistic operators \zl{sample}, \zl{observe}, or \zl{infer}.
The \zl{infer} operator is tailored for \ProbZelus and always takes two arguments: a transition function of the form \zl{fun $x$ -> $e$}, and a distribution of states.
This operator computes the distribution of results and the distribution of possible next steps.
A type system similar to the one of \ProbZelus is used to distinguish deterministic from probabilistic expressions (see
\ifextended
\Cref{app:muF-typing}\else
Appendix~B.1\fi).

\subsection{Compilation to \muF}

The compilation function~$\compile$ generates a function that closely follows the transition function defined by the co-iterative semantics presented in \Cref{sec:co-iteration}.
Each expression is compiled into a function of type $S \rightarrow T \times S$ which given a state returns a value and an updated state (see \ifextended\Cref{app:compilation}\else Appendix C\fi{} for the complete definition).
The compilation of \zl{present} is thus:

\begin{lstlisting}[aboveskip=1em, belowskip=1em]
$\compile$$($present $e$ -> $e_1$ else $e_2$$)$ $=$ fun (s,s1,s2) -> 
  let v,zs' = $\compile$$($$e$$)$(s) in
  if v then let v1,s1' = $\compile$$($$e_1$$)$(s1) in (v1,(s',s1',s2))
  else let v2,s2' = $\compile$$($$e_2$$)$(s2) in (v2,(s',s1,s2'))
\end{lstlisting}

The probabilistic operators \zl{sample}, 
and \zl{observe} are treated as external operators.
The compilation generates code that simply calls the \muF version of these operators.
The compilation of \zl{infer} passes the distribution over states to the \muF version of \zl{infer}.
The inference is thus aware of the distribution over states at the previous step.

\begin{lstlisting}[aboveskip=1em, belowskip=1em]
$\compile$$($infer($e$)$)$ $=$ fun sigma ->
  let mu, sigma' = infer($\compile$$($$e$$)$, sigma) 
  in (mu, sigma')
\end{lstlisting}

The compilation of a node declaration generates two definitions: the transition function \zl!$f$_step! and the initial state \zl!$f$_init!.
The transition function is the result of compiling the body of the node with an additional argument to capture the input.
The initial step is generated by the allocation function~$\alloc$ which follows the definition of the initial state in the semantics of \Cref{sec:co-iteration} (see the \ifextended\Cref{app:compilation}\else Appendix C\fi{} for the complete definition).
\begin{lstlisting}[aboveskip=1em, belowskip=1em]
$\compile$$($let node $f$ $x$ = $e$$)$ =  
  let $f$_step = fun (s,$x$) -> $\compile$$($$e$$)$(s)
  let $f$_init = $\alloc$$($e$)$
\end{lstlisting}

\subsection{Semantics equivalence}
\label{sec:sem-muF}

We showed how to compile \ProbZelus to \muF a simple functional language with no loops, no recursion, and no higher-order functions, extended with the probabilistic operators.
This language is similar to the kernel presented in~\cite{staton17} for which a measure-based probabilistic semantics is defined~(see also
\ifextended
\Cref{fig:muF-sem-full} of \Cref{app:sem-muF}\else
Figure~16 of Appendix B.2\fi).

We can now prove that the semantics of the generated code corresponds to the semantics of the source language described in \Cref{sec:co-iteration}.

\begin{theorem*}
  For all \ProbZelus expression $e$, for all state $s$ and environment $\gamma$:
  \vspace{0.5em}
\begin{itemize}[itemsep=0.5em]
  \item if $\kindof{e} = {\typeD}$ then $\sem{e}^s_\gamma(s) = \sem{\compile(e)}_\gamma(s)$, and,
  \item if $\kindof{e} = {\typeP}$ then $\psem{e}^s_\gamma(s) = \psem{\compile(e)}_\gamma(s)$.
\end{itemize}
\end{theorem*}
\begin{proof}
  The proof is done by induction on the structure of~$e$.
As an example consider the expression \zl{sample($e$)}.
  If this expression is well-typed and since typing is preserved by compilation (see \ifextended\Cref{lem:typing}\else Lemma~C.1\fi{}) $\kindof{\compile(e)} = \typeD$. 
  Using the induction hypothesis on $\sem{\compile(\zl!$e$!)}_\gamma(s)$ = $\sem{$e$}_\gamma(s)$ we have:

\begin{small}
$$
\begin{array}{@{}l@{}}
    \psem{\compile(\zlm!sample($e$)!)}_\gamma(s)\\[1em] 
\quad =
\left \{ \mkern-5mu \left [
\begin{minipage}[t][1em]{19em}
\begin{lstlisting}[aboveskip=-1.5em, xleftmargin=0pt, xrightmargin=0pt]
fun s -> let mu, s' = $\compile$$($$e$$)$(s) in
         let v = sample(mu) in (v, s')
\end{lstlisting}
\end{minipage}
\right ] \mkern-5mu \right \}_\gamma\!\!\!\!(s)\\[1.5em]
\quad =
  \; \fun{U}
  \begin{array}[t]{@{}l@{}}
  \int \delta_{\sem{\compile(\zlm!$e$!)}_\gamma(s)}(d\mu, ds')\\[0.3em]
  \quad \psem{\zlm!let v = sample(mu) in (v, s')!}_{\gamma[\mu/\zlmm!mu!, \; s'/\zlmm!s'!]}(U)
  \end{array}\\[2.25em]
\quad =
  \; \fun{U}
  \begin{array}[t]{@{}l@{}}
  \letin{\zlm!($\mu$, $s'$)! = \sem{\compile(\zl!$e$!)}_\gamma(s)}\\[0.3em]
  \psem{\zlm!let v = sample(mu) in (v, s')!}_{\gamma[\mu/\zlmm!mu!, \; s'/\zlmm!s'!]}(U)
  \end{array}\\[2.25em]
\quad =
  \; \fun{U}
  \letin{(\mu, s') = \sem{\zl!$e$!}^s_\gamma(s)}
  \int \mu(dv) \delta_{v, s'}(U)\\[0.5em]
\quad = \psem{\zlm!sample($e$)!}^s_\gamma(s) \end{array}
$$
\end{small}
\end{proof}

\paragraph{Remark.}
The probabilistic semantics of \muF is commutative (see \cite[Theorem~4]{staton17}).
We can thus show that the semantics of a \ProbZelus program does not depend on the schedule used by the compiler to order the equations of local definitions.

 \section{Inference}
\label{sec:inference}

The measure-based semantics of \zl{infer} presented in \Cref{sec:co-iteration} and \Cref{sec:sem-muF} includes often intractable integrals.
An additional challenge is to design inference techniques that can operate in bounded memory to be practical in a reactive context where the inference is a non-terminating process.

In this section, we show how to adapt particle filtering~\cite{doucet-smc-2006}
to explicitly handle the state of the transition functions.
We then present a novel implementation of \emph{delayed sampling}, a recently proposed semi-symbolic inference, which enables partial exact inference over infinite streams in bounded memory for a large class of models including state-space models like the robot example of \Cref{fig:lqr_src} in \Cref{sec:example}.

\subsection{Particle Filtering}
\label{sec:pf}

In conventional probabilistic programming, the operational interpretation of a model is an \emph{importance sampler}
that randomly generates a sample of the model together with an importance weight measuring the quality of the sample.

Following the conventions of \Cref{sec:co-iteration} we write $\sem{e}_\gamma$ for the semantics of a deterministic expression, and $\psem{e}_{\gamma, w}$ for the semantics of a probabilistic expression.
The additional argument~$w$ captures the weight.
The probabilistic operator \zl{sample} draws a sample from a distribution without changing the score.
\zl{observe} increments the score by the likelihood of the observation.
A deterministic expression can be lifted in a probabilistic context: the corresponding sample is the return value of the expression and the score is unchanged.
The \zl{let} construct illustrates that the score is accumulated along the execution path (the complete semantics is presented in
\ifextended
\Cref{fig:importance-full} of \Cref{app:inference}\else
Figure~19 of Appendix~D\fi).

\noindent
\begin{small}
$$
\begin{array}{@{}l@{\ }l@{\ \,}l}
\psem{\zlm!sample($e$)!}_{\gamma, w} &=&
  ({\rm draw}(\sem{e}_\gamma), w)
\\[0.75em]

\psem{\zlm!observe($e_1$, $e_2$)!}_{\gamma, w} &=&
  \letin{\mu = \sem{e_1}_\gamma} (\zlm!()!, w * \mu_{\textrm{pdf}}(\sem{e_2}_\gamma))
\\[0.75em]

\psem{e}_{\gamma, w} &=&
  (\sem{e}_\gamma, w) \quad \mathit{if} \; \kindof{$e$} = {\typeD}
\\[0.75em]

\psem{\zlm!let $\;p$ =$\;e_1\;$in $\;e_2$!}_{\gamma, w} &=&
  \letin{v_1, w_1 = \psem{e_1}_{\gamma, w}} \psem{e_2}_{\gamma[v_1 / p], w_1}
\end{array}
$$
\end{small}

Such a sampler is the basis of a \emph{particle filter} or a \emph{bootstrap filter}~\cite{doucet-smc-2006}.
\zl{infer} independently launches~$N$ particles.
At each step, each particle samples the distribution of states~$\sigma$ obtained at the previous step and executes the sampler to compute a pair (result, state) along with its weight.
The resulting pairs are then normalized according to their weights to form a categorical distribution~$\mu$ over pairs of values and states (we write $\overline{w_i} = w_i / \sum_{i=1}^N w_i$ for the normalized weights).
This distribution is then split into the distribution of returned values and the distribution of next states.

\begin{small}
$$
\begin{array}{l}
\sem{\zlm!infer(fun\ $s$ ->\ $e$,\ $\sigma$)!}_\gamma = \\ \quad
  \begin{array}[t]{@{}l@{}}
  \letin{\mu =
    \begin{array}[t]{@{}l@{}}
    \fun{U} \sum\limits_{i=1}^N \;
      \begin{array}[t]{@{}l@{}}
        \letin{s_i = \rm{draw}(\sem{\sigma}_\gamma)}\\
        \letin{(v_i, s_i'), w_i = \psem{\zlm!fun\ $s$ ->\ $e$!}_{\gamma, 1}(s_i)}\\
        \overline{w_i} *  \delta_{v_i, s_i'}(U)
      \end{array}
    \end{array}\\}
    (\pi_{1*}(\mu), \pi_{2*}(\mu))
  \end{array}
\end{array}
$$
\end{small}

\paragraph{Remark.}
The resampling step requires the ability to clone particles in the middle of the execution.
A classic technique is to compile the model in \emph{continuation passing style}~(CPS)~\cite{ritchie_c3_2016} and use the probabilistic constructs \zl{sample} and \zl{observe} as checkpoints for resampling.
In our context, the compilation presented in \Cref{sec:compilation} externalizes the state of the transition functions.
Duplicating the state effectively clones a particle during its execution.
The code does not need to be compiled in CPS form and we avoid the alignment problem~\cite{lunden19alignment}.

\subsection{Delayed Sampling}
\label{sec:ds}

\newcommand{\dsassume}{\ensuremath{\mathit{assume}}}
\newcommand{\dsobserve}{\ensuremath{\mathit{observe}}}
\newcommand{\dsvalue}{\ensuremath{\mathit{value}}}
\newcommand{\dsdistribution}{\ensuremath{\mathit{distribution}}}

\newcommand{\dsllinitialize}{\ensuremath{\mathit{initialize}}}
\newcommand{\dsllmarginalize}{\ensuremath{\mathit{marginalize}}}
\newcommand{\dsllrealize}{\ensuremath{\mathit{realize}}}
\newcommand{\dsllcondition}{\ensuremath{\mathit{condition}}}
\newcommand{\dsllinitializeshort}{\ensuremath{\mathit{init}}}
\newcommand{\dsllmarginalizeshort}{\ensuremath{\mathit{marg}}}

The basis of our streaming inference algorithm is delayed sampling, which we first review to explain its conceptual approach.
Delayed sampling is an inference technique combining partial exact inference with approximate particle filtering to reduce estimation errors~\cite{murray18delayed_sampling,lunden17}.

In addition to the importance weight, each particle exploits \emph{conjugacy} relationships between pairs of random variables to maintains a graph: a \emph{Bayesian network} representing closed-form distributions involving subsets of random variables.
Observations are incorporated by analytically \emph{conditioning} the network.
Particles are thus only required to draw sample when forced to, i.e., when exact computation is not possible, or when a concrete value is required.

To perform analytic computations, delayed sampling manipulates symbolic terms where random variables are referenced in the graph.
The semantics of an expression~$\psem{e}_{\gamma, g, w}$ takes an additional argument~$g$ for the graph and returns a symbolic term, an updated weight, and an updated graph.
Given a graph, a symbolic term can be evaluated into a concrete value by sampling the random variables that appear in the term.
The graph can be accessed and modified using the three following functions defined in~\cite{murray18delayed_sampling}.

\begin{figure}
  \begin{small}
  $$
    \begin{array}{l}
\psem{\zlm!$op$($e$)!}_{\gamma, g, w} ~ = \\ \qquad
       \mathit{let} \; (e', g_e, w_e) = \psem{e}_{\gamma, g, w} \; \mathit{in} \;
       (\zlm!app!(\textit{op}, e'), g_e, w_e)
     \\[0.5em]
\psem{\zlm!if$\;e\;$then$\;e_1\;$else$\;e_2$!}_{\gamma, g, w} ~= \\ \qquad
       \begin{array}{@{}l@{}}
       \letin{e', g_e, w_e = \psem{e}_{\gamma, g, w}} \\
       \letin{v, g_v = \dsvalue(e', g_e)}\\
       \mathit{if} \; v \; \mathit{then} \;
           \psem{e_1}_{\gamma, g_v, w_e} \;
       \mathit{else} \;
           \psem{e_2}_{\gamma, g_v, w_e} \;
      \end{array}
      \\[2em]
\psem{\zlm!sample($e$)!}_{\gamma, g, w} ~ = \\ \qquad
      \begin{array}{@{}l@{}}
        \letin{\mu, g_e, w' = \psem{e}_{\gamma, g, w}}\\
        \letin{X, g' = \dsassume (\mu, g_e)}
        (X, g', w')
      \end{array}
      \\[1em]
\psem{\zlm!observe($e_1$, $e_2$)!}_{\gamma, g, w} ~= \\ \qquad
      \begin{array}{@{}l@{}}
        \letin{\mu, g_1, w_1 \; = \psem{e_1}_{\gamma, g, w}}
        \letin{X, g_x = \dsassume(\mu, g_1)}\\
        \letin{e'_2, g_2, w_2 = \psem{e_2}_{\gamma, g_x, w_1}}
        \letin{v, g_v = \dsvalue(e'_2, g_2)}\\
        \letin{g' = \dsobserve(X, v, g_v)}
        (\zlm!()!, g', w_2 * \mu_{\rm pdf}(v))
      \end{array}
      \\

\end{array}
  $$
  \end{small}
\caption{Delayed sampling sampler. Expressions return a pair (symbolic expression, weight).}
  \label{fig:ds-sem}
  \end{figure}

\begin{description}[itemsep=0.3em]
\item[$v, g' = \dsvalue(e, g)$]
  evaluate a symbolic term and return a concrete value.
\item[$X, g' = \dsassume(\mu, g)$]
  add a random variable~$X \sim \mu$ to the graph and return the variable.
\item[$g' = \dsobserve(X, v, g)$]
condition the graph by observing the value~$v$ for the variable~$X$.
\end{description}

Compared to the particle filter, any expression, probabilistic or deterministic, can contribute to a symbolic term.
The evaluation function~$\psem{e}_{\gamma, g, w}$ partially presented in \Cref{fig:ds-sem} must thus be defined on the entire language and not only on probabilistic constructs.
For instance, the application of an operator \zl!$\textit{op}$($e$)! returns a symbolic term $\zl!app!(op, e')$ that represents the application of $\mathit{op}$ on the evaluation of~$e$.
Some terms are partially evaluated when symbolic computation is not possible.
For instance, in the general case, to compute the importance weight of \zl!if $e$ then $e_1$ else $e_2$!, each particle must compute a concrete value for the condition~$e$.

The probabilistic \zl!sample(e)! adds a new random variable to the graph without drawing a sample.
\zl!observe($e_1$, $e_2$)! adds a new random variable $X \sim \mu$ where~$\mu$ is defined by~$e_1$, then computes a concrete value~$v$ for~$e_2$ and conditions the graph by observing the value~$v$ for~$X$.
As for the particle filter, the score is incremented by the likelihood of the observation.

\paragraph{Symbolic Computations.}
The functions $\dsvalue$, $\dsassume$, and $\dsobserve$ used in \Cref{fig:ds-sem} rely on the following mutually recursive lower level operations~($Y$ is the parent of~$X$)~\cite{murray18delayed_sampling}:

\begin{description}[itemsep=0.3em]
\item[$X, g' = \dsllinitialize(\mu, Y, g)$]
    add a new node $X$ with a distribution $p_{X|Y}=\mu$ as a child of $Y$ in $g$.
\item[$g' = \dsllmarginalize(X, g)$]
    compute and store $p_X$ in $g'$ from $p_Y$ and $p_{X|Y}$ where $p_Y$ and $p_{X|Y}$ are in $g$.
\item[$g' = \dsllrealize(X,v,g)$] assign in $g'$ a concrete value to a random variable $X$.
\item[$g' = \dsllcondition(Y, g)$]
  compute $p_{Y|X}$ given $p_X$, $p_{X|Y}$, and a concrete value $X = v$ where $v$ is in $g$.
\end{description}
\noindent
In the class of Bayesian networks maintained by the delayed sampler, marginalization w.r.t. a parent node, and conditioning a parent on the value of a child are tractable operations.

To reflect these operations, nodes are characterized by a state (see \Cref{fig:ds_kalman_naive,fig:bds_kalman}).
\emph{Initialized} nodes~\tikz[baseline=-0.75ex]{\node[initialized] (x1) {};} are random variables with a conditional distribution $p_{X|Y}$ where the parent~$Y$ has no concrete value yet.
\emph{Marginalized} nodes~\tikz[baseline=-0.75ex]{\node[marginalized] (x1) {};} are random variables with a marginal distribution~$p_X$ that incorporate the distributions of the ancestors.
\emph{Realized} nodes~\tikz[baseline=-0.75ex]{\node[realized] (x1) {};} are random variables that have been assigned a concrete value via sampling or observation.

The evaluation function $\dsvalue(e, g)$ forces the realization by sampling of all the random variables referenced in~$e$ to produce a concrete value.
Similarly, the function $\dsobserve(X, v, g)$ realizes a variable~$X$ with a given observation~$v$.
The realization of a random variable comprises three steps: (1)~compute the distribution $p(x)$ by recursively marginalizing the parents from a root node, (2)~sample a value, or use the observation, and (3)~use the concrete value to update the children and condition the parent which removes the dependencies.

The function $\dsassume(\mu, g)$ adds a new node to the graph and is defined case by case on the shape of the symbolic term~$\mu$.
If there is a conjugacy relationship between~$\mu$ and a random variable~$Y$ present in the graph, e.g., $\mu = \mathit{Bernoulli}(Y)$ with $Y \sim \mathit{Beta}(\alpha, \beta)$, a new initialized node~$X \sim \mu$ is added as a child of~$Y$.
Otherwise, since symbolic computation is not possible, dependencies are broken by realizing the random variables that appear in~$\mu$, e.g.,~$\mu' = \mathit{Bernoulli}(\dsvalue(Y, g))$, and~$X \sim \mu'$ is added as a new root node.

\begin{figure*}
\begin{footnotesize}
\begin{tabular}{ccccccc}
\begin{subfigure}[b]{.11\textwidth}
\centering
\begin{tikzpicture}[node distance=0.4cm and 0.4cm,every node/.style={font=\scriptsize, align=center}]
    \node (past) {{\normalsize\dots}};
    \node[marginalized, label={\texttt{pre x}}, right=of past] (x1) {};
    \dep[-stealth]{past}{x1};
    \node[realized, below=of x1] (o1) {};
    \dep[-stealth]{x1}{o1};
    \link[->, densely dotted]{x1}{o1};
\end{tikzpicture}
\caption{Initial state}
\end{subfigure}
& \begin{subfigure}[b]{.13\textwidth}
\centering
\begin{tikzpicture}[node distance=0.4cm and 0.4cm,every node/.style={font=\scriptsize, align=center}]
    \node (past) {{\normalsize\dots}};
    \node[marginalized, label={\texttt{pre x}}, right=of past] (x1) {};
    \dep[-stealth]{past}{x1};
    \node[realized, below=of x1] (o1) {};
    \dep[-stealth]{x1}{o1};
    \link[->, densely dotted]{x1}{o1};
    \node[initialized, label={\texttt{\phantom{p}x}}, right=of x1] (x2) {};
    \dep[-stealth]{x1}{x2};
    \link[->, densely dotted]{x2}{x1};
    \node[label={[label distance=0.1ex]0:\phantom{\texttt{acc}}}, below=of x2] (o2) {};
\end{tikzpicture}
\caption{$\dsllinitializeshort(\zlm!x!, \zlm!pre x!)$}
\label{fig:ds_kalman_addx}
\end{subfigure}
& \begin{subfigure}[b]{.12\textwidth}
\centering
\begin{tikzpicture}[node distance=0.4cm and 0.4cm,every node/.style={font=\scriptsize, align=center}]
    \node (past) {{\normalsize\dots}};
    \node[marginalized, label={\texttt{pre x}}, right=of past] (x1) {};
    \dep[-stealth]{past}{x1};
    \node[realized, below=of x1] (o1) {};
    \dep[-stealth]{x1}{o1};
    \link[->, densely dotted]{x1}{o1};
    \node[initialized, label={\texttt{\phantom{p}x}}, right=of x1] (x2) {};
    \dep[-stealth]{x1}{x2};
    \link[->, densely dotted]{x2}{x1};
    \node[initialized, label={[label distance=0.1ex]0:\texttt{acc}}, below=of x2] (o2) {};
    \dep[-stealth]{x2}{o2};
    \link[->, densely dotted]{o2}{x2};
\end{tikzpicture}
\caption{$\dsllinitializeshort(\zlm!y!, \zlm!x!)$}
\label{fig:ds_kalman_addo}
\end{subfigure}
& \begin{subfigure}[b]{.12\textwidth}
\centering
\begin{tikzpicture}[node distance=0.4cm and 0.4cm,every node/.style={font=\scriptsize, align=center}]
    \node (past) {{\normalsize\dots}};
b
    \node[marginalized, label={\texttt{pre x}}, right=of past] (x1) {};
    \dep[-stealth]{past}{x1};
    \node[realized, below=of x1] (o1) {};
    \dep[-stealth]{x1}{o1};
    \link[->, densely dotted]{x1}{o1};
    \node[marginalized, label={\texttt{\phantom{p}x}}, right=of x1] (x2) {};
    \dep[-stealth]{x1}{x2};
    \link[->, densely dotted]{x1}{x2};
    \link[->, densely dotted]{x2}{x1}[\xmark];
    \node[initialized, label={[label distance=0.1ex]0:\texttt{acc}}, below=of x2] (o2) {};
    \dep[-stealth]{x2}{o2};
    \link[->, densely dotted]{o2}{x2};
\end{tikzpicture}
\caption{$\dsllmarginalizeshort(\zlm!x!)$}
\label{fig:ds_kalman_margx}
\end{subfigure}
&
\begin{subfigure}[b]{.12\textwidth}
\centering
\begin{tikzpicture}[node distance=0.4cm and 0.4cm,every node/.style={font=\scriptsize, align=center}]
    \node (past) {{\normalsize\dots}};
    \node[marginalized, label={\texttt{pre x}}, right=of past] (x1) {};
    \dep[-stealth]{past}{x1};
    \node[realized, below=of x1] (o1) {};
    \dep[-stealth]{x1}{o1};
    \link[->, densely dotted]{x1}{o1};
    \node[marginalized, label={\texttt{\phantom{p}x}}, right=of x1] (x2) {};
    \dep[-stealth]{x1}{x2};
    \link[->, densely dotted]{x1}{x2};
    \node[marginalized, label={[label distance=0.1ex]0:\texttt{acc}}, below=of x2] (o2) {};
    \dep[-stealth]{x2}{o2};
    \link[->, densely dotted]{x2}{o2};
    \link[->, densely dotted]{o2}{x2}[\xmark];
\end{tikzpicture}
\caption{$\dsllmarginalizeshort(\zlm!y!)$}
\label{fig:ds_kalman_margo}
\end{subfigure}
& \begin{subfigure}[b]{.12\textwidth}
\centering
\begin{tikzpicture}[node distance=0.4cm and 0.4cm, every node/.style={font=\scriptsize}, ]
    \node (past) {{\normalsize\dots}};
    \node[marginalized, label={\texttt{pre x}}, right=of past] (x1) {};
    \dep[-stealth]{past}{x1};
    \node[realized, below=of x1] (o1) {};
    \dep[-stealth]{x1}{o1};
    \link[->, densely dotted]{x1}{o1};
    \node[marginalized, label={\texttt{\phantom{p}x}}, right=of x1] (x2) {};
    \dep[-stealth]{x1}{x2};
    \link[->, densely dotted]{x1}{x2};
    \node[realized, label={[label distance=0.1ex]0:\texttt{acc}}, below=of x2] (o2) {};
    \dep[-stealth]{x2}{o2};
    \link[->, densely dotted]{x2}{o2};
\end{tikzpicture}
\caption{$\dsllrealize(\zlm!y!)$}
\label{fig:ds_kalman_realo}
\end{subfigure}
& \begin{subfigure}[b]{.12\textwidth}
\centering
\begin{tikzpicture}[node distance=0.4cm and 0.4cm,every node/.style={font=\scriptsize, align=center}]
    \node (past) {{\normalsize\dots}};
    \node[marginalized, right=of past] (x1) {};
    \path let \p1 = (x1) in node  at (\x1,\y1) {\rmark};
    \dep[-stealth]{past}{x1};
    \node[realized, below=of x1] (o1) {};
    \path let \p2 = (o1) in node  at (\x2,\y2) {\rmark};
    \dep[-stealth]{x1}{o1};
    \link[->, densely dotted]{x1}{o1};
    \node[marginalized, label={\texttt{pre x}}, right=of x1] (x2) {};
    \dep[-stealth]{x1}{x2};
    \link[->, densely dotted]{x1}{x2};
    \node[realized, label={[label distance=0.1ex]0:\phantom{\texttt{acc}}}, below=of x2] (o2) {};
    \dep[-stealth]{x2}{o2};
    \link[->, densely dotted]{x2}{o2};
\end{tikzpicture}
\caption{Update state}
\label{fig:ds_kalman_update}
\end{subfigure}
\end{tabular}
\end{footnotesize}
\caption{One step of the robot example of \Cref{fig:lqr_src} with SDS.
Plain arrows represent dependencies and dotted arrows represent pointers at runtime.
The \zlm{sample} statement adds the initialized nodes~\zlm{x}~(b).
The \zlm{observe} statement adds the initialized node~\zlm{a}~(c),
triggers the marginalizations of~\zlm{x}~(d) and~\zlm{a}~(e),
and assigns to~\zlm{a} its observed value~(f).
When the state is updated, the value \zlm{x} becomes \zlm{pre x}. 
The previous values are not referenced anymore and can be removed~(g).}
\label{fig:bds_kalman}
\end{figure*}
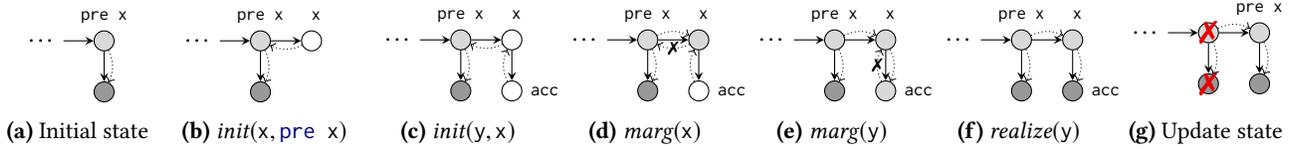

\paragraph{Inference.}
The inference scheme is similar to the particle filter.
At each step, the inference draws~$N$ states from~$\sigma$ to execute the transition function.
For each particle, execution starts with the graph computed at the previous step and returns a pair of symbolic terms (result, state), the particle weight, and the updated graph.
The function~$\dsdistribution(e, g)$ returns the distribution corresponding to the expression~$e$ without altering the graph.
Results are then aggregated in a mixture distribution (concrete values are lifted to Dirac distribution) where the distribution~$d_i$ operates on the value component of~$U$ and we use the pair (symbolic term, graph) computed by the transition function for the distribution of state.
This distribution is then split into the distribution of returned values and the distribution of next states.
\vspace{-0.5em}
\begin{center}
  \begin{small}
  $
  \begin{array}{l}
  \sem{\zlm!infer(fun\ $s$ ->\ $e$,\ $\sigma$)!}_\gamma = \\ \quad
    \begin{array}[t]{@{}l@{}}
    \letin{\mu =
      \begin{array}[t]{@{}l@{}}
      \fun{U} \sum\limits_{i=1}^N \;
        \begin{array}[t]{@{}l@{}}
          \letin{s_i, g_i = \rm{draw}(\sem{\sigma}_\gamma)}\\
          \letin{(e_i, s_i'), w_i, g_i' = \psem{\zlm!fun\ $s$ ->\ $e$!}_{\gamma, 1, g_i}(s_i)}\\
          \letin{d_i = \dsdistribution(e_i, g_i')}\\
          \overline{w_i} *  d_i(\pi_1(U)) * \delta_{s_i', g_i'}(\pi_2(U))
        \end{array}
      \end{array}\\}
      (\pi_{1*}(\mu), \pi_{2*}(\mu))
    \end{array}
  \end{array}
  $
  \end{small}
\end{center}

\subsection{Streaming Delayed Sampling}
\label{sec:sds}

As illustrated in \Cref{sec:example-inference}, a notable challenge with the traditional delayed sampling algorithm is that graph grows linearly in the number of samples. 
In the original formulation of delayed sampling~\cite{murray18delayed_sampling},
graph edges are only removed when a node is realized.
All nodes that have been neither sampled nor observed are thus kept in the graph even if they are no longer referenced by the program.
In a reactive programming context, such an implementation can consume unbounded memory.

\paragraph{Bounded Delayed Sampling.}
A simple mitigation is to limit the scope of symbolic computations to one time step and discard the graph at the end of each time step.
We call this inference technique \emph{bounded delayed sampling}~(BDS).

BDS performs symbolic computations during the execution of a time step, and whenever possible, delays the sampling until the end of the instant.
Like the particle filter, BDS guarantees a bounded-memory execution.
For each particle, the size of the graph is bounded by the number of variables introduced during a time step, which by construction, is bounded for any valid \ProbZelus program. 

\paragraph{Streaming Delayed Sampling.}
Compared to the original delayed sampling algorithm, BDS loses the ability to perform symbolic computations using variables defined at different time steps.
This can result in a significant loss of precision for models with inter-steps dependencies such as the robot example of \Cref{fig:lqr_src}.
To adapt delayed sampling to streaming settings while keeping its maximum accuracy, we designed a delayed sampler that is \emph{pointer-minimal} where nodes that are no longer referenced by the program can be eventually removed.
We call this inference \emph{streaming delayed sampling}~(SDS).
SDS enables partial exact inference in bounded memory for a large class of models.

In the original implementation of delayed sampling, graph nodes need to access their parents and children.
Marginalization requires access to the parent to incorporate the ancestor distribution.
Realization requires access to both the parent and the children of a node to update their respective distributions with the concrete value assigned to the node.

In the pointer minimal implementation, initialized nodes only keep a pointer to their parent to follow the ancestor chain during marginalization and marginalized nodes only keep a pointer to a marginalized child (Delayed Sampling imposes that a node always has at most one marginalized child).
Compared to the original implementation, marginal nodes only keep track of one child, and marginalization turns backward pointers to the parent node into forward pointers to the marginalized child.
Note that this implementation prevents updating the children when the parent is realized, and prevents conditioning a parent when a child is realized.
Instead, when marginalizing a node, the sampler first checks if the parent is realized to apply the update.
Symmetrically, to realize a node, the sampler first checks if the children are realized and, if necessary, conditions the distribution before assigning the concrete value.

\Cref{fig:bds_kalman} shows the evolution of the graph during one step for the robot example of \Cref{fig:lqr_src}.
At the end of the step, the value of \zl!pre x! is updated.
The previous value is not referenced anymore by the program and the node can be removed from the graph.
In the original implementation, backward pointers between marginalized nodes prevent the collection (see \Cref{fig:ds_kalman_naive}).

\paragraph{Limitations.}
With SDS, models like the robot example that only maintain bounded chains of dependencies between variables are guaranteed to be executed in bounded memory.
The class of models that can be executed in bounded memory with our pointer-minimal implementation thus already comprises state-space models like Kalman filters, and models for learning unknown constant parameters from a series of observations (e.g., computing the bias of a coin from a succession of flips) where variables introduced at each step are immediately realized.

However, unbounded chains can still be formed if the program keeps a reference to a constant variable that is never realized.
In the following example, at each step, a new variable \zl{x} is added as a child of \zl{pre x} and then marginalized for the observation.
But \zl{p1} keeps a reference to the initial variable \zl{i} which is never realized and thus forms an unbounded chain between \zl{i} and \zl{x}.

\begin{lstlisting}
let proba p1 (xo, obs) = (i, o) where
  rec init i = sample (gaussian (xo, 1.))
  and x = sample (gaussian (i -> pre x, 1.))
  and () = observe (gaussian (x, 1.), obs)
\end{lstlisting}

In addition, in \ProbZelus, at each step the inference returns a snapshot of the current distribution without forcing the realization of any node in the graph.
Compared to the original delayed sampling implementation, initialized nodes can be inspected without being realized.
It is thus possible to form unbounded chains of initialized nodes which cannot be pruned even when nodes are no longer referenced in the program due to the backward pointers to the parent in initialized nodes.
In the following example, at each step, the variable \zl{x} is added as a child of \zl{pre x}, but without observation these variables remain initialized for ever.

\begin{lstlisting}
let proba p2 (xo) = x where
  rec x = sample (gaussian (xo -> pre x, 1.))
\end{lstlisting}

To mitigate these issues, we can force the realization of trailing nodes at each step as in \emph{bounded delayed sampling} or use a sliding window.
Alternatively, the $\dsvalue$ (\zl{eval}) function is available to the programmer and can be used to implement any strategy to force the evaluation of the nodes.
For instance, the previous example can be adapted to execute in bounded~memory:
\begin{lstlisting}
let proba p2' (xo) = x where
  rec x = sample (gaussian (xo -> pre x, 1.))
  and _ = eval (xo -> pre x)
\end{lstlisting}

 \section{Evaluation}
\label{sec:eval}
\pgfplotstableread{
x	y	y-min	y-max	p
Beta-Bernoulli	40.965897	0.7262259999999969	1.0591480000000004	200
Gaussian-Gaussian	382.816989	18.293389999999988	26.311890000000005	900
Kalman-1D	0.351831	0.0018060000000000298	0.025310999999999972	1
Outlier	32.408477	0.3080449999999999	0.972565000000003	65
Robot	212.052246	13.782837999999998	76.51051199999998	7
SLAM	9900	0.0	0.0	{> 9,500\phantom{XX}}
MTT	9900	0.0	0.0	{> 750\phantom{XX}}
}{\bdstimeout}

\pgfplotstableread{
x	y	y-min	y-max	p
Beta-Bernoulli	0.478769	0.01849400000000001	0.08760500000000004	1
Gaussian-Gaussian	62.674843	2.6271850000000043	2.361094999999999	150
Kalman-1D	0.309388	0.006307000000000007	0.021992999999999985	1
Outlier	44.302562	0.7344000000000008	2.684731999999997	65
Robot	74.901168	4.256067000000002	76.76723	1
SLAM	2066.227669	376.7682419999999	377.5444080000002	700
MTT	2712.195982	334.111367	781.8919329999999	60
}{\sdstimeout}

\pgfplotstableread{
x	y	y-min	y-max	p
Beta-Bernoulli	22.350183	0.615006000000001	1.002883999999998	200
Gaussian-Gaussian	899.991385	23.69422700000007	39.55073499999992	3,500
Kalman-1D	3.111199	0.0440520000000002	1.8031040000000003	15
Outlier	198.834664	1.7712370000000135	3.006056000000001	650
Robot	604.295107	22.403173000000038	151.61449099999993	85
SLAM	9900	0.0	0.0	{> 15,000\phantom{XX}}
MTT	9900	0.0	0.0	{> 2,500\phantom{XX}}
}{\particlestimeout}

\begin{figure*}
\centering
\begin{tikzpicture}
\begin{axis} [
    ybar,
    ymode=log,
    ymax=9900,
after end axis/.code={ \draw [ultra thick, white, decoration={snake, amplitude=1pt}, decorate] (rel axis cs:0,0.95) -- (rel axis cs:1,0.95);
        },
    clip=false,
    width=0.95\textwidth,
    height=15.5em,
    log origin=infty,
    major grid style={draw=black!30},
    ymajorgrids, tick align=outside,
    axis x line*=bottom,
    y axis line style={opacity=0},
    ylabel = {\small\sffamily Total Execution Time (ms)},
    max space between ticks=20,
    x axis line style={opacity=0},
    x tick label style={major tick length=0pt},
    xticklabels={Beta-\\[-0.25em]Bernoulli, Gaussian-\\[-0.25em]Gaussian, Kalman-1D, Outlier, Robot, SLAM, MTT},
    xticklabel style={align=center},
    xtick=data,
    tick label style={font=\footnotesize\sffamily},
    visualization depends on=y \as \rawy,
    every node near coord/.append style={
        font=\footnotesize,
        rotate=90,
        anchor=east,
        xshift=-0.2em,
    },
    nodes near coords,
    point meta=explicit symbolic
]

\addplot [draw=none, 
          fill=pfcolor!50,
          error bars/.cd, 
          error bar style={pfcolor!50!black}] 
  plot [error bars/.cd, y dir=both, y explicit]
  table [x expr=\coordindex, y error plus=y-max, y error minus=y-min, meta=p]{\particlestimeout};
    
\addplot [draw=none,
          fill=bdscolor!50,
          error bars/.cd, 
          error bar style={bdscolor!50!black}]  
  plot [error bars/.cd, y dir=both, y explicit]
  table [x expr=\coordindex, y error plus=y-max, y error minus=y-min, meta=p]{\bdstimeout};
  
\addplot [draw=none,
          fill=pmdscolor!50,
          error bars/.cd, 
          error bar style={pmdscolor!50!black}]
  plot [error bars/.cd, y dir=both, y explicit]
  table [x expr=\coordindex, y error plus=y-max, y error minus=y-min, meta=p]{\sdstimeout};
  
\end{axis}  
\node at (6, 3.5) {\small\sf 
  {\color{pfcolor!50}{$\blacksquare$}}~PF $\quad$
  {\color{bdscolor!50}{$\blacksquare$}}~BDS $\quad$
  {\color{pmdscolor!50}{$\blacksquare$}}~SDS}; 
\end{tikzpicture}
\vspace{-0.4em}
\caption{Execution time comparison when $90$\% of $1000$~runs reach an accuracy similar to the baseline (median accuracy of SDS with $1000$ particles) after $500$ steps. The number of particles required to reach this accuracy is shown on top of the bars. The error bars show the $10$th and $90$th percentiles.}
\label{fig:evaluation}
\end{figure*}
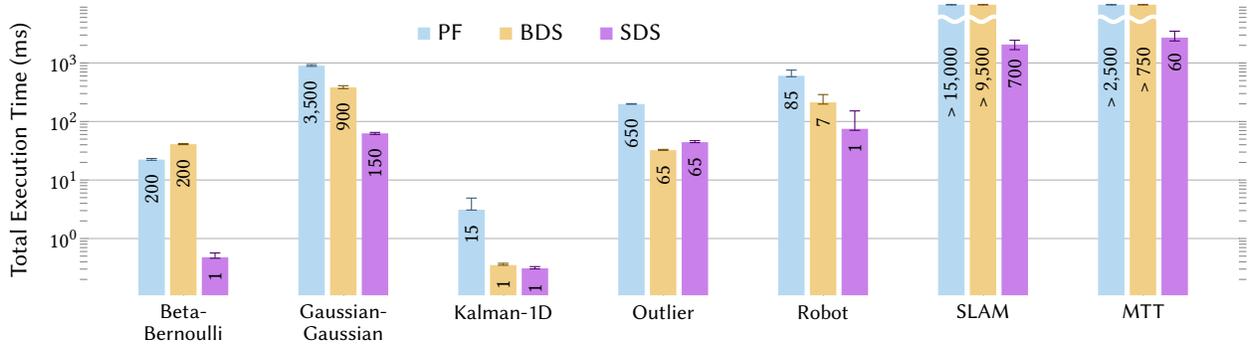

\begin{table}
\caption{Benchmarks with: inference of \emph{fixed} parameters from observations, estimation of a \emph{moving} state (state-space model), and \emph{inference-in-the-loop}~(IITL).}
\begin{small}
\begin{tabular}{@{}lcccc@{}}
&  Fixed
&  Moving
&  IITL
&  Metric \\
\toprule
\textsf{Beta-Bernoulli} & \checkmark & & & MSE\\\textsf{Gaussian-Gaussian} & \checkmark &  & & MSE\\\textsf{Kalman-1D} & & \checkmark & & MSE\\\textsf{Outlier} & \checkmark & \checkmark  & & MSE\\\textsf{Robot} & & \checkmark & \checkmark  & LQR\\\textsf{SLAM} & \checkmark & \checkmark &  \checkmark & MSE\\\textsf{MTT} &  & \checkmark &  & MOTA* \\
\bottomrule 
\end{tabular}
\end{small}
\label{tab:benchmarks}
\end{table}

We next evaluate the performance of \ProbZelus on a set of benchmarks that illustrate multiple aspects of the language: inferring fixed parameters from observations, online trajectory estimation, inference-in-the-loop.
For these examples, we compare the accuracy and the latency cost of the three inference techniques: PF, BDS, and SDS.
\ifextended\Cref{app:implementation}\else Appendix~E\fi{} details our implementation as an extension of the Zelus compiler.

\paragraph{Benchmarks.}
The models used in the experiments are summarized in \Cref{tab:benchmarks} (a detailed description along with the code of the benchmarks is given in
\ifextended
\Cref{sec:benchmarks}\else
Appendix~F.1\fi).
Two models infer fixed parameters from a series of observations.
\textsf{Beta-Bernoulli} estimates the parameter of a Bernoulli distribution from a series of binary observations (e.g., the bias of a coin).
\textsf{Gaussian-Gaussian} estimates the mean and variance of a Gaussian distribution from a series of observations.
The accuracy metric is the \emph{Mean Squared Error}~(MSE) of the inferred parameters compared to their exact values.

Two models infer the state of a moving agent from noisy observations.
\textsf{Kalman-1D} is a one-dimensional Kalman filter that models an agent that estimates its trajectory from noisy observations.
\textsf{Outlier} adapted from~\cite{ep} models the same situation as \textsf{Kalman-1D}, but the sensor occasionally produces invalid readings.
This models infer both the trajectory of the agent, and the bias of the sensor.
The accuracy metric is the \emph{Mean Squared Error}~(MSE) of the inferred trajectory compared to the exact positions.

Two models use inference-in-the-loop~(IITL).
\textsf{Robot} is the robot example of \Cref{sec:example} and the accuracy metric is the LQR loss.
\textsf{SLAM} (\emph{Simultaneous Localization and Mapping}) adapted from~\cite{rbpf} models an agent that estimates its position and a map of its environment. In this simplified version, the agent moves in a one-dimensional grid where each cell is either black or white.
The robot's wheels may slip causing the robot to unknowingly stay in place (noisy motion), and the sensor is not perfect and may accidentally report the wrong color (noisy observations).
At each step the robot uses the inferred position to decide its next move.
The accuracy metric is the MSE of both the position and the map.

\textsf{MTT} (\emph{Multi-Target Tracker}) adapted from~\cite{MurrayS18} is a model where there are a variable number of targets with linear-Gaussian motion models with a state space of 2D position and velocity, producing linear-Gaussian measurements of the position at each time step.
Targets randomly appear according to a Poisson process and each disappear with fixed probability at each step.
Measurements do not identify which target they came from, and ``clutter'' measurements that come not from targets but from some underlying distribution add to observations, complicating inference of which measurements are associated to which targets.
The accuracy metric is expected $\mathsf{MOTA*} = (1/\mathsf{MOTA}) - 1$ where~$\mathsf{MOTA} \in [0,1]$ is the \emph{Multiple Object Tracking Accuracy}~\cite{MOTA08}.

\paragraph{Experimental Setup.}
All the experiments were run on a server with $32$~CPUs ($2.60$~GHz) and $128$~GB memory.
We ran all the benchmarks for $500$ steps.
In all cases, the inference runs in bounded memory~(see
\ifextended
\Cref{sec:perf_memory}\else
Appendix~F.3\fi).

For each algorithm, we evaluated how much time it required to achieve 90\% of runs close to a loss target (out of $1000$ runs total):
\vspace{-0.5em}
$$
| \log (P_{90\%}(\mathit{loss}))  - \log (\mathit{loss}_{\mathit{target}}) | < 0.5.
$$

For each benchmark, the baseline is the median loss of SDS at 1000 particles as $\mathit{loss}_{\mathit{target}}$ for that benchmark.
We measured the number of particles required to achieve this loss, and then measured the total execution time at this particle count for $500$ steps (in \ifextended\Cref{sec:perf_accuracy}\else Appendix F.2\fi{} we also evaluate loss and step latency across a fixed range of particle counts).

\paragraph{Results.}
\Cref{fig:evaluation} shows the results. The height of each bar is the median total execution time, and the error bars are 90\% and 10\% quantiles, aggregated over 1000 runs. Each bar is labeled with the minimum number of particles required to achieve the accuracy threshold, accurate to 1.5 significant digits (100, 150, 200, 250, \dots).
We observe that SDS is able to compute an exact solution for \textsf{Beta-Bernoulli}, \textsf{Kalman-1D}, and \textsf{Robot}. In all these examples $1$ particle is already enough to reach the target accuracy.
Overall, the results show that the number of particles required to reach the desired accuracy with PF implies a significant slowdown compared to SDS.
Moreover, the \textsf{SLAM} and \textsf{MTT} benchmarks show that, in some cases, PF is not an option: the target accuracy was not reached with $15,000$ and $2,500$ particles, respectively, at which point PF was already $10$~times slower than SDS and we stopped the experiments.

As expected, BDS performance numbers are between those of PF and SDS.
At worst, when there is no possible intra-step symbolic computations (e.g., \textsf{Beta-Bernoulli}), BDS behaves like a particle filter and requires as many particles as PF.
At best, BDS performs as well as SDS (e.g., \textsf{Outlier}).

Additionally, \Cref{fig:evaluation} also shows that for a given number of particles, the overhead induced by managing the delayed sampling graph is significant.
Compared to BDS and SDS, depending on the benchmark, it is possible to use~$2$ to~$4$ times as many particles for PF with the same execution time.
However, this is not enough to match the gain in accuracy.

\paragraph{Alternative Baselines.}

The results presented in \Cref{fig:evaluation} do not quantify the speedup of SDS on the SLAM and MTT benchmarks because the other inference algorithms time out.
To evaluate speedups on these two benchmarks, we used PF as an alternative baseline instead of SDS.
\Cref{fig:evaluation-notimeout} presents the execution time of PF, BDS, and SDS to reach a loss \textit{close} to the median of PF with 2000 and 4000 particles.

We observe that SDS requires a much smaller number of particles to reach similar accuracy which translates into speedups ranging from $10^1$ (MTT-2000) to $10^4$ (SLAM-4000).
BDS requires either a similar or smaller numbers of particles. 
But the overhead introduced by the graph manipulations mostly translates in slowdowns compared to PF.

\pgfplotstableread{
x	y	y-min	y-max	p
SLAM-2000	389.403996	11.46191600000003	26.752634999999998	350
SLAM-4000	18756.56644	1206.8172259999992	1445.8157640000027	8,000
MTT-2000	10643.371589	656.5615460000008	2806.2434840000005	300
MTT-4000	13976.663295	743.723868000001	1332.2320170000003	400
}{\bds}

\pgfplotstableread{
x	y	y-min	y-max	p
SLAM-2000	0.557001	0.003850999999999938	0.041920999999999986	1
SLAM-4000	0.557001	0.003850999999999938	0.041920999999999986	1
MTT-2000	767.655314	128.46169999999995	381.4307110000001	20
MTT-4000	767.655314	128.46169999999995	381.4307110000001	20
}{\sds}

\pgfplotstableread{
x	y	y-min	y-max	p
SLAM-2000	90.748196	1.0917589999999961	2.019789000000003	300
SLAM-4000	21279.041365	1093.9340620000003	1441.1740959999988	15,000
MTT-2000	4158.015648	121.87468199999967	775.040731	350
MTT-4000	4978.815385	153.8039229999995	1365.9762330000003	400
}{\particles}

 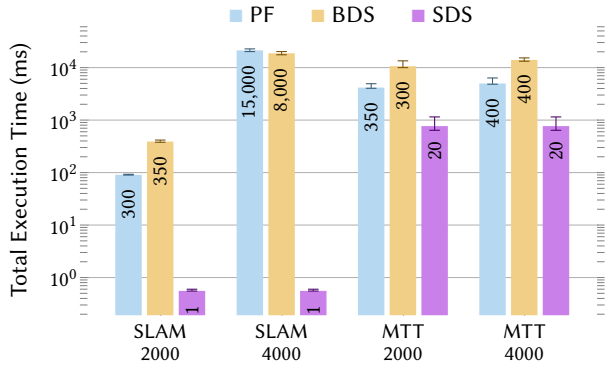
\begin{figure}
\centering
\begin{tikzpicture}
  \begin{axis} [
    ybar,
    ymode=log,
    width=0.47\textwidth,
    height=15.5em,
    log origin=infty,
    major grid style={draw=black!30},
    ymajorgrids, tick align=outside,
    axis x line*=bottom,
    y axis line style={opacity=0},
    ylabel = {\small\sffamily Total Execution Time (ms)},
x axis line style={opacity=0},
    enlarge x limits=0.20,
    x tick label style={major tick length=0pt},
    xticklabels={SLAM\\[-0.25em]2000, SLAM\\[-0.25em]4000, MTT\\[-0.25em]2000, MTT\\[-0.25em]4000},
    xticklabel style={align=center},
    xtick=data,
    tick label style={font=\footnotesize\sffamily},
    visualization depends on=y \as \rawy,
    every node near coord/.append style={
        font=\footnotesize,
        rotate=90,
        anchor=east,
        xshift=-0.2em,
    },
    nodes near coords,
    point meta=explicit symbolic
]

\addplot [draw=none, 
          fill=pfcolor!50,
          error bars/.cd, 
          error bar style={pfcolor!50!black}] 
  plot [error bars/.cd, y dir=both, y explicit]
  table [x expr=\coordindex-1, y error plus=y-max, y error minus=y-min, meta=p]{\particles};
    
\addplot [draw=none,
          fill=bdscolor!50,
          error bars/.cd, 
          error bar style={bdscolor!50!black}]  
  plot [error bars/.cd, y dir=both, y explicit]
  table [x expr=\coordindex-1, y error plus=y-max, y error minus=y-min, meta=p]{\bds};
  
\addplot [draw=none,
          fill=pmdscolor!50,
          error bars/.cd, 
          error bar style={pmdscolor!50!black}]
  plot [error bars/.cd, y dir=both, y explicit]
  table [x expr=\coordindex-1, y error plus=y-max, y error minus=y-min, meta=p]{\sds};

\end{axis}  
\node at (3.5, 4) {\small\sf 
  {\color{pfcolor!50}{$\blacksquare$}}~PF $\quad$
  {\color{bdscolor!50}{$\blacksquare$}}~BDS $\quad$
  {\color{pmdscolor!50}{$\blacksquare$}}~SDS}; 
\end{tikzpicture}
\caption{
Execution time comparison with two different baselines: median accuracy of PF with $2000$ and $4000$ particles, respectively.}
\label{fig:evaluation-notimeout}
\end{figure}

 \section{Related Work}

\paragraph{Probabilistic Programming.}
Over the last few years there has been a growing interest on probabilistic programming languages.
Some languages like BUGS~\cite{lunn2009bugs}, Stan~\cite{carpenter2017stan}, or Augur~\cite{HuangTM17} offer optimized inference technique for a constrained subset of models.
Other languages like WebPPL~\cite{goodman_stuhlmuller_2014}, Edward~\cite{TranHSB0B17}, Pyro~\cite{BinghamCJOPKSSH19}, or Birch~\cite{MurrayS18} focus on expressivity allowing the specification of arbitrary complex models.
Compared to these languages, \ProbZelus can be used to program \emph{reactive models} that typically do not terminate, and inference can be run in parallel with deterministic components that interact with an environment. 

\paragraph{Reactive Languages with Uncertainty.}
Lutin is a language for describing non-deterministic reactive
systems for testing and simulation~\cite{raymond-lutin-2008}, but while Lutin 
supports weighted sampling to describe constrained random scenarios, it does not support inference.
ProPL~\cite{Pfeffer05} is a language to describe probabilistic models for process that evolve over a period of time.
This language also extends a probabilistic language with a notion of processes that can be composed in parallel, but compared to \ProbZelus, ProPL focuses on a constrained class of models that can be interpreted as \emph{Dynamic Bayesian Networks}~(DBN), and relies on standard DBN inference techniques.
In the same vein, CTPPL~\cite{Pfeffer09} is a language to describe continuous-time processes where the amount of time taken by a sub-process can be specified by a probabilistic model.
These models cannot be expressed in \ProbZelus which relies on the synchronous model of computation.
It would be interesting to investigate how to extend \ProbZelus to continuous-time models based on \zelus' support for ordinary differential equations~(ODE)~\cite{lucy:hscc13}.

\paragraph{Inference.}
Researchers have proposed streaming inference algorithms, including
variational~\cite{BroderickBWWJ13}, or
sampling-based~\cite{particlefilter,rbpf} approaches.
Popular languages like Stan, Edward, or Pyro, offer support to stream data through the model during inference to handle large datasets.
However, compared to \ProbZelus, the model must be defined a priori and does not evolve during the inference.

The Anglican and Birch probabilistic programming
languages support delayed sampling~\cite{murray18delayed_sampling}. 
These languages do not support streaming inference or reactive programming. Again, their interfaces only support inference on a complete probabilistic model.

 \section{Conclusion}

Modeling uncertainty is a primary element of control systems for tasks that operate under the assumption of a probabilistic model of their environment (e.g., object tracking). 
While synchronous languages have developed as a prominent way to develop control applications,
to date there has been limited work in these languages on programming language support for
modeling uncertainty.

In this paper we present \ProbZelus, the first synchronous probabilistic programming language that lifts 
emerging abstractions for probabilistic programming into the reactive setting thus enabling \emph{inference-in-the-loop}. 
Moreover, our streaming delayed sampling algorithm provides efficient semi-symbolic inference while
still satisfying a key requirement of control applications in that they must execute
with bounded resources.

Our results demonstrate that \ProbZelus enables us to write, in the very same source, a deterministic model for the control software and a probabilistic model for its behavior and environment with complex interactions between the two.

\ifextended
\else
\balance
\fi

\ifextended
\clearpage
\appendix

\section{\ProbZelus}
\label{app:probzelus}

In this section, we provide the complete definitions of the \ProbZelus type system and semantics for the kernel language introduced \Cref{sec:language} extended with the probabilistic operator \zl{factor($e$)} which is equivalent to \zl{observe(exp(1), $e$)}.
Intuitively, \zl{factor} can directly update the weight of the execution path with the value of an expression~$e$.

\subsection{Typing}
\label{app:typing}

\begin{figure*}
  \begin{center}
  \begin{small}
  \(
  \begin{array}{c}
\infrule
    {\type{G}{$e$}{t}{\typeD}}
    {\type{G}{$e$}{t}{\typeP}}
    \qquad
    \infrule{\typeof{c} = t}{\type{G}{$c$}{t}{\typeD}}
    \qquad
    \infrule{G(x) = t}{\type{G}{$x$}{t}{\typeD}}
    \qquad
    \infrule
    {\type{G}{$e_1$}{t_1}{k} \quad \type{G}{$e_2$}{t_2}{k}}
    {\type{G}{($e_1$,\ $e_2$)}{\typetimes{t_1}{t_2}}{k}}
    \qquad
    \infrule
    {\typeof{\mathit{op}} = \typefun{t_1}{\typeD}{t_2} \quad \type{G}{$e$}{t_1}{k}}
    {\type{G}{$\mathit{op}$($e$)}{t_2}{k}}
    \\[1em]
    \infrule
    {G(f) = \typefun{t_1}{k}{t_2} \quad \type{G}{$e$}{t_1}{\typeD}}
    {\type{G}{$f$($e$)}{t_2}{k}}
    \qquad
    \infrule{G(x) = t}{\type{G}{last\ $x$}{t}{\typeD}}
    \qquad
    \infrule
    {\type{G}{$E$}{G'}{k} \quad \type{G + G'}{$e$}{t}{k}}
    {\type{G}{$e$ where rec\ $E$}{t}{k}}
    \\[1em]
    \infrule
    {\type{G}{$e$}{\typebool}{k} \quad \type{G}{$e_1$}{t}{k} \quad \type{G}{$e_2$}{t}{k}}
    {\type{G}{present\ $e$ ->\ $e_1$ else\ $e_2$}{t}{k}}
    \qquad
    \infrule
    {\type{G}{$e_1$}{t}{k} \quad \type{G}{$e_2$}{\typebool}{k}}
    {\type{G}{reset\ $e_1$ every\ $e_2$}{t}{k}}
    \\[1em]
    \infrule
    {\type{G}{$e$}{\typesampler{t}}{\typeD}}
    {\type{G}{sample($e$)}{t}{\typeP}}
    \qquad
    \infrule
    {\type{G}{$e_1$}{\typedist{t}}{\typeD} \quad \type{G}{$e_2$}{t}{\typeD}}
    {\type{G}{observe($e_1$,\ $e_2$)}{\typeunit}{\typeP}}
    \qquad
    \infrule
    {\type{G}{$e$}{\typefloat}{\typeD}}
    {\type{G}{factor($e$)}{\typeunit}{\typeP}}
    \\[1em]
    \infrule
    {\type{G}{$e$}{t}{\typeP}}
    {\type{G}{infer($e$)}{\typesampler{t}}{\typeD}}
    \qquad
    \infrule
    {\type{G}{$e$}{\typedist{T}}{\typeD}}
    {\type{G}{$e$}{\typesampler{T}}{\typeD}}
    \\[1em]

\infrule
    {\type{G}{$e$}{t}{k}}
    {\type{G}{$x$ =\ $e$}{[t/x]}{k}}
    \qquad
    \infrule
    {\type{G}{$e$}{t}{k}}
    {\type{G}{init\ $x$ =\ $e$}{[t/x]}{k}}
    \qquad
    \infrule
    {\type{G + G_1 + G_2}{$E_1$}{G_1}{k} \quad \type{G + G_1 + G_2}{$E_2$}{G_2}{k}}
    {\type{G}{$E_1$ and\ $E_2$}{G_1 + G_2}{k}}
    \\[1em]
  
\infrule
    {\type{G+[t_1/x]}{$e$}{t_2}{\typeD}}
    {\type{G}{let node\ $f$\ $x$ =\ $e$}{G + [\typefun{t_1}{\typeD}{t_2}/f]}{\typeD}}
    \qquad
    \infrule
    {\type{G+[t_1/x]}{$e$}{t_2}{\typeP}}
    {\type{G}{let proba\ $f$\ $x$ =\ $e$}{G + [\typefun{t_1}{\typeP}{t_2}/f]}{\typeD}}
    \qquad
    \infrule
    {\type{G}{$d_1$}{G_1}{\typeD} \quad \type{G_1}{$d_2$}{G_2}{\typeD}}
    {\type{G}{$d_1$\ $d_2$}{G_2}{\typeD}}
  \end{array}
  \)
\end{small}
\end{center}
\caption{Typing with deterministic and probabilistic kinds.}
  \label{fig:typing}
\end{figure*}

The type system that discriminates deterministic from probabilistic expressions is defined \Cref{fig:typing}. To simplify the presentation, we ignored datatypes polymorphism.

The sub-typing rule indicates that any deterministic expression can be lifted into a probabilistic one.
Expressions like constants, variables, and \zl{last} are deterministic.
The kind of classic \zelus expressions (pairs, $\textit{op}$, local definitions, \zl{present}, and \zl{reset}) is the kind of their body.
Similarly, the kind of equations is the kind of their defining expression, and parallel composition imposes the same kind for all the equations.
Note that it is always possible to compose deterministic and probabilistic computations. For rules where all sub-expressions share the same kind~$k$ we enforce the use of the sub-typing rule to lift deterministic expressions.

The expressions \zl{sample}, \zl{factor}, and \zl{observe} are probabilistic.
The transition from probabilistic to deterministic is realized via \zl{infer}: a deterministic expression whose body is always probabilistic.
Probabilistic expressions can thus only occur under an~\zl{infer}.

\paragraph{Other Static Analyses.}
The \zelus compiler statically checks initialization, and causality of the program~\cite{zelus-manual}.
These two analyses guarantee that there exists a schedule of parallel equations that makes the streams productive.
Extending these analyses to the probabilistic operators is straightforward: probabilistic operators can be treated as external operators.

\subsection{Co-iterative Semantics}
\label{app:semantics}

\begin{figure*}
  \vspace{4em}
  $
  \begin{small}
  \begin{array}{lcl}
  \sem{c}^i_\gamma &=& ()\\
  \sem{c}^s_\gamma &=& \fun{s} (c,s)
  \\[0.5em]
  \sem{x}^i_\gamma &=& ()\\
  \sem{x}^s_\gamma &=& \fun{s} (\gamma(x),s)
  \\[0.5em]
  \sem{\zlm!last\ $x$!}^i_\gamma &=& ()\\
  \sem{\zlm!last\ $x$!}^s_\gamma &=& \fun{s} (\gamma(\zlm!$x$_last!),s)
  \\[0.5em]
  \sem{\zlm!($e_1$, $e_2$)!}^i_\gamma &=&
    (\sem{e_1}^i_\gamma, \sem{e_2}^i_\gamma)\\
  \sem{\zlm!($e_1$, $e_2$)!}^s_\gamma &=&
    \fun{(s_1, s_2)}
      \begin{array}[t]{@{}l@{}}
        \letin{v_1,s_1' = \sem{e_1}^s_\gamma(s_1)}\\
        \letin{v_2,s_2' = \sem{e_2}^s_\gamma(s_2)}
        (\zlm!($v_1$,$v_2$)!,(s_1',s_2'))\\
      \end{array}
  \\[0.5em]
  \sem{\zlm!$\mathit{op}$($e$)!}^i_\gamma &=&
    \sem{e}^i_\gamma\\
  \sem{\zlm!$\mathit{op}$($e$)!}^s_\gamma &=&
    \fun{s}
      \letin{v,s' = \sem{e}^s_\gamma(s)}
        (\zlm!$\mathit{op}$($v$)!, s')
  \\[0.5em]
  \sem{\zlm!$f$($e$)!}^i_\gamma & = &
    (\sem{e}^i_\gamma, \gamma(f\zlm!_init!))
  \\
  \sem{\zlm!$f$($e$)!}^s_\gamma & = &
    \fun{(s_1, s_2)}
    \begin{array}[t]{@{}l@{}}
      \letin{v_1,s_1' = \sem{e}^s_\gamma(s_1)}\\
      \letin{v_2,s_2' = \gamma(f\zlm!_step!)(v_1)(s_2)}
      (v_2,(s_1',s_2'))
    \end{array}
  \\[2.25em]
  \sem{\zlm!present\ $e$ ->\ $e_1$ else\ $e_2$!}^i_\gamma &=& 
    (\sem{e}^i_\gamma, \sem{e_1}^i_\gamma, \sem{e_2}^i_\gamma)\\
  \sem{\zlm!present\ $e$ ->\ $e_1$ else\ $e_2$!}^s_\gamma &=& 
    \begin{array}[t]{@{}l@{}}
    \fun{(s, s_1, s_2)} \letin{v,s' = \sem{e}^s_\gamma(s)}\\
      \quad \mathit{if} \; v \; \begin{array}[t]{@{}l@{}} 
      \mathit{then} \; 
        \letin{v_1,s_1' = \sem{e_1}^s_\gamma(s_1)}
        (v_1, (s', s_1', s_2))\\
      \mathit{else} \;\;
        \letin{v_2,s_2' = \sem{e_2}^s_\gamma(s_2)}
        (v_2, (s', s_1, s_2'))
      \end{array}
    \end{array}
  \\[4em]
  \sem{\zlm!reset\ $e_1$ every\ $e_2$!}^i_\gamma & = &
    (\sem{e_1}^i_\gamma, \sem{e_1}^i_\gamma, \sem{e_2}^i_\gamma)\\
  \sem{\zlm!reset\ $e_1$ every\ $e_2$!}^s_\gamma & = &
    \fun{(s_0, s_1, s_2)}
    \begin{array}[t]{@{}l@{}}
      \letin{v_2,s_2' = \sem{e_2}^s_\gamma(s_2)}\\
      \letin{v_1,s_1' = \sem{e_1}^s_\gamma(\mathit{if}\; v_2\; \mathit{then}\; s_0\; \mathit{else}\; s_1)}\\
      (\zlm!$v_1$!,(s_0,s_1',s_2'))\\
    \end{array}
  \\[0.5em]
  \left \llbracket
  \begin{minipage}[t][3.1em]{14em}
  \begin{lstlisting}[aboveskip=-3.5em, xleftmargin=0pt, xrightmargin=0pt]
  $e$ where 
    rec init $x_1$ = $c_1$ and ... 
    and init $x_k$ = $c_k$
    and $y_1$ = $e_1$ and ... 
    and $y_n$ = $e_n$        
  \end{lstlisting}
  \end{minipage}
  \right \rrbracket^i_\gamma
  &=&
  \left(
  \begin{array}{@{}c@{}}
  (c_1, \dots, c_k),\\ 
  (\sem{e_1}^i_\gamma, \dots, \sem{e_n}^i_\gamma),\\
  \sem{e}^i_\gamma
  \end{array}
  \right)\\
  \left \llbracket
  \begin{minipage}[t][3.1em]{14em}
  \begin{lstlisting}[aboveskip=-3.5em, xleftmargin=0pt, xrightmargin=0pt]
  $e$ where 
    rec init $x_1$ = $c_1$ and ... 
    and init $x_k$ = $c_k$
    and $y_1$ = $e_1$ and ... 
    and $y_n$ = $e_n$        
  \end{lstlisting}
  \end{minipage}
  \right \rrbracket^s_\gamma
  &=&
  \begin{array}{l}
  \fun{((m_1, \dots, m_k), (s_1, \dots, s_n), s)}\\ 
  \qquad \begin{array}[t]{@{}l@{}}
  \letin{\gamma_1 = \gamma[m_1/\zlm!$x_1$_last!]}\\
  \dots \\
  \letin{\gamma_k = \gamma_{k-1}[m_k/\zlm!$x_k$_last!]}\\
  \letin{v_1, s_1' = \sem{e_1}^s_{\gamma_k}(s_1)}
  \letin{\gamma'_1 = \gamma_k[v_1/y_1]}\\
  \dots\\
  \letin{v_n, s_n' = \sem{e_n}^s_{\gamma'_{n-1}}(s_n)}
  \letin{\gamma'_n = \gamma'_{n-1}[v_n/y_n]}\\
  \letin{v, s' = \sem{e}^s_{\gamma'_n}(s)}\\
  v, ((\gamma'_n[x_1], \dots, \gamma'_n[x_k]), (s_1', \dots, s_n'), s')
  \end{array}
  \end{array}
  \\[7.5em]

  \sem{\zlm!let\ node\ $f$\ x =\ $e$!}_\gamma &=& 
  \gamma[\sem{e}^i_{\gamma}/f\zlm!_init!, \quad \fun{v}\fun{s}\sem{e}^s_{\gamma[v/\zlmm!x!]}/f\zlm!_step!]
\\[0.5em]
\sem{\zlm!let\ proba\ $f$\ x =\ $e$!}_\gamma &=&
  \gamma[\psem{e}^i_{\gamma}/f\zlm!_init!, \quad \fun{v}\fun{s}\psem{e}^s_{\gamma[v/\zlmm!x!]}/f\zlm!_step!]
\\[0.5em]
  
  \sem{d_1 \; d_2}_\gamma &=&
    \letin{\gamma_1 = \sem{d_1}_\gamma} \sem{d_2}_{\gamma_1}
  \end{array}
  \end{small}
  $
  \caption{Co-iterative semantics of deterministic \ProbZelus programs. For local definitions each initialized variable is defined in a subsequent equation, i.e., $\{x_i\}_{1..k} \cap \{y_j\}_{1..n} = \{x_i\}_{1..k}$.}
  \label{fig:sem_deterministic_full}
  \vspace{2em}
\end{figure*}

\begin{figure*}
  $
  \begin{small}
  \begin{array}{lcl}
  \psem{e}^i_\gamma &=& \sem{e}^i_\gamma
    \qquad \mathit{if}\ \kindof{e} = \typeD \\
  \psem{e}^s_\gamma &=& 
    \fun{s} \fun{U} \delta_{\sem{e}^s_\gamma(s)}(U)
    \qquad \mathit{if}\ \kindof{e} = {\typeD}
  \\ &=& 
    \fun{s} \fun{U}
      \left\{\begin{array}{@{}l@{}}
        1 \; \mathit{if} \; \sem{e}^s_\gamma(s) \in U\\
        0 \; \mathit{otherwise}
      \end{array}
      \right.
  \\[1.5em]
  \psem{\zlm!($e_1$, $e_2$)!}^i_\gamma &=&
    (\psem{e_1}^i_\gamma, \psem{e_2}^i_\gamma)\\
  \psem{\zlm!($e_1$, $e_2$)!}^s_\gamma &=&
    \fun{(s_1, s_2)} \fun{U}
      \begin{array}[t]{@{}l@{}}
        \letin{\mu_1 = \psem{e_1}^s_\gamma(s_1)}\\
        \int\mu_1(dv_1,ds_1')\;
        \begin{array}[t]{@{}l@{}}
          \letin{\mu_2 = \psem{e_2}^s_\gamma(s_2)}\\
          \int\mu_2(dv_2,ds_2')\;
            \delta_{\zlmm!($v_1$,$v_2$)!, (s_1', s_2')}(U)
        \end{array}
      \end{array}
  \\[0.5em]
  \psem{\zlm!$\mathit{op}$($e$)!}^i_\gamma &=&
    \psem{e}^i_\gamma\\
  \psem{\zlm!$\mathit{op}$($e$)!}^s_\gamma &=&
    \fun{s} \fun{U}
      \letin{\mu = \psem{e}^s_\gamma(s)}
      \int\mu(dv,ds')\;
        \delta_{\zlmm!$\mathit{op}$($v$)!, s'}(U)
  \\[0.5em]
  \psem{\zlm!$f$($e$)!}^i_\gamma & = &
    (\sem{e}^i_\gamma, \gamma(f\zlm!_init!))
  \\
  \psem{\zlm!$f$($e$)!}^s_\gamma & = &
    \fun{(s_1, s_2)} \fun{U}
    \begin{array}[t]{@{}l@{}}
      \letin{v_1,s_1' = \sem{e}^s_\gamma(s_1)}\\
      \letin{\mu_2 = \gamma(f\zlm!_step!)(v_1)(s_1)}
      \int\mu_2(dv_2,ds_2')\;
        \delta_{v_2,(s_1',s_2')}(U)
    \end{array}
  \\[2.25em]
  \psem{\zlm!present\ $e$ ->\ $e_1$ else\ $e_2$!}^i_\gamma &=& 
    (\psem{e}^i_\gamma, \psem{e_1}^i_\gamma, \psem{e_2}^i_\gamma)\\
  \psem{\zlm!present\ $e$ ->\ $e_1$ else\ $e_2$!}^s_\gamma &=& 
    \begin{array}[t]{@{}l@{}}
    \fun{(s, s_1, s_2)} \fun{U} \\ \qquad
      \begin{array}[t]{@{}l@{}}
        \letin{\mu = \psem{e}^s_\gamma(s)}\\
        \int\mu(dv,ds')\;
        \begin{array}[t]{@{}l@{}}
          \mathit{if} \; v \\
          \mathit{then} \; 
            \letin{\mu_1 = \psem{e_1}^s_\gamma(s_1)}
            \int\mu_1(dv_1, ds_1')\;
            \delta_{v_1, (s', s_1', s_2)}(U)\\
          \mathit{else} \;\;
            \letin{\mu_2 = \sem{e_2}^s_\gamma(s_2)}
            \int\mu_2(dv_2, ds_2')\;
            \delta_{v_2, (s', s_1, s_2')}(U)
        \end{array}
      \end{array}
    \end{array}
  \\[6em]
  \psem{\zlm!reset\ $e_1$ every\ $e_2$!}^i_\gamma & = &
    (\psem{e_1}^i_\gamma, \psem{e_1}^i_\gamma, \psem{e_2}^i_\gamma)\\
  \psem{\zlm!reset\ $e_1$ every\ $e_2$!}^s_\gamma & = &
    \fun{(s_0, s_1, s_2)} \fun{U}
    \begin{array}[t]{@{}l@{}}
      \letin{\mu_2 = \sem{e_2}^s_\gamma(s_2)}\\
      \int\mu(dv_2,ds_2')\;
      \begin{array}[t]{@{}l@{}}
        \letin{\mu_1 = \sem{e_1}^s_\gamma(\mathit{if}\; v_2\; \mathit{then}\; s_0\; \mathit{else}\; s_1)}\\
        \int\mu(dv_1,ds_1')\;
        \begin{array}[t]{@{}l@{}}
          \delta_{\zlmm!$v_1$!,(s_0,s_1',s_2')}(U)
        \end{array}
      \end{array}
    \end{array}
  \\[0.5em]
  \left \{ \mkern-6.5mu \left [
  \begin{minipage}[t][3.1em]{14em}
  \begin{lstlisting}[aboveskip=-3.5em, xleftmargin=0pt, xrightmargin=0pt]
  $e$ where 
    rec init $x_1$ = $c_1$ and ... 
    and init $x_k$ = $c_k$
    and $y_1$ = $e_1$ and ... 
    and $y_n$ = $e_n$        
  \end{lstlisting}
  \end{minipage}
  \right ] \mkern-6.5mu \right \}^i_\gamma
  &=&
  \left(
  \begin{array}{@{}c@{}}
  (c_1, \dots, c_k),\\ 
  (\psem{e_1}^i_\gamma, \dots, \psem{e_n}^i_\gamma),\\
  \psem{e}^i_\gamma
  \end{array}
  \right)\\
  \left \{ \mkern-6.5mu \left [
  \begin{minipage}[t][3.1em]{14em}
  \begin{lstlisting}[aboveskip=-3.5em, xleftmargin=0pt,   xrightmargin=0pt]
  $e$ where 
    rec init $x_1$ = $c_1$ and ... 
    and init $x_k$ = $c_k$
    and $y_1$ = $e_1$ and ... 
    and $y_n$ = $e_n$        
  \end{lstlisting}
  \end{minipage}
  \right ] \mkern-6.5mu \right \}^s_\gamma
  &=&
  \begin{array}{@{}l}
  \fun{((m_1, \dots, m_k), (s_1, \dots, s_n), s)} \fun{U}\\ 
  \quad \begin{array}[t]{@{}l@{}}
  \letin{\gamma_1 = \gamma[m_1/\zlm!$x_1$_last!]}
  \dots
  \letin{\gamma_k = \gamma_{k-1}[m_k/\zlm!$x_k$_last!]}\\
  \letin{\mu_1 = \psem{e_1}^s_{\gamma_k}(s_1)}\\
  \int \mu_1(dv_1, ds_1') \; \begin{array}[t]{@{}l@{}}\letin{\gamma'_1   = \gamma_k[v_1/y_1]}\\
\int \dots\\
  \qquad\begin{array}[t]{@{}l@{}}
  \letin{\mu_n = \sem{e_n}^s_{\gamma'_{n-1}}(s_n)}\\
  \int \mu_n(dv_n, ds'_n) \; \begin{array}[t]{@{}l@{}} \letin{\gamma'_n   = \gamma'_{n-1}[v_n/y_n]}\\
\letin{\mu = \psem{e}^s_{\gamma'_n}(s)}\\
  \int \mu(dv, ds') \; \delta_{v, ((\gamma'_n[x_1], \dots, \gamma'_n  [x_k]), (s_1', \dots, s_n'), s')}(U)
  \end{array}
  \end{array}
  \end{array}
  \end{array}
  \end{array}\\[7.5em]
  \psem{\zlm!sample($e$)!}^i_\gamma &=& \sem{e}^i_\gamma\\
  \psem{\zlm!sample($e$)!}^s_\gamma &=& \fun{s}\fun{U}
    \begin{array}[t]{@{}l@{}}
    \letin{\mu, s' = \sem{e}^s_\gamma(s)}
    \int_T \mu(dv) \; \delta_{v, s'}(U)
    \end{array}
  \\[0.5em]
  \psem{\zlm!factor($e$)!}^i_\gamma &=& \sem{e}^i_\gamma\\
  \psem{\zlm!factor($e$)!}^s_\gamma &=& \fun{s}\fun{U}
    \begin{array}[t]{@{}l@{}}
      \letin{v, s' = \sem{e}^s_\gamma(s)}
      \exp(v) \;  \delta_{\mathtt{()}, s'}(U)
    \end{array}
  \\[0.5em]
  \psem{\zlm!observe($e_1$,\ $e_2$)!}^i_\gamma &=&
    \zlm!($\sem{e_1}^i_\gamma$,$\sem{e_2}^i_\gamma$)!
  \\
  \psem{\zlm!observe($e_1$,\ $e_2$)!}^s_\gamma &=&
    \fun{(s_1, s_2)}\fun{U} 
    \begin{array}[t]{@{}l@{}} 
      \letin{\mu,s_1' = \sem{e_1}^s_\gamma(s_1)}\\
      \letin{v,s_2' = \sem{e_2}^s_\gamma(s_2)}\\
      \mu_{\rm{pdf}}(v) * \delta_{\mathtt{()}, (s_1', s_2)}(U)
    \end{array}
\end{array}
\end{small}
  $
  \caption{Co-iterative semantics of probabilistic \ProbZelus expressions~(\textit{i.e.}, $\kindof{e} = {\typeP}$). For local definitions each initialized variable is defined in a subsequent equation, i.e., $\{x_i\}_{1..k} \cap \{y_j\}_{1..n} = \{x_i\}_{1..k}$.}
  \label{fig:sem-probabilistic-full}
\end{figure*}

The co-iterative semantics of \ProbZelus's  deterministic processes is inspired by~\cite{pouzet-cmcs98} and defined \Cref{fig:sem_deterministic_full}.

A \emph{node} is a stream function of type $\typefun{T}{\typeD}{T'}$.
In addition to the state, the transition function thus takes an additional input of type~$T$ and returns a pair (result, next state)
\begin{center}
$
\mathit{CoNode}(T, T', S) = S \times (S \rightarrow T \rightarrow T' \times S).
$
\end{center}

The transition function of a variable always returns the corresponding value stored in the environment~$\gamma$.
The semantics of \zl{last $x$} is a simple access to a special variable \zl{$x$_last}.
\zl{present $e$ -> $e_1$ else $e_2$} introduced in \Cref{sec:example} returns the value of~$e_1$ when~$e$ is true and the value of~$e_2$ otherwise.
The state $(s, s_1, s_2)$ stores the state of the three sub-expressions.
The transition function lazily executes $e_1$ or $e_2$ depending on the value of~$e$ and returns the updated state.

The state of a set of scheduled locally recursive definitions \zl{$e$ where rec $E$} comprises three parts: the value of the local variables at the previous step which can be accessed via the \zl{last} operator, the state of the defining expressions, and the state of expression~$e$.
The initialization stores the initial values introduced by \zl{init} and the initial states of all sub-expressions.
The transition function incrementally builds the local environment defined by~$E$. 
First the environment is populated with a set of fresh variables \zl{$x_i$_last} initialized with the values stored in the state that can then be accessed via the \zl{last} operator.  
Then the environment is extended with the definition of all the variables~$y_i$ by executing all the defining expressions (where $\{x_i\}_{1..k} \cap \{y_j\}_{1..n} = \{x_i\}_{1..k}$).
Finally, the expression~$e$ is executed in the final environment.
The updated state contains the value of the initialized variables defined in~$E$ that will the be used to start the next step, and the updated state of the sub-expressions.

\paragraph{Probabilistic Extensions.} 
The semantics of the probabilistic part of \ProbZelus, defined \Cref{fig:sem-probabilistic-full}, follows the same structure as the deterministic semantics but defines measures over all possible executions as in~\cite{staton17}.
In particular a succession of computation is interpreted as sequentially integrating over the results of the preceding computations.

As for deterministic nodes, the transition function of a probabilistic node of type ${\typefun{T}{\typeP}{T'}}$ takes an additional argument and returns a measure over pairs (result, next state).
\begin{center}
\begin{small}
$
\mathit{CoPNode}(T, T', S) =  S \times (S \rightarrow T \rightarrow (\Sigma_{T' \times S} \rightarrow [0, \infty]))
$
\end{small}
\end{center}

\subsection{Alternative semantics}
\label{app:sem-alt}

We could give different semantics to \ProbZelus. For example, consider the following probabilistic node.
\begin{lstlisting}
let proba kahn_vs_scott () = p where
  rec init p = sample(beta(1, 1))
  and () = observe(bernoulli(p), true)
\end{lstlisting}

With the semantics defined~\Cref{sec:language}, this program produces the stream of distribution: 
$\mathit{Beta}(2,1), \mathit{Beta}(3,1), \dots$
Note that, even though, \zl{p} is defined as a constant, its distribution evolves at each steps.

Since the \zl{observe} statement uses the constant \zl{true}, we know that \zl{p} is necessarily $1$. An alternative semantics could thus returns the constant stream of distributions: $\delta_{\zlmm{1}}$.

\section{The \muF language}

Similarly to \ProbZelus, we extend \muF with the probabilistic operator \zl{factor}.
We now present the complete type system and semantics for \muF.

\subsection{Typing}
\label{app:muF-typing}
\begin{figure*}
\begin{center}
\begin{small}
\(
\begin{array}{c}
\infrule
  {\type{G}{$e$}{t}{\typeD}}
  {\type{G}{$e$}{t}{\typeP}}
  \qquad
  \infrule{\typeof{c} = t}{\type{G}{$c$}{t}{\typeD}}
  \qquad
  \infrule{G(x) = t}{\type{G}{$x$}{t}{\typeD}}
  \qquad
  \infrule
  {\type{G}{$e_1$}{t_1}{\typeD} \quad \type{G}{$e_2$}{t_2}{\typeD}}
  {\type{G}{($e_1$,\ $e_2$)}{\typetimes{t_1}{t_2}}{\typeD}}
  \qquad
  \infrule
  {\typeof{\mathit{op}} = \typefun{t_1}{\typeD}{t_2} \quad \type{G}{$e$}{t_1}{\typeD}}
  {\type{G}{$\mathit{op}$($e$)}{t_2}{\typeD}}
  \\[1em]
  \infrule
  {G(f) = \typefun{t_1}{k}{t_2} \quad \type{G}{$e$}{t_1}{\typeD}}
  {\type{G}{$f$($e$)}{t_2}{k}}
  \qquad
  \infrule
  {\type{G+[t_1/x]}{$e_1$}{t_2}{k} \quad \type{G}{$e_2$}{t_1}{\typeD}}
  {\type{G}{(fun\ $x$ ->\ $e_1$)($e_2$)}{t_2}{k}}
  \qquad
  \infrule
  {\type{G}{$e$}{\typebool}{\typeD} \quad \type{G}{$e_1$}{t}{k} \quad \type{G}{$e_2$}{t}{k}}
  {\type{G}{if\ $e$ then\ $e_1$ else\ $e_2$}{t}{k}}
  \\[1em]
  \infrule
  {\type{G}{$e_1$}{t_1}{k} \quad \type{G+[t_1/x]}{$e_2$}{t_2}{k}}
  {\type{G}{let\ $x$ =\ $e_1$ in\ $e_1$}{t_2}{k}}
  \qquad
  \infrule
  {\type{G+[t_1/x]}{$e$}{t_2}{k}}
  {\type{G}{fun\ $x$ ->\ $e$}{\typefun{t_1}{k}{t_2}}{\typeD}}
  \qquad
  \infrule
  {\type{G}{$e$}{\typesampler{t}}{\typeD}}
  {\type{G}{sample($e$)}{t}{\typeP}}
  \\[1em]
  \infrule
  {\type{G}{$e_1$}{\typedist{t}}{\typeD} \quad \type{G}{$e_2$}{t}{\typeD}}
  {\type{G}{observe($e_1$,\ $e_2$)}{\typeunit}{\typeP}}
  \qquad
  \infrule
  {\type{G}{$e$}{\typefloat}{\typeD}}
  {\type{G}{factor($e$)}{\typeunit}{\typeP}}
  \qquad
  \infrule
  {\type{G}{$e_1$}{\typetimes{t}{t_{\mathit{state}}}}{\typeP} \quad \type{G}{$e_2$}{\typesampler{t_{\mathit{state}}}}{\typeD}}
  {\type{G}{infer((fun\ $x$ ->\ $e_1$), $e_2$)}{\typesampler{t}}{\typeD}}
  \qquad
  \infrule
  {\type{G}{$e$}{\typedist{T}}{\typeD}}
  {\type{G}{$e$}{\typesampler{T}}{\typeD}}
  \\[1em]

\infrule
  {\type{G}{$e$}{t}{\typeD}}
  {\type{G}{let\ $f$\ =\ $e$}{G + [t/f]}{\typeD}}
  \qquad
  \infrule
  {\type{G}{$d_1$}{G_1}{\typeD} \quad \type{G_1}{$d_2$}{G_2}{\typeD}}
  {\type{G}{$d_1$\ $d_2$}{G_2}{\typeD}}
\end{array}
\)
\end{small}
\end{center}
\caption{Typing of \muF with deterministic and probabilistic kinds.}
\label{fig:typing-muF}
\end{figure*}

The type system defined \Cref{fig:typing-muF} is similar to the one \Cref{fig:typing} to distinguish deterministic from probabilistic expressions, but with additional restrictions since the compiled code is in a more constrained form.
Whenever possible we require sub-expressions to be deterministic, that is, in pairs, operator applications (including \zl{sample}, \zl{factor}, and \zl{observe}), function calls, and the condition of a \zl{if/then/else}.
These restrictions simplify the presentation of the semantics but do not reduce the expressiveness of the language since it is always possible to introduce additional local definitions to name intermediate probabilistic expressions.
For example \zl!if sample(bernoulli(0.5)) then ...!
can be rewritten \zl!let b = sample(bernoulli(0.5)) in if b then ...!

\subsection{Semantics of \muF}
\label{app:sem-muF}

The semantics of \muF follows~\cite{staton17}.
In a deterministic context $\kindof{e} = {\typeD}$, the semantics~$\sem{e}_\gamma$ of an expression is the classic interpretation of a strict functional language.
In a probabilistic context~($\kindof{e} = {\typeP}$), we define a the measure-based semantics~$\psem{e}_\gamma$.

The probabilistic semantics of \muF is presented in \Cref{fig:muF-sem-full}.
A deterministic expression is lifted to a probabilistic expression using the the Dirac delta measure applied to the value of the expression computed by the deterministic semantics.
As in \Cref{sec:co-iteration}, a local definition \zl{let $x$ = $e_1$ in $e_2$} is interpreted as integrating~$e_2$ over the measure defined by~$e_1$.
The semantics of the probabilistic operators is the following:
\zl{sample(e)} returns the distribution~$\sem{e}_\gamma$.
\zl{factor(e)} returns a measure defined on the singleton space \zl{()} whose value is $\exp(\sem{e}_\gamma)$.
\zl{observe($e_1$, $e_2$)} is similar but the score is the density function of the distribution~$\sem{e_1}_\gamma$ applied to~${\sem{e_2}_\gamma}$.

\paragraph{Inference.}
\zl{infer} handles the transition function generated by the compilation of \Cref{sec:compilation}.
The first argument of \zl{infer} is a transition function, and the second argument a distribution over state~$\sigma$.
The inference first integrates over the distribution~$\sigma$ and then normalizes the result~$\mu$ to produce a distribution~$\nu$ of pairs (result, next state). 
The special value~$\top$ denotes the entire space (value, state). 
This distribution is then decomposed into a pair of distributions using the pushforward of~$\mu$.

\begin{figure}
  $$
  \begin{array}{@{}lcl}
  \psem{\zl!let\ $f$ =\ $e$!}_\gamma &=& \gamma[\psem{e}_\gamma/f]\\[0.5em]
  
  \psem{d_1 \; d_2}_\gamma &=&
    \letin{\gamma_1 = \psem{d_1}_\gamma} \psem{d_2}_{\gamma_1}\\[1em]
  
  \psem{e}_\gamma &=&
    \fun{U} \delta_{\sem{e}_\gamma}(U)\; \mathit{if} \; \kindof{e} = {\typeD}\\[0.5em]
  
  \psem{e_1(e_2)}_\gamma &=& \fun{U} (\sem{e_1}_\gamma(\sem{e_2}_\gamma))(U)
  \\[0.5em]
  
  \psem{\zl!let\ $p$ =\ $e_1$ in\ $e_2$!}_\gamma &=&
\fun{U} \int_T \psem{e_1}_{\gamma}(du) \psem{e_2}_{\gamma+[u/p]}(U)\\[0.5em]
  
  \multicolumn{3}{@{}l}{
  \psem{\zl!if\ $e$ then\ $e_1$ else\ $e_2$!}_\gamma ~=~
  } \\
  \multicolumn{3}{l}{
   \quad
    \fun{U}
    \mathit{if} \; \sem{e}_\gamma
    \; \mathit{then} \; \psem{e_1}_\gamma(U)
    \; \mathit{else} \;  \psem{e_2}_\gamma(U)
  }
  \\[0.5em]
  
  \psem{\zl!fun\ $p$ ->\ $e$!}_\gamma &=& \fun{v} \psem{e}_{[v/p]}
  \\[0.5em]
  
  \psem{\zl!sample($e$)!}_\gamma &=& \fun{U} \sem{e}_\gamma(U)
  \\[0.5em]
  
  \multicolumn{3}{@{}l}{
  \psem{\zl!observe($e_1$,\ $e_2$)!}_\gamma ~=~
   } \\
   \multicolumn{3}{l}{
    \quad
    \fun{U}
    \letin{\mu = \sem{e_1}_\gamma} \mu_{\textrm{pdf}}(\sem{e_2}_\gamma) * \delta_{\mathtt{()}}(U)
   }
  \\[0.5em]
  
  \psem{\zl!factor($e$)!}_\gamma &=&
   \fun{U} \exp(\sem{e}_\gamma) * \delta_{\mathtt{()}}(U)
  \\[0.5em]
  
  \multicolumn{3}{@{}l}{
  \;\sem{\zl!infer(fun\ $x$ ->\ $e_1$,\ $e_2$)!}_\gamma ~=~
   } \\[0.25em]
   \multicolumn{3}{l}{
  \quad
  \begin{array}[t]{@{}l@{}}
    \letin{\sigma = \sem{e_2}_\gamma}\\
    \letin{\mu = \fun{U} \int_{S} \sigma(ds) \psem{e}^s_{\gamma + [s/x]}(U)}\\
    \letin{\nu = \fun{U} \mu(U) / \mu(\top)}\\
    (\pi_{1*}(\nu), \pi_{2*}(\nu))
\end{array}
   }
  \end{array}
  $$
  \caption{Probabilistic semantics of \muF. The semantics is defined only for probabilistic expressions ($\kindof{e} = {\typeP}$).}
  \label{fig:muF-sem-full}
  \end{figure}

\section{Compilation}
\label{app:compilation}

\begin{figure}
\begin{small}
\hspace*{-3.9em}
\begin{minipage}[t]{0.25\textwidth}
\begin{lstlisting}
$\alloc$$($$c$$)$ = ()

$\alloc$$($$x$$)$ = ()

$\alloc$$($last $x$$)$ = ()

$\alloc$$($($e_1$,$e_2$)$)$ = ($\alloc$$($$e_1$$)$,$\alloc$$($$e_2$$)$)

$\alloc$$($$e$ where
  rec init $x_1$ = $c_1$ ...
  and init $x_k$ = $c_k$
  and $y_1$ = $e_1$ ...
  and $y_n$ = $e_n$$)$ $=$
    (($c_1$,..., $c_k$),
     ($\alloc$$($$e_1$$)$,..., $\alloc$$($$e_n$$)$),
     $\alloc$$($$e$$)$)

$\alloc$$($present $e$ -> $e_1$ else $e_2$$)$ = ($\alloc$$($$e$$)$,$\alloc$$($$e_1$$)$,$\alloc$$($$e_2$$)$)
\end{lstlisting}
\end{minipage}
\begin{minipage}[t]{0.12\textwidth}
\begin{lstlisting}
$\alloc$$($reset $e_1$ every $e2$$)$ =
  ($\alloc$$($$e_1$$)$,$\alloc$$($$e_1$$)$,$\alloc$$($$e_2$$)$)
  
$\alloc$$($$\mathit{op}$($e$)$)$ = $\alloc$$($$e$$)$

$\alloc$$($$f$($e$)$)$ = ($f$_init, $\alloc$$($$e$$)$)  

$\alloc$$($sample($e$)$)$ = $\alloc$$($$e$$)$

$\alloc$$($factor($e$)$)$ = $\alloc$$($$e$$)$

$\alloc$$($observe($e_1$, $e_2$)$)$ = 
  ($\alloc$$($$e_1$$)$, $\alloc$$($$e_2$$)$)

$\alloc$$($infer($e$)$)$ = ($\alloc$$($$e$$)$)
\end{lstlisting}
\end{minipage}
\end{small}
\caption{Memory allocation, i.e., initialization for the \muF step functions.}
\label{fig:allocation-full}
\end{figure}

\begin{figure*}
\hspace*{-4em}
\begin{small}
\begin{minipage}[t]{0.3\textwidth}
\begin{lstlisting}[xleftmargin=0pt]
$\compile$$($let node $f$ $x$ = $e$$)$ =
  let $f$_init = $\alloc$$($$e$$)$
  let $f$_step =
    fun (s,$x$) -> $\compile$$($$e$$)$(s)

$\compile$$($$d1$ $d2$$)$ $=$ $\compile$$($$d_1$$)$ $\compile$$($$d_2$$)$
$\compile$$($$c$$)$ $=$ fun s -> ($c$, s)
$\compile$$($$x$$)$ $=$ fun s -> ($x$, s)
$\compile$$($last $x$$)$ $=$ fun s -> ($x$_last, s)

$\compile$$($($e_1$, $e_2$)$)$ $=$ fun (s1,s2) ->
  let v1,s1' = $\compile$$($$e_1$$)$(s1) in
  let v2,s2' = $\compile$$($$e_2$$)$(s2) in
  ((v1,v2), (s1',s2'))

$\compile$$($$\mathit{op}$($e$)$)$ $=$ fun s ->
  let v,s' = $\compile$$($$e$$)$(s) in
  ($\mathit{op}$(v), s') 
  
$\compile$$($$f$($e$)$)$ $=$ fun (s1,s2) ->
  let v1,s1' = $\compile$$($$e$$)$(s1) in
  let v2,s2' = $f$_step(s2,v) in
  (v2, (s1',s2')) 
\end{lstlisting}
\end{minipage}
\begin{minipage}[t]{0.35\textwidth}
\begin{lstlisting}
$\compile$$($$e$ where 
  rec init $x_1$ = $c_1$ ... and init $x_k$ = $c_k$
  and $y_1$ = $e_1$ ... and $y_n$ = $e_n$$)$ $=$
fun ((m1,...,mk),(s1, ...,sn),s) ->
  let $x_1$_last = m1 in ... 
  let $x_k$_last = mk in
  let $y_1$, s1' = $\compile$$($$e_1$$)$(s1) in
  let $y_n$, sn' = $\compile$$($$e_n$$)$(sn) in
  let v,s' = $\compile$$($$e$$)$(s) in 
  (v, (s1', ..., sn'), s') 
  
$\compile$$($present $e$ -> $e_1$ else $e_2$$)$ $=$ 
fun (s,s1,s2) ->
  let v, s' = $\compile$$($$e$$)$(s) in
  if v then let v1,s1' = $\compile$$($$e_1$$)$(s1) in
    (v1, (s',s1',s2))
  else let v2,s2' = $\compile$$($$e_2$$)$(s2) in
    (v2, (s',s1,s2'))

$\compile$$($reset $e_1$ every $e_2$$)$ $=$
fun (s0,s1,s2) ->
  let v2,s2' = $\compile$$($$e_2$$)$(s2) in
  let s = if v2 then s0 else s1 in
  let v1,s1' = $\compile$$($$e_1$$)$(s) in
  (v1, (s0,s1',s2'))
\end{lstlisting}
\end{minipage}
\begin{minipage}[t]{0.25\textwidth}
\begin{lstlisting}
$\compile$$($sample($e$)$)$ $=$ fun s ->
  let mu,s' = $\compile$$($$e$$)$(s) in
  let v = sample(mu) in (v, s')

$\compile$$($observe($e_1$, $e_2$)$)$ $=$ fun (s1,s2) ->
  let v1,s1' = $\compile$$($$e_1$$)$(s1) in
  let v2,s2' = $\compile$$($$e_2$$)$(s2) in
  let _ = observe(v1,v2) in
  ((), (s1',s2'))

$\compile$$($factor($e$)$)$ $=$ fun s ->
  let v,s' = $\compile$$($$e$$)$(s) in
  let _ = factor(v) in ((), s')

$\compile$$($infer($e$)$)$ $=$ fun sigma ->
  let mu,sigma' = infer($\compile$$($$e$$)$, sigma) in
  (mu, sigma')
  
$\compile$$($let proba $f$ $x$ = $e$$)$ =
  let $f$_init = $\alloc$$($$e$$)$
  let $f$_step = fun (s,$x$) -> $\compile$$($$e$$)$(s)
\end{lstlisting}
\end{minipage}
\end{small}
\vspace{-1em}
\caption{Compilation of \ProbZelus to \muF.}
\label{fig:compilation-full}
\end{figure*}

\Cref{fig:compilation-full} presents the entire compilation function from \muZ to \muF introduced \Cref{sec:compilation}.  \Cref{fig:allocation-full} presents the allocation function.

\begin{lemma}
  \label{lem:typing}
  The compilation preserves the kind (deterministic~$\typeD$, or probabilistic~$\typeP$) of the expressions.
  For any expression~$e$, if $\type{G}{e}{t}{k}$, there exists~$G'$ and~$t'$ such that \(\type{G'}{$\compile$$($e$)$}{t'}{k}\).
\end{lemma}
\begin{proof}
  By induction on the structure of $e$.
\end{proof}

\paragraph{Remark.}
The compilation presented in \Cref{fig:compilation-full} generates a function for each sub-expression.
However, in most cases it is possible to simplify the code using static reduction.
For instance, a constant can directly be compiled into a constant.
  
\section{Inference}
\label{app:inference}

\subsection{Importance Sampling}

\begin{figure}
  \begin{center}
$
  \begin{array}{@{}l@{\ \ \ \ \ \ \ \ }c@{\ \ }l}
  \sem{\zlm!let\ $f$ =\ $e$!}_\gamma &=&
    \gamma[\psem{e}_{\gamma,1}/f] \quad \mathit{if} \ \kindof{e} = {\typeP}
  \\[1.5em]
    
  \psem{e}_{\gamma, w} &=&
    (\sem{e}_\gamma, w) \quad \mathit{if} \ \kindof{e} = {\typeD}\\[0.5em]
   
  \psem{e_1(e_2)}_{\gamma, w} &=&
      \letin{v_2 = \sem{e_2}_\gamma} \sem{e_1}_{\gamma}(v_2, w)\\[0.5em]
  
  \multicolumn{3}{@{}l}{
  \psem{\zlm!if$\;e$ then$\;e_1$ else$\;e_2$!}_{\gamma, w}  ~=~}
  \\
  \multicolumn{3}{l}{
    \qquad
      \mathit{if} \; \sem{e}_\gamma
    \; \mathit{then} \; \psem{e_1}_{\gamma, w}
    \; \mathit{else} \;  \psem{e_2}_{\gamma, w}}
  \\[0.5em]
  
  \multicolumn{3}{@{}l}{
    \psem{\zlm!let$\;p\;$=$\;e_1\;$in$\;e_2$!}_{\gamma, w} ~=~}
  \\
  \multicolumn{3}{l}{
    \qquad
    \letin{v_1, w_1 = \psem{e_1}_{\gamma, w}} \psem{e_2}_{\gamma[v_1 / p], w_1}
  }\\[0.5em]
  
  \psem{\zlm!fun$\;p$ ->$\;e$!}_{\gamma, w} & = &
    \letin{f = \fun{(v, w')} \psem{e}_{[v/p], w'}} (f, w)
  \\[0.5em]
  
  \psem{\zlm!sample($e$)!}_{\gamma, w} &=& ({\rm draw}(\sem{e}_\gamma), w) \\[0.5em]
  
  \psem{\zlm!factor($e$)!}_{\gamma, w} &=& (\zlm!()!, w * \exp(\sem{e}_\gamma))\\[0.5em]
  
  \multicolumn{3}{@{}l}{
    \psem{\zlm!observe($e_1$,$e_2$)!}_{\gamma, w} ~=~}
  \\
  \multicolumn{3}{l}{
    \qquad
    \letin{\mu = \sem{e_1}_\gamma} (\zlm!()!, w * \mu_{\textrm{pdf}}(\sem{e_2}_\gamma))
  }\\[0.5em]
  \end{array}
  $
\end{center}
  \caption{ Importance sampler. Probabilistic expressions return a pair (value, weight). \zlm{sample} draws a sample from a distribution, \zlm{factor} and \zlm{observe} update the weight.}
  \label{fig:importance-full}
  \end{figure}

\paragraph{Importance sampling.}

The most simple inference independently launches~$N$ \emph{particles}.
Each particle executes the importance sampler to compute a pair (result, weight).
Results are then normalized in a \emph{categorical distribution}, i.e., a discrete distribution over the results.

The \zl{infer} operator takes a transition function $\zl!fun $s$ -> $e$!$  and an array of pairs (state, weight)~$S$ of size~$N$ which represents the distribution of possible states across the particles.
\begin{center}
$
\begin{small}
\begin{array}{l}
\sem{\zlm!infer(fun\ $s$ ->\ $e$,\ $S$)!}_\gamma = \\ \quad
  \begin{array}[t]{@{}l@{}}
  \letin{\mu =
    \begin{array}[t]{@{}l@{}}
    \fun{U}\sum\limits_{i = 1}^N \;
      \begin{array}[t]{@{}l@{}}
        \letin{s_i, w_i = \sem{S}_\gamma[i]}\\
        \letin{(v_i, s_i'), w_i' = \psem{\zlm!fun\ $s$ ->\ $e$!}_{\gamma, w_i}(s_i)}\\
        \overline{w_i'} *  \delta_{v_i}(U)
      \end{array}
    \end{array}\\}
    (\mu, [(s_i', w_i')]_{1 \leq i \leq N})
  \end{array}
\end{array}
\end{small}
$
\end{center}
\noindent
At each step, the inference executes one step of all the particles and normalizes the scores to return the distribution~$\mu$ of possible results and an updated array of pairs (state, weight) for the next step.

The weights of the particles are multiplied at each step and never reset.
In other words, the inference reports at each step how likely is the execution path since the beginning of the program for each particle  w.r.t. the model.
Obviously the probability of each individual path quickly collapses to~$0$ after a few steps which makes this inference technique not practical in a reactive context where the inference process never terminates.
The particle filter mitigates this issue by periodically \emph{re-sampling} the set of particles.

\section{Implementation}
\label{app:implementation}

\ProbZelus is open source~(\url{https://github.com/IBM/probzelus}).
It is implemented on top of \zelus~(\url{http://zelus.di.ens.fr/}).
The new constructs \zl{sample}, \zl{observe}, and \zl{factor} are \zelus nodes implemented directly in OCaml.
The \zl{infer} construct is a node that take as argument the \zelus node that represents the probabilistic model.
The \zl{infer} node thus takes as argument the allocation and step functions of the model as argument which corresponds to the compilation described in \Cref{sec:compilation}.

\paragraph{Relationship with the paper}

The code corresponding to the paper is available as a release \url{https://github.com/IBM/probzelus/tree/pldi20}.
The example of \Cref{fig:lqr_src} is in \lstinline[breaklines=true,columns=fullflexible]`examples/tracker/tracker_ds.zls`.

The compiler implements the compilation scheme presented in \Cref{sec:compilation} with a few optimizations: (1)~intermediate step functions are statically reduced (2)~useless state is removed when possible, and (3)~state is updated imperatively.
Moreover, the compilation of \zl{proba} nodes introduces an extra argument to the step functions in order to pass the extra information~$w$ or $(w, g)$ needed by the inference algorithms.

The code of the inferences algorithms is in the \lstinline[breaklines=true,columns=fullflexible]`inference` directory.
The particle filter presented in \Cref{sec:pf} is in \lstinline[breaklines=true,columns=fullflexible]`infer_pf.ml`.
The entry point of the Delayed Sampling algorithm presented in \Cref{sec:ds} is \lstinline[breaklines=true,columns=fullflexible]`infer_ds_naive.ml` and the core of the algorithm is \lstinline[breaklines=true,columns=fullflexible]`ds_naive_graph.ml`.
The entry point and the core of the algorithm for the Streaming Delayed Sampling algorithm presented in \Cref{sec:sds} are respectively in \lstinline[breaklines=true,columns=fullflexible]`infer_ds_streaming.ml` and \lstinline[breaklines=true,columns=fullflexible]`ds_streaming_graph.ml`.

The Bounded Delayed Sampling algorithm presented in \Cref{sec:ds} can be implemented on top of both classical and streaming delayed sampling.
The code is in the functor defined in \lstinline[breaklines=true,columns=fullflexible]`ds_high_level.ml`.

Finally, the code for the benchmarks presented \Cref{sec:eval} and \Cref{sec:evaluation} is available in \lstinline[breaklines=true,columns=fullflexible]`examples/benchmarks`.

\paragraph{Artifact}

There is an artifact associated to the paper which is available with~\cite{rppl-short}. It is distributed as a Linux image in the Open Virtualization Format that can be launch using a virtualization player like VirtualBox~(\url{https://www.virtualbox.org}). The credential to log into the virtual machine are:
\begin{lstlisting}
  user: probzelus
  passord: probzelus
\end{lstlisting}

\section{Performance Evaluation}
\label{sec:evaluation}

This section presents the experimental results.
We ran each inference algorithm on a series of benchmarks and
measured properties of the execution: accuracy, execution time, memory consumption.
All the experiments were run on a server with 32 CPUs (2.60 GHz) and 128 GB memory.

\subsection{Benchmarks}
\label{sec:benchmarks}

\paragraph{Beta-Bernoulli.} The Beta-Bernoulli benchmark models an agent that estimates the bias of a coin.
\begin{lstlisting}
let proba coin (yobs) = xt where
  rec init xt = sample (beta (1., 1.))
  and () = observe (bernoulli xt, yobs)
\end{lstlisting}
The model samples \zl{zt} from a
$\mathit{Beta}(1,1)$ distribution, and thereafter evaluates the observations with a $\mathit{Bernoulli}$ distribution of parameter~\zl{xt}.
Running SDS on this model is
equivalent to exact inference in a Beta-Bernoulli conjugate
model~\cite{conjprior} where each particle returns the exact solution. The benchmark's error metric is the mean squared error over time
between the true coin probability and the expected probability conditioned on
the stream of observations.

\paragraph{Gaussian-Gaussian.} The Gaussian-Gaussian benchmark models an agent that estimates the mean and the standard deviation of a Gaussian.
\pagebreak
\begin{lstlisting}
let proba gaussian_model (o) = (mu, sigma) where
  rec init mu = sample (gaussian (0., 10.))
  and init sqrt_sigma = sample (gaussian (0., 1.))
  and sigma = sqrt_sigma *. sqrt_sigma
  and () = observe (gaussian (mu, sigma), o)
\end{lstlisting}
The initial values for the distribution of the mean \zl{mu} follows a distribution $\mathcal{N}(0, 10)$ and the distribution of $\sqrt{\zl!sigma!}$ is $\mathcal{N}(0, 1)$.
The distributions of \zl{mu} and \zl{sigma} are conditioned by the observations that follow a distribution $\mathcal{N}(\zl{mu}, \zl{sigma})$.
In the current implementations of delayed sampling we are doing exact inference only on the mean and not on the standard deviation (even if it would be possible).
The benchmark's error metric is the mean squared error over time
between the true mean and standard deviation and the expected probability conditioned on the stream of observations.

\paragraph{Kalman.} The Kalman benchmark models an agent that estimates its position based on noisy observations.
\begin{lstlisting}
let proba delay_kalman (yobs) = xt where
  rec xt = sample (gaussian ((0., 2500.) ->
                             (pre xt, 1.)))
  and () = observe (gaussian (xt, 1.), yobs)
\end{lstlisting}
The model chooses an initial position from $\mathcal{N}(0, 2500)$, and
chooses subsequent positions from $\mathcal{N}(\zl!pre x!, 1)$ where \zl!pre x! denote the previous position. The model draws the observation at each
time step from $\mathcal{N}(\zl!x!, 1)$ where \zl{x} is the true position.
Running SDS on this model is equivalent to a Kalman
filter~\cite{kalman} where each particle returns the exact solution.  The benchmark's error metric is the mean squared error over time
between the true position and the expected position conditioned on all previous
observations.

\paragraph{Outlier.} The Outlier benchmark, adapted from Section 2 of~\cite{ep},
models the same situation as the Kalman benchmark, but with a sensor that 
occasionally produces invalid readings.
\begin{lstlisting}
let proba outlier (yobs) = (is_outlier, xt) where
  rec xt = sample (gaussian ((0., 2500.) ->
                             (pre xt, 1.)))
  and init outlier_prob = sample (beta (100., 1000.))
  and is_outlier = sample (bernoulli outlier_prob)
  and () = present is_outlier ->
             observe (gaussian (0., 10000.), yobs)
           else observe (gaussian (xt, 1.), yobs)
\end{lstlisting}
The model chooses the probability of an
invalid reading from a $\mathrm{Beta}(100,1000)$ distribution, so that invalid
readings occur approximately 10\% of the time. At each time step, with the
previously chosen probability, the model either chooses the observation from
the invalid distribution $\mathcal{N}(0, 10000)$, or it
chooses the observation from the Kalman model. Running SDS on this
model is equivalent to a Rao-Blackwellized particle filter~\cite{rbpf}
that combines exact inference with approximate particle filtering.  The benchmark's error metric is the mean
squared error over time between the true position and the expected position
conditioned on all previous observations.

\paragraph{Robot.} The Robot benchmark is detailed \Cref{sec:example}.

\paragraph{SLAM.} Simultaneous Location And Mapping (SLAM)~\cite{Montemerlo02-fastslam}.
Consider the simple case where a robot evolves in a discrete one-dimensional world and each position corresponds to a black or white cell.
The robot can move from left to right and can observe the color of the cell on which it stands with a sensor.
There are two sources of uncertainty: (1)~the robot's wheels are slippery, so the robot can sometimes stay on the spot thinking about moving,
(2)~the sensor is making read errors, and can reverse the colors.
The controller tries to infer the map~(color of the cells) and the current position of the robot~(\Cref{fig:slam-screenshots}).

\begin{figure}[t]
  \begin{center}
    \includegraphics[scale=0.35]{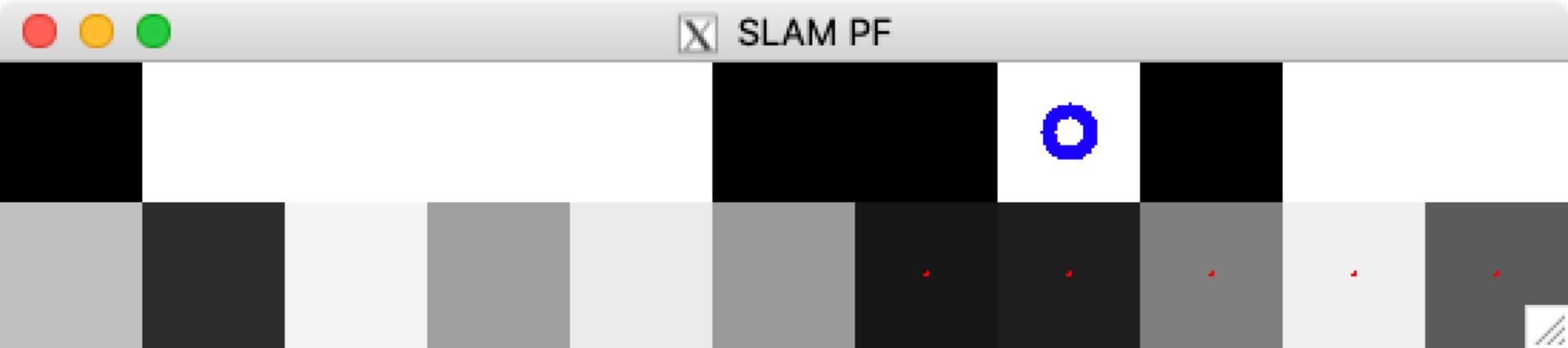}\\[0.5em]
    \includegraphics[scale=0.35]{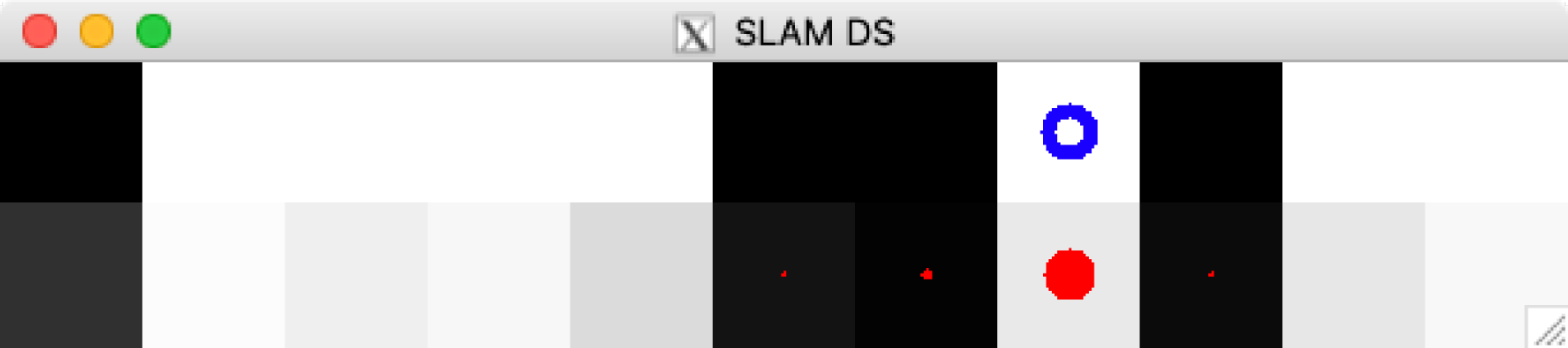}
  \end{center}
  \caption{Screenshots of the execution of the SLAM with the PF and SDS inferences.
    For each screenshot, the top line shows the map, and the blue circle the exact position of the robot.
    The lower line represents the inferred map where the gray level indicates the probability for the cell to be black and the red dots the probability of presence of the robot on the cell.}
   \label{fig:slam-screenshots}
\end{figure}

The robot maintains a map where each box is a random variable that represents the probability of being black or white (gray level in the \Cref{fig:slam-screenshots}).
The \textit{a priori} distribution of these random variables is a $\mathit{Bernoulli}(0.5)$ distribution:
\begin{lstlisting}
let proba bernoulli_priors i = sample (bernoulli 0.5)
\end{lstlisting}

The robot starts from the position~\zl{x0} and receives at each step a command \zl{Right} or \zl{Left}.
It then moves to the left or right following the command with a probability of~$10\%$ of remaining in place~(modeled by a Bernoulli distribution of parameter~$0.1$).

\begin{lstlisting}
let proba move (x0, cmd) = x where
  rec slip = sample (bernoulli 0.1)
  and xp = x0 -> pre x
  and x = match cmd with
      | Right ->
          min max_pos (if slip then xp else xp + 1)
      | Left  ->
          max min_pos (if slip then xp else xp - 1)
      end
\end{lstlisting}

The sensor has a constant probability of reading error of \zl{sensor_noise}.
At each instant, the robot computes its current position~\zl{x}.
The observation of the sensor follows a Bernoulli distribution parameterized by \zl{1 - sensor_noise} if the position is white and \zl{sensor_noise} if the position is black.
\begin{lstlisting}
let proba slam (obs, cmd) = (map, x) where
  rec init map = Array_misc.ini (max_pos + 1)
                   bernoulli_priors ()
  and x = move (0, cmd)
  and o = Array_misc.get map x
  and p = if o then (1. -. sensor_noise)
          else sensor_noise
  and () = observe (bernoulli p, obs)
\end{lstlisting}

The benchmark's error metrics is the mean squared error over time between the exact map and position and the expected map and position.

\paragraph{Multi-Target Tracker.}
\textsf{MTT} (\emph{Multi-Target Tracker}) adapted from~\cite{MurrayS18} is a
model where there are a variable number of targets with linear-Gaussian motion
models producing linear-Gaussian measurements of the position at each time
step.  Targets randomly appear according to a Poisson process and each
disappear with some fixed probability at each time step.  Measurements do not
identify which target they came from, and ``clutter'' measurements that come
not from targets but from some underlying distribution add to observations,
complicating inference of which measurements are associated to which targets.

We model this with a ProbZelus program that has a state consisting of a list of
position-velocity pairs that encode the {\em track} of each target. In this
example we consider two-dimensional targets, giving us a 4-dimensional vector
representing position and velocity together.

The first step is to define helper
functions that will be mapped over the list of tracks. 
The first function tells us how frequently tracks die. 
They do so with probability \zl{p_dead} which we set to $e^{-.02}$.

\begin{lstlisting}
let proba death_fn _ = sample (bernoulli (p_dead))
\end{lstlisting}

We now define how tracks are initialized when they are first created. They are
sampled from a multivariate Gaussian distribution with mean \zl{mu_new} set to zero and covariance \zl{sigma_new} set to a diagonal with
variance $1$ on the positions and variance $0.001$ on the velocities.

\begin{lstlisting}
let proba new_track_init_fn _ =
  (new_track_num (),
   sample (mv_gaussian (mu_new, sigma_new)))
\end{lstlisting}

Next, we define the motion model and update model. Each track \zl{tr} is
multiplied with the motion matrix \zl{a_u} which encodes discrete time
integration of position and velocity with time constant $1$. We then sample a
Gaussian distribution around the new position and velocity with covariance
\zl{sigma_update} which is a diagonal matrix with $0.01$ variance for the
position and $0.1$ for the variance of the velocity. For the observation model,
we project out the position with the projection matrix \zl{proj_pos} and
observe it with covariance matrix \zl{sigma_obs} which we set to a diagonal of~$0.1$.

\begin{lstlisting}
let proba state_update_fn (tr_num, tr) =
  (tr_num,
   sample (mv_gaussian(a_u *@ tr, sigma_update)))

let observe_fn (_, tr) =
  (mv_gaussian (proj_pos *@ tr, sigma_obs))
\end{lstlisting}

We next define the model for clutter data. We assume that each clutter point is
drawn from a multivariate Gaussian with mean \zl{mu_clutter} which is zero an
covariance \zl{sigma_clutter} which we set to $10$.

\begin{lstlisting}
let proba clutter_init_fn _ =
  (mv_gaussian (mu_clutter, sigma_clutter))
\end{lstlisting}

The model proceeds as follows. For every track, we use the \zl{filter} list
operator to remove all the tracks that died in this time step. We then sample
the number of new tracks\zl{n_new} from a Poisson distribution with parameter
\zl{lambda_new} which we set to $0.1$. After forcing a sample of this value, we
use the list constructer \zl{ini} to build a list of new tracks. We append the
survived and new tracks together, and then use the \zl{map} list operator to
apply the motion and observation models to each track. Next, we determine the
amount of clutter by subtracting the number of observations from the number of
surviving tracks. We then observe that this comes from a Poisson distribution
with parameter \zl{lambda_clutter} which we set to $1$. Note that this sets the
particle weight to~$-\infty$ if the particle yields a negative amount of clutter.
Next, we shuffle the track observations and the clutter together by forcing a
sample of the \zl{shuffle} random primitive, and finally observe that the
resulting list yields the observed values.

\begin{lstlisting}
let proba obsfn (var, value) = observe (var, value)

let proba model inp = t where
  rec init t = []
  and t_survived = filter death_fn (last t)
  and n_new = sample (poisson lambda_new)
  and t_new = ini new_track_init_fn n_new
  and t_tot = append t_survived t_new
  and t = map state_update_fn t_tot
  and obs = map observe_fn t
  and n_clutter = (length inp) - (length obs)
  and () = observe(poisson lambda_clutter, n_clutter)
  and clutter = ini clutter_init_fn n_clutter
  and obs_shuffled =
    sample (shuffle (append obs clutter))
  and present (not (n_clutter < 0)) ->
    do () = (iter2 obsfn (obs_shuffled, inp)) done
\end{lstlisting}

The accuracy metric is based on the {\em Multiple Object Tracking
Accuracy}~\cite{MOTA08}. This evaluates whether a track estimate contains the
right targets across all time steps within a sufficient tolerance (we set the
tolerance to $5$ in our example). Conventional MOTA is in $[0,1]$ with $1$
being the best; we have modified it to be in $[0,\infty]$ with $0$ being the
best by transforming it to $\mathsf{MOTA*} = 1/\mathsf{MOTA} - 1$.

Because we estimate a distribution of track estimates, we draw a sample from
the track distribution to estimate the expected $\mathsf{MOTA}*$.

\paragraph{Data.} For each benchmark except Robot and SLAM, we obtained
observation data by sampling from the benchmark's model. In these benchmarks,
every run of each benchmark across all experiments uses the same data as input.
For SLAM, we pre-sampled the map from the model, but sampled position data on
the fly as this data depends on the controller. For the Robot benchmark, we
sampled all observations on the fly because they all depend on the command from
the controller. This means that for SLAM and Robot, each run uses different
position observations.

\subsection{DS vs. PF}
\label{sec:perf_accuracy}

\newcommand{\accuracyTrials}[0]{1000\xspace}

We compare both the accuracy and runtime performance of BDS, SDS, and PF to investigate whether the delayed
samplers can achieve better accuracy than the particle filter with the same
amount of computational resources.

\paragraph{Accuracy Methodology.} For a range of selected particle counts, we
execute each benchmark multiple times and record the resulting
accuracy. To measure accuracy we use the end-to-end error metrics for each
benchmark as described in Section~\ref{sec:benchmarks}. We record the median
and the 90\% and 10\% quantiles after 1000 runs.

\paragraph{Accuracy Results.} \Cref{fig:accuracy1,fig:accuracy2} show the results of the accuracy experiment for the different benchmarks.
The error bars show 90\% and 10\% quantiles, and the center is the median.
The vertical lines corresponds to the number reported in \Cref{fig:evaluation} where there is enough particles to achieve similar accuracy to delayed sampling with 1000 particles.
In all cases, SDS is able to achieve equal or better
accuracy than BDS which is itself equal or better than PF, but the results vary widely by benchmark.
Note that SDS returns the exact posterior distribution for the Coin and Kalman
benchmarks therefore its accuracy is independent of the
number of particles.
On the other-hand, BDS is not exact since the symbolic distributions are sampled at the end of each the step.

\newcommand{\gcds}[0]{pointer minimal delayed sampling}
\newcommand{\performanceTrials}[0]{\accuracyTrials}
\newcommand{\performanceWarmup}{1}
\newcommand{\numParticles}{100\xspace}
\newcommand{\performanceSteps}[0]{YYY}

\paragraph{Performance Methodology.} 
For a range of selected particle counts, we execute each benchmark multiple
times (the same number as for the accuracy experiments described above)  after a
warm-up of {\performanceWarmup} run and record the resulting performance: the
latency of one step of computation.  In the following graphs we report the
median latency as well as the 90\% and 10\% quantiles of the collected data.

\paragraph{Performance Results.}
\Cref{fig:perf-particles1,fig:perf-particles2} shows how the latency for a single step varies with the number of particles for each benchmark.
The error bars show 90\% and 10\% quantiles, and the center is the median.
With the three algorithms, the execution time increases linearly with the number of particles. In all cases, PF has lower latency than BDS which has lower latency than SDS.

\paragraph{Conclusions.} These experiments show that the delayed samplers
achieve better accuracy than the particle filter with the same computational
resources. 
For some models SDS is able to compute the exact solution with only one particle (Kalman, Coin).
BDS achieves better accuracy when relationships between variables defined in the same step can be exploited (Kalman).
At worst the delayed samplers performs as a well as the particle filter (BDS on the Coin, SDS and BDS on the Outlier).

\subsection{SDS vs. DS}
\label{sec:perf_memory}

We next evaluate the performance of SDS and BDS relative to our own OCaml implementation of the original delayed sampler~(DS). We compare both the performance and memory consumption of the three algorithms at each time step to investigate whether, as the size of the
input stream grows large, they can retain  constant performance.

\newcommand{\performancePerStepWarmup}{1}
\newcommand{\performancePerStepTrials}{1000}

\paragraph{Performance Methodology.}
We execute each benchmark
{\performancePerStepTrials} times after a warm-up of {\performancePerStepWarmup}
run and record the latency. We execute each benchmark with
\numParticles particles~(even if only one particle is necessary for
DS and SDS on the Coin and Kalman benchmarks to compute the exact distribution) and plot latency as a function of the time
step.  We report the median latency as well as the 90\% and 10\%
quantiles of the collected data.

\paragraph{Performance Results.}
\Cref{fig:perf-step1,fig:perf-step2} shows the latency at each step of a run, aggregated over {\performancePerStepTrials} runs.
PF, BDS, and SDS show nearly constant performance in time but DS gets linearly worse performance for the Kalman and Outlier benchmarks.
For the Coin benchmark, the graph of DS remains of constant size because there is only one \zl+sample+ at the first step and then only \zl{observe} statements.

\paragraph{Memory Methodology.} We next evaluate the memory consumption of the
algorithms. For all benchmarks except the multi-target tracker, memory
consumption is deterministic even in the presence of random choices. Therefore,
we measure the {\em ideal} memory consumption of the execution of each
benchmark after each step. The ideal memory consumption is the total amount of
live words in the program's heap. In our implementation, we measure these
numbers by forcing a garbage collection after each step. We use OCaml's
standard facilities for forcing garbage collection as well as for measuring the
amount of live words. We ran each algorithms 10 times with {\numParticles} particles.

For the multi-target tracker, the memory is not deterministic because it is
determined by the number of hypothetical tracks, which is random. We report
median and 10\% and 90\% values for memory consumption for this benchmark.

\paragraph{Memory Results.} \Cref{fig:mem1,fig:mem2} shows the results of the
memory consumption experiment. For all benchmarks, PF, BDS, SDS use constant
memory over time, including for the multi-target tracker where their memory
consumption is random at each time step. However, DS has increasing memory
consumption over time for the Kalman, Outlier, and Robot benchmarks. The
memory consumption of DS is constant for the Coin benchmark because the graph
remains of constant size.

For the mutli-target tracker, the memory consumption of DS is based both on the
number of hypothesized tracks and the length of the hypothesized tracks. We can
see that the memory consumption of DS increases as the first generation of
tracks becomes longer, but eventually curtails its memory consuption when these
tracks die. MTT's memory consumption thereafter increases again as the second
generation of tracks starts to increase in length.

\paragraph{Conclusions.}
The original DS implementation consumes an increasing amount of memory over
time for models that introduce new variables at each step (Kalman, Outlier, and
Robot) in contrast to BDS and SDS whose memory consumptions are constant over
time. For the multi-target tracker, the DS memory consumption is based on the
length of the track which is in principle probabilistically bounded. However,
DS still consumes much more memory than PF, BDS, and SDS because the tracks are
long-lived.

Furthermore,
DS step latency increases without bound as the number of steps becomes large on
benchmarks where the memory increases.  These observations confirm that the
original DS implementation is not practical in a reactive settings.

\begin{figure*}[p]
  \begin{center}
    \small\sf
    \pfbullet~PF $\quad$ \bdsbullet~BDS $\quad$ \pmdsbullet~SDS
  \end{center}
    \begin{minipage}[t]{0.48\textwidth}
      \small\sf
        \hspace*{-1.8em}
        \begingroup
  \makeatletter
  \providecommand\color[2][]{\GenericError{(gnuplot) \space\space\space\@spaces}{Package color not loaded in conjunction with
      terminal option `colourtext'}{See the gnuplot documentation for explanation.}{Either use 'blacktext' in gnuplot or load the package
      color.sty in LaTeX.}\renewcommand\color[2][]{}}\providecommand\includegraphics[2][]{\GenericError{(gnuplot) \space\space\space\@spaces}{Package graphicx or graphics not loaded}{See the gnuplot documentation for explanation.}{The gnuplot epslatex terminal needs graphicx.sty or graphics.sty.}\renewcommand\includegraphics[2][]{}}\providecommand\rotatebox[2]{#2}\@ifundefined{ifGPcolor}{\newif\ifGPcolor
    \GPcolorfalse
  }{}\@ifundefined{ifGPblacktext}{\newif\ifGPblacktext
    \GPblacktexttrue
  }{}\let\gplgaddtomacro\g@addto@macro
\gdef\gplbacktext{}\gdef\gplfronttext{}\makeatother
  \ifGPblacktext
\def\colorrgb#1{}\def\colorgray#1{}\else
\ifGPcolor
      \def\colorrgb#1{\color[rgb]{#1}}\def\colorgray#1{\color[gray]{#1}}\expandafter\def\csname LTw\endcsname{\color{white}}\expandafter\def\csname LTb\endcsname{\color{black}}\expandafter\def\csname LTa\endcsname{\color{black}}\expandafter\def\csname LT0\endcsname{\color[rgb]{1,0,0}}\expandafter\def\csname LT1\endcsname{\color[rgb]{0,1,0}}\expandafter\def\csname LT2\endcsname{\color[rgb]{0,0,1}}\expandafter\def\csname LT3\endcsname{\color[rgb]{1,0,1}}\expandafter\def\csname LT4\endcsname{\color[rgb]{0,1,1}}\expandafter\def\csname LT5\endcsname{\color[rgb]{1,1,0}}\expandafter\def\csname LT6\endcsname{\color[rgb]{0,0,0}}\expandafter\def\csname LT7\endcsname{\color[rgb]{1,0.3,0}}\expandafter\def\csname LT8\endcsname{\color[rgb]{0.5,0.5,0.5}}\else
\def\colorrgb#1{\color{black}}\def\colorgray#1{\color[gray]{#1}}\expandafter\def\csname LTw\endcsname{\color{white}}\expandafter\def\csname LTb\endcsname{\color{black}}\expandafter\def\csname LTa\endcsname{\color{black}}\expandafter\def\csname LT0\endcsname{\color{black}}\expandafter\def\csname LT1\endcsname{\color{black}}\expandafter\def\csname LT2\endcsname{\color{black}}\expandafter\def\csname LT3\endcsname{\color{black}}\expandafter\def\csname LT4\endcsname{\color{black}}\expandafter\def\csname LT5\endcsname{\color{black}}\expandafter\def\csname LT6\endcsname{\color{black}}\expandafter\def\csname LT7\endcsname{\color{black}}\expandafter\def\csname LT8\endcsname{\color{black}}\fi
  \fi
    \setlength{\unitlength}{0.0500bp}\ifx\gptboxheight\undefined \newlength{\gptboxheight}\newlength{\gptboxwidth}\newsavebox{\gptboxtext}\fi \setlength{\fboxrule}{0.5pt}\setlength{\fboxsep}{1pt}\begin{picture}(4818.00,3174.00)\gplgaddtomacro\gplbacktext{\csname LTb\endcsname \put(1188,550){\makebox(0,0)[r]{\strut{}$10^{-4}$}}\put(1188,1068){\makebox(0,0)[r]{\strut{}$10^{-3}$}}\put(1188,1587){\makebox(0,0)[r]{\strut{}$10^{-2}$}}\put(1188,2105){\makebox(0,0)[r]{\strut{}$10^{-1}$}}\put(1188,2623){\makebox(0,0)[r]{\strut{}$10^{0}$}}\put(1377,330){\makebox(0,0){\strut{}$1$}}\put(2185,330){\makebox(0,0){\strut{}$10$}}\put(2993,330){\makebox(0,0){\strut{}$100$}}\put(3801,330){\makebox(0,0){\strut{}$1000$}}\put(4609,330){\makebox(0,0){\strut{}$10000$}}}\gplgaddtomacro\gplfronttext{\csname LTb\endcsname \put(3035,2733){\makebox(0,0){\strut{}Beta-Bernoulli Accuracy}}}\gplbacktext
    \put(0,0){\includegraphics{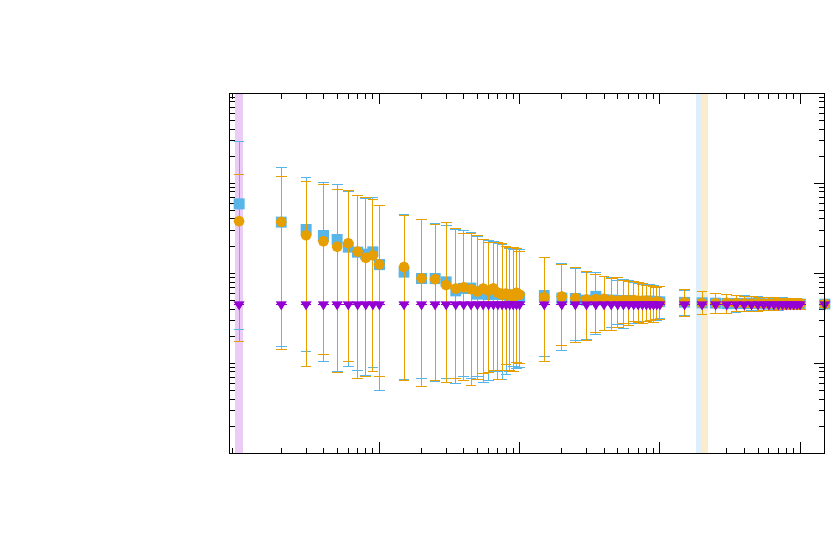}}\gplfronttext
  \end{picture}\endgroup
 \\[-0.3cm]
        \hspace*{-1.8em}
        \begingroup
  \makeatletter
  \providecommand\color[2][]{\GenericError{(gnuplot) \space\space\space\@spaces}{Package color not loaded in conjunction with
      terminal option `colourtext'}{See the gnuplot documentation for explanation.}{Either use 'blacktext' in gnuplot or load the package
      color.sty in LaTeX.}\renewcommand\color[2][]{}}\providecommand\includegraphics[2][]{\GenericError{(gnuplot) \space\space\space\@spaces}{Package graphicx or graphics not loaded}{See the gnuplot documentation for explanation.}{The gnuplot epslatex terminal needs graphicx.sty or graphics.sty.}\renewcommand\includegraphics[2][]{}}\providecommand\rotatebox[2]{#2}\@ifundefined{ifGPcolor}{\newif\ifGPcolor
    \GPcolorfalse
  }{}\@ifundefined{ifGPblacktext}{\newif\ifGPblacktext
    \GPblacktexttrue
  }{}\let\gplgaddtomacro\g@addto@macro
\gdef\gplbacktext{}\gdef\gplfronttext{}\makeatother
  \ifGPblacktext
\def\colorrgb#1{}\def\colorgray#1{}\else
\ifGPcolor
      \def\colorrgb#1{\color[rgb]{#1}}\def\colorgray#1{\color[gray]{#1}}\expandafter\def\csname LTw\endcsname{\color{white}}\expandafter\def\csname LTb\endcsname{\color{black}}\expandafter\def\csname LTa\endcsname{\color{black}}\expandafter\def\csname LT0\endcsname{\color[rgb]{1,0,0}}\expandafter\def\csname LT1\endcsname{\color[rgb]{0,1,0}}\expandafter\def\csname LT2\endcsname{\color[rgb]{0,0,1}}\expandafter\def\csname LT3\endcsname{\color[rgb]{1,0,1}}\expandafter\def\csname LT4\endcsname{\color[rgb]{0,1,1}}\expandafter\def\csname LT5\endcsname{\color[rgb]{1,1,0}}\expandafter\def\csname LT6\endcsname{\color[rgb]{0,0,0}}\expandafter\def\csname LT7\endcsname{\color[rgb]{1,0.3,0}}\expandafter\def\csname LT8\endcsname{\color[rgb]{0.5,0.5,0.5}}\else
\def\colorrgb#1{\color{black}}\def\colorgray#1{\color[gray]{#1}}\expandafter\def\csname LTw\endcsname{\color{white}}\expandafter\def\csname LTb\endcsname{\color{black}}\expandafter\def\csname LTa\endcsname{\color{black}}\expandafter\def\csname LT0\endcsname{\color{black}}\expandafter\def\csname LT1\endcsname{\color{black}}\expandafter\def\csname LT2\endcsname{\color{black}}\expandafter\def\csname LT3\endcsname{\color{black}}\expandafter\def\csname LT4\endcsname{\color{black}}\expandafter\def\csname LT5\endcsname{\color{black}}\expandafter\def\csname LT6\endcsname{\color{black}}\expandafter\def\csname LT7\endcsname{\color{black}}\expandafter\def\csname LT8\endcsname{\color{black}}\fi
  \fi
    \setlength{\unitlength}{0.0500bp}\ifx\gptboxheight\undefined \newlength{\gptboxheight}\newlength{\gptboxwidth}\newsavebox{\gptboxtext}\fi \setlength{\fboxrule}{0.5pt}\setlength{\fboxsep}{1pt}\begin{picture}(4818.00,3174.00)\gplgaddtomacro\gplbacktext{\csname LTb\endcsname \put(1188,550){\makebox(0,0)[r]{\strut{}$10^{-2}$}}\put(1188,965){\makebox(0,0)[r]{\strut{}$10^{-1}$}}\put(1188,1379){\makebox(0,0)[r]{\strut{}$10^{0}$}}\put(1188,1794){\makebox(0,0)[r]{\strut{}$10^{1}$}}\put(1188,2208){\makebox(0,0)[r]{\strut{}$10^{2}$}}\put(1188,2623){\makebox(0,0)[r]{\strut{}$10^{3}$}}\put(1377,330){\makebox(0,0){\strut{}$1$}}\put(2185,330){\makebox(0,0){\strut{}$10$}}\put(2993,330){\makebox(0,0){\strut{}$100$}}\put(3801,330){\makebox(0,0){\strut{}$1000$}}\put(4609,330){\makebox(0,0){\strut{}$10000$}}}\gplgaddtomacro\gplfronttext{\csname LTb\endcsname \put(583,1586){\rotatebox{-270}{\makebox(0,0){\strut{}Loss (log scale)}}}\put(3035,2733){\makebox(0,0){\strut{}Gaussian-Gaussian Accuracy}}}\gplbacktext
    \put(0,0){\includegraphics{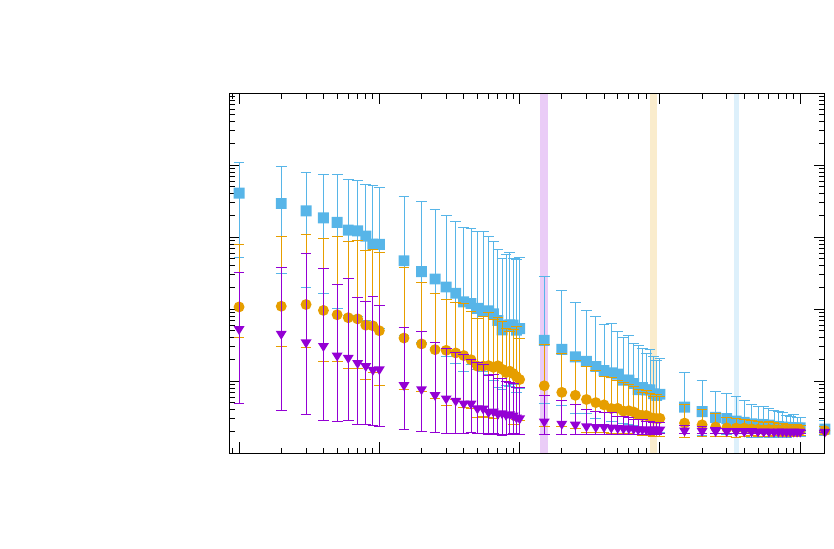}}\gplfronttext
  \end{picture}\endgroup
 \\[-0.3cm]
        \hspace*{-1.8em}
        \begingroup
  \makeatletter
  \providecommand\color[2][]{\GenericError{(gnuplot) \space\space\space\@spaces}{Package color not loaded in conjunction with
      terminal option `colourtext'}{See the gnuplot documentation for explanation.}{Either use 'blacktext' in gnuplot or load the package
      color.sty in LaTeX.}\renewcommand\color[2][]{}}\providecommand\includegraphics[2][]{\GenericError{(gnuplot) \space\space\space\@spaces}{Package graphicx or graphics not loaded}{See the gnuplot documentation for explanation.}{The gnuplot epslatex terminal needs graphicx.sty or graphics.sty.}\renewcommand\includegraphics[2][]{}}\providecommand\rotatebox[2]{#2}\@ifundefined{ifGPcolor}{\newif\ifGPcolor
    \GPcolorfalse
  }{}\@ifundefined{ifGPblacktext}{\newif\ifGPblacktext
    \GPblacktexttrue
  }{}\let\gplgaddtomacro\g@addto@macro
\gdef\gplbacktext{}\gdef\gplfronttext{}\makeatother
  \ifGPblacktext
\def\colorrgb#1{}\def\colorgray#1{}\else
\ifGPcolor
      \def\colorrgb#1{\color[rgb]{#1}}\def\colorgray#1{\color[gray]{#1}}\expandafter\def\csname LTw\endcsname{\color{white}}\expandafter\def\csname LTb\endcsname{\color{black}}\expandafter\def\csname LTa\endcsname{\color{black}}\expandafter\def\csname LT0\endcsname{\color[rgb]{1,0,0}}\expandafter\def\csname LT1\endcsname{\color[rgb]{0,1,0}}\expandafter\def\csname LT2\endcsname{\color[rgb]{0,0,1}}\expandafter\def\csname LT3\endcsname{\color[rgb]{1,0,1}}\expandafter\def\csname LT4\endcsname{\color[rgb]{0,1,1}}\expandafter\def\csname LT5\endcsname{\color[rgb]{1,1,0}}\expandafter\def\csname LT6\endcsname{\color[rgb]{0,0,0}}\expandafter\def\csname LT7\endcsname{\color[rgb]{1,0.3,0}}\expandafter\def\csname LT8\endcsname{\color[rgb]{0.5,0.5,0.5}}\else
\def\colorrgb#1{\color{black}}\def\colorgray#1{\color[gray]{#1}}\expandafter\def\csname LTw\endcsname{\color{white}}\expandafter\def\csname LTb\endcsname{\color{black}}\expandafter\def\csname LTa\endcsname{\color{black}}\expandafter\def\csname LT0\endcsname{\color{black}}\expandafter\def\csname LT1\endcsname{\color{black}}\expandafter\def\csname LT2\endcsname{\color{black}}\expandafter\def\csname LT3\endcsname{\color{black}}\expandafter\def\csname LT4\endcsname{\color{black}}\expandafter\def\csname LT5\endcsname{\color{black}}\expandafter\def\csname LT6\endcsname{\color{black}}\expandafter\def\csname LT7\endcsname{\color{black}}\expandafter\def\csname LT8\endcsname{\color{black}}\fi
  \fi
    \setlength{\unitlength}{0.0500bp}\ifx\gptboxheight\undefined \newlength{\gptboxheight}\newlength{\gptboxwidth}\newsavebox{\gptboxtext}\fi \setlength{\fboxrule}{0.5pt}\setlength{\fboxsep}{1pt}\begin{picture}(4818.00,3174.00)\gplgaddtomacro\gplbacktext{\csname LTb\endcsname \put(1188,550){\makebox(0,0)[r]{\strut{}$10^{-1}$}}\put(1188,1241){\makebox(0,0)[r]{\strut{}$10^{0}$}}\put(1188,1932){\makebox(0,0)[r]{\strut{}$10^{1}$}}\put(1188,2623){\makebox(0,0)[r]{\strut{}$10^{2}$}}\put(1377,330){\makebox(0,0){\strut{}$1$}}\put(2185,330){\makebox(0,0){\strut{}$10$}}\put(2993,330){\makebox(0,0){\strut{}$100$}}\put(3801,330){\makebox(0,0){\strut{}$1000$}}\put(4609,330){\makebox(0,0){\strut{}$10000$}}}\gplgaddtomacro\gplfronttext{\csname LTb\endcsname \put(3035,2733){\makebox(0,0){\strut{}Kalman-1D Accuracy}}}\gplbacktext
    \put(0,0){\includegraphics{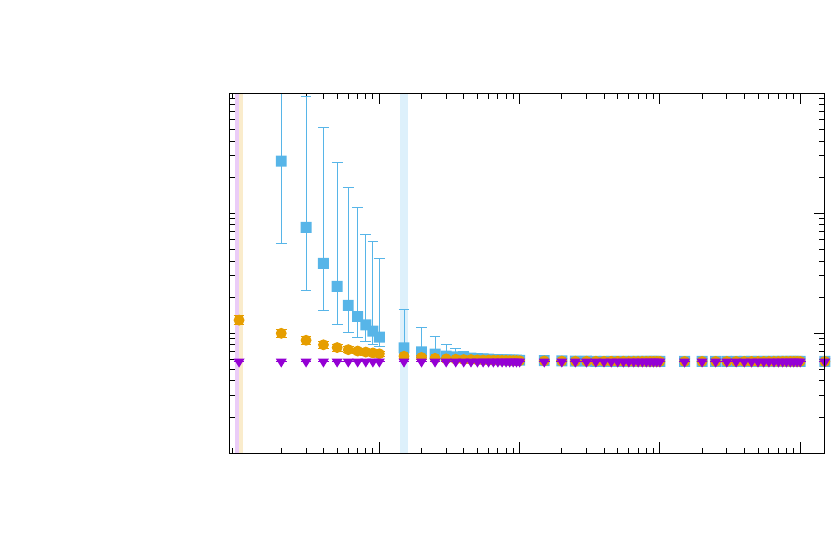}}\gplfronttext
  \end{picture}\endgroup
 \\[-0.3cm]
        \hspace*{-1.8em}
        \begingroup
  \makeatletter
  \providecommand\color[2][]{\GenericError{(gnuplot) \space\space\space\@spaces}{Package color not loaded in conjunction with
      terminal option `colourtext'}{See the gnuplot documentation for explanation.}{Either use 'blacktext' in gnuplot or load the package
      color.sty in LaTeX.}\renewcommand\color[2][]{}}\providecommand\includegraphics[2][]{\GenericError{(gnuplot) \space\space\space\@spaces}{Package graphicx or graphics not loaded}{See the gnuplot documentation for explanation.}{The gnuplot epslatex terminal needs graphicx.sty or graphics.sty.}\renewcommand\includegraphics[2][]{}}\providecommand\rotatebox[2]{#2}\@ifundefined{ifGPcolor}{\newif\ifGPcolor
    \GPcolorfalse
  }{}\@ifundefined{ifGPblacktext}{\newif\ifGPblacktext
    \GPblacktexttrue
  }{}\let\gplgaddtomacro\g@addto@macro
\gdef\gplbacktext{}\gdef\gplfronttext{}\makeatother
  \ifGPblacktext
\def\colorrgb#1{}\def\colorgray#1{}\else
\ifGPcolor
      \def\colorrgb#1{\color[rgb]{#1}}\def\colorgray#1{\color[gray]{#1}}\expandafter\def\csname LTw\endcsname{\color{white}}\expandafter\def\csname LTb\endcsname{\color{black}}\expandafter\def\csname LTa\endcsname{\color{black}}\expandafter\def\csname LT0\endcsname{\color[rgb]{1,0,0}}\expandafter\def\csname LT1\endcsname{\color[rgb]{0,1,0}}\expandafter\def\csname LT2\endcsname{\color[rgb]{0,0,1}}\expandafter\def\csname LT3\endcsname{\color[rgb]{1,0,1}}\expandafter\def\csname LT4\endcsname{\color[rgb]{0,1,1}}\expandafter\def\csname LT5\endcsname{\color[rgb]{1,1,0}}\expandafter\def\csname LT6\endcsname{\color[rgb]{0,0,0}}\expandafter\def\csname LT7\endcsname{\color[rgb]{1,0.3,0}}\expandafter\def\csname LT8\endcsname{\color[rgb]{0.5,0.5,0.5}}\else
\def\colorrgb#1{\color{black}}\def\colorgray#1{\color[gray]{#1}}\expandafter\def\csname LTw\endcsname{\color{white}}\expandafter\def\csname LTb\endcsname{\color{black}}\expandafter\def\csname LTa\endcsname{\color{black}}\expandafter\def\csname LT0\endcsname{\color{black}}\expandafter\def\csname LT1\endcsname{\color{black}}\expandafter\def\csname LT2\endcsname{\color{black}}\expandafter\def\csname LT3\endcsname{\color{black}}\expandafter\def\csname LT4\endcsname{\color{black}}\expandafter\def\csname LT5\endcsname{\color{black}}\expandafter\def\csname LT6\endcsname{\color{black}}\expandafter\def\csname LT7\endcsname{\color{black}}\expandafter\def\csname LT8\endcsname{\color{black}}\fi
  \fi
    \setlength{\unitlength}{0.0500bp}\ifx\gptboxheight\undefined \newlength{\gptboxheight}\newlength{\gptboxwidth}\newsavebox{\gptboxtext}\fi \setlength{\fboxrule}{0.5pt}\setlength{\fboxsep}{1pt}\begin{picture}(4818.00,3174.00)\gplgaddtomacro\gplbacktext{\csname LTb\endcsname \put(1188,550){\makebox(0,0)[r]{\strut{}$10^{-1}$}}\put(1188,896){\makebox(0,0)[r]{\strut{}$10^{0}$}}\put(1188,1241){\makebox(0,0)[r]{\strut{}$10^{1}$}}\put(1188,1587){\makebox(0,0)[r]{\strut{}$10^{2}$}}\put(1188,1932){\makebox(0,0)[r]{\strut{}$10^{3}$}}\put(1188,2278){\makebox(0,0)[r]{\strut{}$10^{4}$}}\put(1188,2623){\makebox(0,0)[r]{\strut{}$10^{5}$}}\put(1377,330){\makebox(0,0){\strut{}$1$}}\put(2185,330){\makebox(0,0){\strut{}$10$}}\put(2993,330){\makebox(0,0){\strut{}$100$}}\put(3801,330){\makebox(0,0){\strut{}$1000$}}\put(4609,330){\makebox(0,0){\strut{}$10000$}}}\gplgaddtomacro\gplfronttext{\csname LTb\endcsname \put(3035,0){\makebox(0,0){\strut{}Number of Particles (log scale)}}\put(3035,2733){\makebox(0,0){\strut{}Outlier Accuracy}}}\gplbacktext
    \put(0,0){\includegraphics{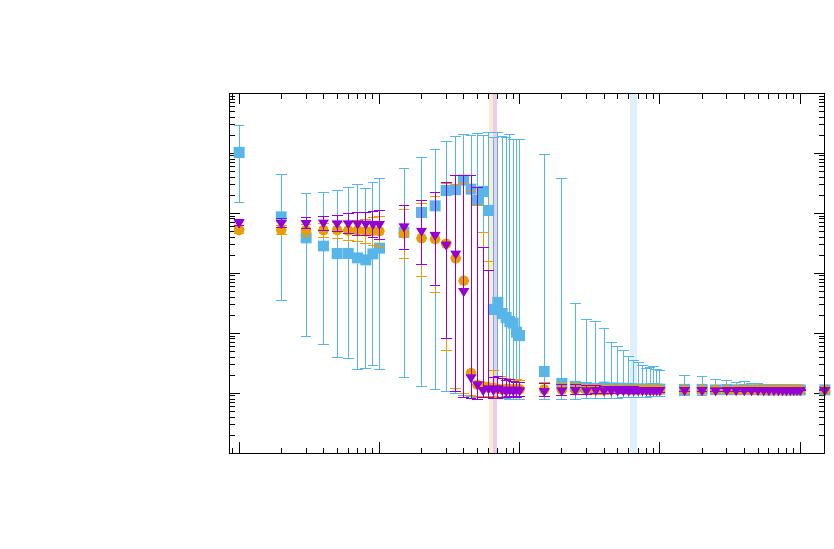}}\gplfronttext
  \end{picture}\endgroup
         \caption{Accuracy as a function of the number of particles.}
        \label{fig:accuracy1}
    \end{minipage}
    \hfill
    \begin{minipage}[t]{0.48\textwidth}
      \small\sf
        \hspace*{-1.8em}
        \begingroup
  \makeatletter
  \providecommand\color[2][]{\GenericError{(gnuplot) \space\space\space\@spaces}{Package color not loaded in conjunction with
      terminal option `colourtext'}{See the gnuplot documentation for explanation.}{Either use 'blacktext' in gnuplot or load the package
      color.sty in LaTeX.}\renewcommand\color[2][]{}}\providecommand\includegraphics[2][]{\GenericError{(gnuplot) \space\space\space\@spaces}{Package graphicx or graphics not loaded}{See the gnuplot documentation for explanation.}{The gnuplot epslatex terminal needs graphicx.sty or graphics.sty.}\renewcommand\includegraphics[2][]{}}\providecommand\rotatebox[2]{#2}\@ifundefined{ifGPcolor}{\newif\ifGPcolor
    \GPcolorfalse
  }{}\@ifundefined{ifGPblacktext}{\newif\ifGPblacktext
    \GPblacktexttrue
  }{}\let\gplgaddtomacro\g@addto@macro
\gdef\gplbacktext{}\gdef\gplfronttext{}\makeatother
  \ifGPblacktext
\def\colorrgb#1{}\def\colorgray#1{}\else
\ifGPcolor
      \def\colorrgb#1{\color[rgb]{#1}}\def\colorgray#1{\color[gray]{#1}}\expandafter\def\csname LTw\endcsname{\color{white}}\expandafter\def\csname LTb\endcsname{\color{black}}\expandafter\def\csname LTa\endcsname{\color{black}}\expandafter\def\csname LT0\endcsname{\color[rgb]{1,0,0}}\expandafter\def\csname LT1\endcsname{\color[rgb]{0,1,0}}\expandafter\def\csname LT2\endcsname{\color[rgb]{0,0,1}}\expandafter\def\csname LT3\endcsname{\color[rgb]{1,0,1}}\expandafter\def\csname LT4\endcsname{\color[rgb]{0,1,1}}\expandafter\def\csname LT5\endcsname{\color[rgb]{1,1,0}}\expandafter\def\csname LT6\endcsname{\color[rgb]{0,0,0}}\expandafter\def\csname LT7\endcsname{\color[rgb]{1,0.3,0}}\expandafter\def\csname LT8\endcsname{\color[rgb]{0.5,0.5,0.5}}\else
\def\colorrgb#1{\color{black}}\def\colorgray#1{\color[gray]{#1}}\expandafter\def\csname LTw\endcsname{\color{white}}\expandafter\def\csname LTb\endcsname{\color{black}}\expandafter\def\csname LTa\endcsname{\color{black}}\expandafter\def\csname LT0\endcsname{\color{black}}\expandafter\def\csname LT1\endcsname{\color{black}}\expandafter\def\csname LT2\endcsname{\color{black}}\expandafter\def\csname LT3\endcsname{\color{black}}\expandafter\def\csname LT4\endcsname{\color{black}}\expandafter\def\csname LT5\endcsname{\color{black}}\expandafter\def\csname LT6\endcsname{\color{black}}\expandafter\def\csname LT7\endcsname{\color{black}}\expandafter\def\csname LT8\endcsname{\color{black}}\fi
  \fi
    \setlength{\unitlength}{0.0500bp}\ifx\gptboxheight\undefined \newlength{\gptboxheight}\newlength{\gptboxwidth}\newsavebox{\gptboxtext}\fi \setlength{\fboxrule}{0.5pt}\setlength{\fboxsep}{1pt}\begin{picture}(4818.00,3174.00)\gplgaddtomacro\gplbacktext{\csname LTb\endcsname \put(1188,550){\makebox(0,0)[r]{\strut{}$10^{-1}$}}\put(1188,896){\makebox(0,0)[r]{\strut{}$10^{0}$}}\put(1188,1241){\makebox(0,0)[r]{\strut{}$10^{1}$}}\put(1188,1587){\makebox(0,0)[r]{\strut{}$10^{2}$}}\put(1188,1932){\makebox(0,0)[r]{\strut{}$10^{3}$}}\put(1188,2278){\makebox(0,0)[r]{\strut{}$10^{4}$}}\put(1188,2623){\makebox(0,0)[r]{\strut{}$10^{5}$}}\put(1377,330){\makebox(0,0){\strut{}$1$}}\put(2185,330){\makebox(0,0){\strut{}$10$}}\put(2993,330){\makebox(0,0){\strut{}$100$}}\put(3801,330){\makebox(0,0){\strut{}$1000$}}\put(4609,330){\makebox(0,0){\strut{}$10000$}}}\gplgaddtomacro\gplfronttext{\csname LTb\endcsname \put(3035,2733){\makebox(0,0){\strut{}Beta-Bernoulli Latency}}}\gplbacktext
    \put(0,0){\includegraphics{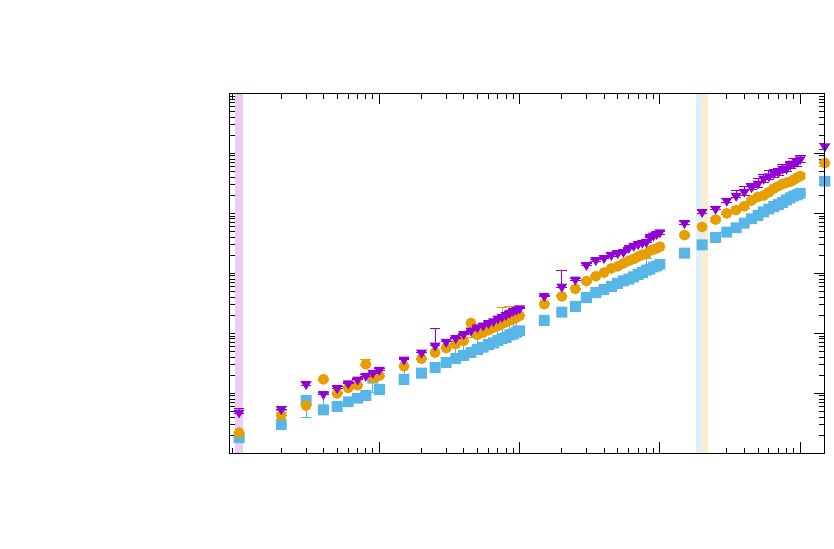}}\gplfronttext
  \end{picture}\endgroup
 \\[-0.3cm]
        \hspace*{-1.8em}
        \begingroup
  \makeatletter
  \providecommand\color[2][]{\GenericError{(gnuplot) \space\space\space\@spaces}{Package color not loaded in conjunction with
      terminal option `colourtext'}{See the gnuplot documentation for explanation.}{Either use 'blacktext' in gnuplot or load the package
      color.sty in LaTeX.}\renewcommand\color[2][]{}}\providecommand\includegraphics[2][]{\GenericError{(gnuplot) \space\space\space\@spaces}{Package graphicx or graphics not loaded}{See the gnuplot documentation for explanation.}{The gnuplot epslatex terminal needs graphicx.sty or graphics.sty.}\renewcommand\includegraphics[2][]{}}\providecommand\rotatebox[2]{#2}\@ifundefined{ifGPcolor}{\newif\ifGPcolor
    \GPcolorfalse
  }{}\@ifundefined{ifGPblacktext}{\newif\ifGPblacktext
    \GPblacktexttrue
  }{}\let\gplgaddtomacro\g@addto@macro
\gdef\gplbacktext{}\gdef\gplfronttext{}\makeatother
  \ifGPblacktext
\def\colorrgb#1{}\def\colorgray#1{}\else
\ifGPcolor
      \def\colorrgb#1{\color[rgb]{#1}}\def\colorgray#1{\color[gray]{#1}}\expandafter\def\csname LTw\endcsname{\color{white}}\expandafter\def\csname LTb\endcsname{\color{black}}\expandafter\def\csname LTa\endcsname{\color{black}}\expandafter\def\csname LT0\endcsname{\color[rgb]{1,0,0}}\expandafter\def\csname LT1\endcsname{\color[rgb]{0,1,0}}\expandafter\def\csname LT2\endcsname{\color[rgb]{0,0,1}}\expandafter\def\csname LT3\endcsname{\color[rgb]{1,0,1}}\expandafter\def\csname LT4\endcsname{\color[rgb]{0,1,1}}\expandafter\def\csname LT5\endcsname{\color[rgb]{1,1,0}}\expandafter\def\csname LT6\endcsname{\color[rgb]{0,0,0}}\expandafter\def\csname LT7\endcsname{\color[rgb]{1,0.3,0}}\expandafter\def\csname LT8\endcsname{\color[rgb]{0.5,0.5,0.5}}\else
\def\colorrgb#1{\color{black}}\def\colorgray#1{\color[gray]{#1}}\expandafter\def\csname LTw\endcsname{\color{white}}\expandafter\def\csname LTb\endcsname{\color{black}}\expandafter\def\csname LTa\endcsname{\color{black}}\expandafter\def\csname LT0\endcsname{\color{black}}\expandafter\def\csname LT1\endcsname{\color{black}}\expandafter\def\csname LT2\endcsname{\color{black}}\expandafter\def\csname LT3\endcsname{\color{black}}\expandafter\def\csname LT4\endcsname{\color{black}}\expandafter\def\csname LT5\endcsname{\color{black}}\expandafter\def\csname LT6\endcsname{\color{black}}\expandafter\def\csname LT7\endcsname{\color{black}}\expandafter\def\csname LT8\endcsname{\color{black}}\fi
  \fi
    \setlength{\unitlength}{0.0500bp}\ifx\gptboxheight\undefined \newlength{\gptboxheight}\newlength{\gptboxwidth}\newsavebox{\gptboxtext}\fi \setlength{\fboxrule}{0.5pt}\setlength{\fboxsep}{1pt}\begin{picture}(4818.00,3174.00)\gplgaddtomacro\gplbacktext{\csname LTb\endcsname \put(1188,550){\makebox(0,0)[r]{\strut{}$10^{-1}$}}\put(1188,896){\makebox(0,0)[r]{\strut{}$10^{0}$}}\put(1188,1241){\makebox(0,0)[r]{\strut{}$10^{1}$}}\put(1188,1587){\makebox(0,0)[r]{\strut{}$10^{2}$}}\put(1188,1932){\makebox(0,0)[r]{\strut{}$10^{3}$}}\put(1188,2278){\makebox(0,0)[r]{\strut{}$10^{4}$}}\put(1188,2623){\makebox(0,0)[r]{\strut{}$10^{5}$}}\put(1377,330){\makebox(0,0){\strut{}$1$}}\put(2185,330){\makebox(0,0){\strut{}$10$}}\put(2993,330){\makebox(0,0){\strut{}$100$}}\put(3801,330){\makebox(0,0){\strut{}$1000$}}\put(4609,330){\makebox(0,0){\strut{}$10000$}}}\gplgaddtomacro\gplfronttext{\csname LTb\endcsname \put(583,1586){\rotatebox{-270}{\makebox(0,0){\strut{}Execution time of 500 steps in ms (log scale)}}}\put(3035,2733){\makebox(0,0){\strut{}Gaussian-Gaussian Latency}}}\gplbacktext
    \put(0,0){\includegraphics{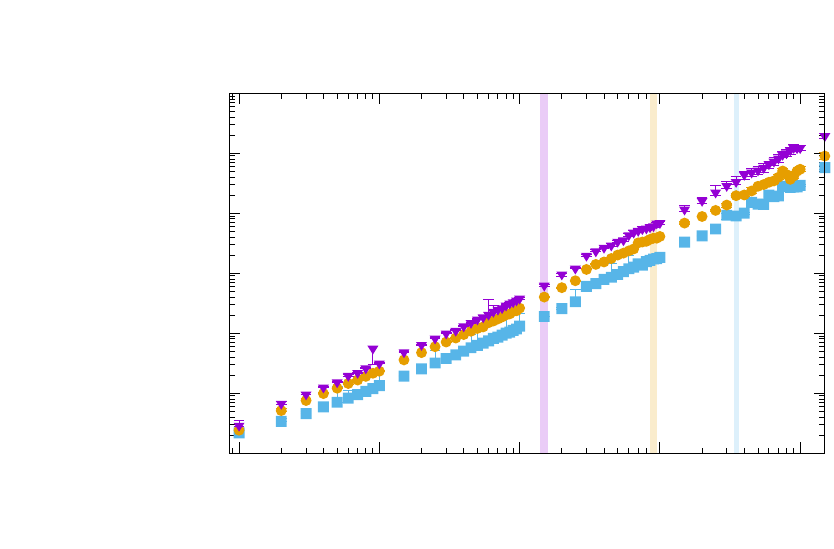}}\gplfronttext
  \end{picture}\endgroup
 \\[-0.3cm]
        \hspace*{-1.8em}
        \begingroup
  \makeatletter
  \providecommand\color[2][]{\GenericError{(gnuplot) \space\space\space\@spaces}{Package color not loaded in conjunction with
      terminal option `colourtext'}{See the gnuplot documentation for explanation.}{Either use 'blacktext' in gnuplot or load the package
      color.sty in LaTeX.}\renewcommand\color[2][]{}}\providecommand\includegraphics[2][]{\GenericError{(gnuplot) \space\space\space\@spaces}{Package graphicx or graphics not loaded}{See the gnuplot documentation for explanation.}{The gnuplot epslatex terminal needs graphicx.sty or graphics.sty.}\renewcommand\includegraphics[2][]{}}\providecommand\rotatebox[2]{#2}\@ifundefined{ifGPcolor}{\newif\ifGPcolor
    \GPcolorfalse
  }{}\@ifundefined{ifGPblacktext}{\newif\ifGPblacktext
    \GPblacktexttrue
  }{}\let\gplgaddtomacro\g@addto@macro
\gdef\gplbacktext{}\gdef\gplfronttext{}\makeatother
  \ifGPblacktext
\def\colorrgb#1{}\def\colorgray#1{}\else
\ifGPcolor
      \def\colorrgb#1{\color[rgb]{#1}}\def\colorgray#1{\color[gray]{#1}}\expandafter\def\csname LTw\endcsname{\color{white}}\expandafter\def\csname LTb\endcsname{\color{black}}\expandafter\def\csname LTa\endcsname{\color{black}}\expandafter\def\csname LT0\endcsname{\color[rgb]{1,0,0}}\expandafter\def\csname LT1\endcsname{\color[rgb]{0,1,0}}\expandafter\def\csname LT2\endcsname{\color[rgb]{0,0,1}}\expandafter\def\csname LT3\endcsname{\color[rgb]{1,0,1}}\expandafter\def\csname LT4\endcsname{\color[rgb]{0,1,1}}\expandafter\def\csname LT5\endcsname{\color[rgb]{1,1,0}}\expandafter\def\csname LT6\endcsname{\color[rgb]{0,0,0}}\expandafter\def\csname LT7\endcsname{\color[rgb]{1,0.3,0}}\expandafter\def\csname LT8\endcsname{\color[rgb]{0.5,0.5,0.5}}\else
\def\colorrgb#1{\color{black}}\def\colorgray#1{\color[gray]{#1}}\expandafter\def\csname LTw\endcsname{\color{white}}\expandafter\def\csname LTb\endcsname{\color{black}}\expandafter\def\csname LTa\endcsname{\color{black}}\expandafter\def\csname LT0\endcsname{\color{black}}\expandafter\def\csname LT1\endcsname{\color{black}}\expandafter\def\csname LT2\endcsname{\color{black}}\expandafter\def\csname LT3\endcsname{\color{black}}\expandafter\def\csname LT4\endcsname{\color{black}}\expandafter\def\csname LT5\endcsname{\color{black}}\expandafter\def\csname LT6\endcsname{\color{black}}\expandafter\def\csname LT7\endcsname{\color{black}}\expandafter\def\csname LT8\endcsname{\color{black}}\fi
  \fi
    \setlength{\unitlength}{0.0500bp}\ifx\gptboxheight\undefined \newlength{\gptboxheight}\newlength{\gptboxwidth}\newsavebox{\gptboxtext}\fi \setlength{\fboxrule}{0.5pt}\setlength{\fboxsep}{1pt}\begin{picture}(4818.00,3174.00)\gplgaddtomacro\gplbacktext{\csname LTb\endcsname \put(1188,550){\makebox(0,0)[r]{\strut{}$10^{-1}$}}\put(1188,896){\makebox(0,0)[r]{\strut{}$10^{0}$}}\put(1188,1241){\makebox(0,0)[r]{\strut{}$10^{1}$}}\put(1188,1587){\makebox(0,0)[r]{\strut{}$10^{2}$}}\put(1188,1932){\makebox(0,0)[r]{\strut{}$10^{3}$}}\put(1188,2278){\makebox(0,0)[r]{\strut{}$10^{4}$}}\put(1188,2623){\makebox(0,0)[r]{\strut{}$10^{5}$}}\put(1377,330){\makebox(0,0){\strut{}$1$}}\put(2185,330){\makebox(0,0){\strut{}$10$}}\put(2993,330){\makebox(0,0){\strut{}$100$}}\put(3801,330){\makebox(0,0){\strut{}$1000$}}\put(4609,330){\makebox(0,0){\strut{}$10000$}}}\gplgaddtomacro\gplfronttext{\csname LTb\endcsname \put(3035,2733){\makebox(0,0){\strut{}Kalman-1D Latency}}}\gplbacktext
    \put(0,0){\includegraphics{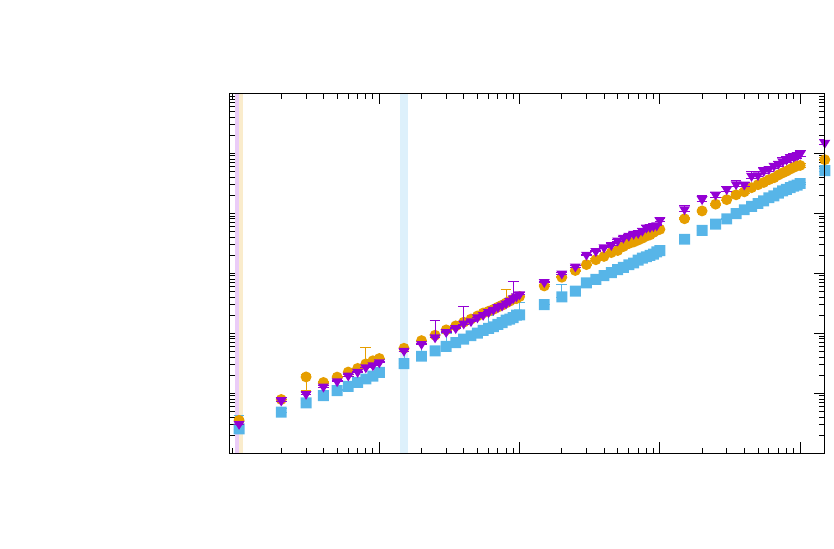}}\gplfronttext
  \end{picture}\endgroup
 \\[-0.3cm]
        \hspace*{-1.8em}
        \begingroup
  \makeatletter
  \providecommand\color[2][]{\GenericError{(gnuplot) \space\space\space\@spaces}{Package color not loaded in conjunction with
      terminal option `colourtext'}{See the gnuplot documentation for explanation.}{Either use 'blacktext' in gnuplot or load the package
      color.sty in LaTeX.}\renewcommand\color[2][]{}}\providecommand\includegraphics[2][]{\GenericError{(gnuplot) \space\space\space\@spaces}{Package graphicx or graphics not loaded}{See the gnuplot documentation for explanation.}{The gnuplot epslatex terminal needs graphicx.sty or graphics.sty.}\renewcommand\includegraphics[2][]{}}\providecommand\rotatebox[2]{#2}\@ifundefined{ifGPcolor}{\newif\ifGPcolor
    \GPcolorfalse
  }{}\@ifundefined{ifGPblacktext}{\newif\ifGPblacktext
    \GPblacktexttrue
  }{}\let\gplgaddtomacro\g@addto@macro
\gdef\gplbacktext{}\gdef\gplfronttext{}\makeatother
  \ifGPblacktext
\def\colorrgb#1{}\def\colorgray#1{}\else
\ifGPcolor
      \def\colorrgb#1{\color[rgb]{#1}}\def\colorgray#1{\color[gray]{#1}}\expandafter\def\csname LTw\endcsname{\color{white}}\expandafter\def\csname LTb\endcsname{\color{black}}\expandafter\def\csname LTa\endcsname{\color{black}}\expandafter\def\csname LT0\endcsname{\color[rgb]{1,0,0}}\expandafter\def\csname LT1\endcsname{\color[rgb]{0,1,0}}\expandafter\def\csname LT2\endcsname{\color[rgb]{0,0,1}}\expandafter\def\csname LT3\endcsname{\color[rgb]{1,0,1}}\expandafter\def\csname LT4\endcsname{\color[rgb]{0,1,1}}\expandafter\def\csname LT5\endcsname{\color[rgb]{1,1,0}}\expandafter\def\csname LT6\endcsname{\color[rgb]{0,0,0}}\expandafter\def\csname LT7\endcsname{\color[rgb]{1,0.3,0}}\expandafter\def\csname LT8\endcsname{\color[rgb]{0.5,0.5,0.5}}\else
\def\colorrgb#1{\color{black}}\def\colorgray#1{\color[gray]{#1}}\expandafter\def\csname LTw\endcsname{\color{white}}\expandafter\def\csname LTb\endcsname{\color{black}}\expandafter\def\csname LTa\endcsname{\color{black}}\expandafter\def\csname LT0\endcsname{\color{black}}\expandafter\def\csname LT1\endcsname{\color{black}}\expandafter\def\csname LT2\endcsname{\color{black}}\expandafter\def\csname LT3\endcsname{\color{black}}\expandafter\def\csname LT4\endcsname{\color{black}}\expandafter\def\csname LT5\endcsname{\color{black}}\expandafter\def\csname LT6\endcsname{\color{black}}\expandafter\def\csname LT7\endcsname{\color{black}}\expandafter\def\csname LT8\endcsname{\color{black}}\fi
  \fi
    \setlength{\unitlength}{0.0500bp}\ifx\gptboxheight\undefined \newlength{\gptboxheight}\newlength{\gptboxwidth}\newsavebox{\gptboxtext}\fi \setlength{\fboxrule}{0.5pt}\setlength{\fboxsep}{1pt}\begin{picture}(4818.00,3174.00)\gplgaddtomacro\gplbacktext{\csname LTb\endcsname \put(1188,550){\makebox(0,0)[r]{\strut{}$10^{-1}$}}\put(1188,896){\makebox(0,0)[r]{\strut{}$10^{0}$}}\put(1188,1241){\makebox(0,0)[r]{\strut{}$10^{1}$}}\put(1188,1587){\makebox(0,0)[r]{\strut{}$10^{2}$}}\put(1188,1932){\makebox(0,0)[r]{\strut{}$10^{3}$}}\put(1188,2278){\makebox(0,0)[r]{\strut{}$10^{4}$}}\put(1188,2623){\makebox(0,0)[r]{\strut{}$10^{5}$}}\put(1377,330){\makebox(0,0){\strut{}$1$}}\put(2185,330){\makebox(0,0){\strut{}$10$}}\put(2993,330){\makebox(0,0){\strut{}$100$}}\put(3801,330){\makebox(0,0){\strut{}$1000$}}\put(4609,330){\makebox(0,0){\strut{}$10000$}}}\gplgaddtomacro\gplfronttext{\csname LTb\endcsname \put(3035,0){\makebox(0,0){\strut{}Number of Particles (log scale)}}\put(3035,2733){\makebox(0,0){\strut{}Outlier Latency}}}\gplbacktext
    \put(0,0){\includegraphics{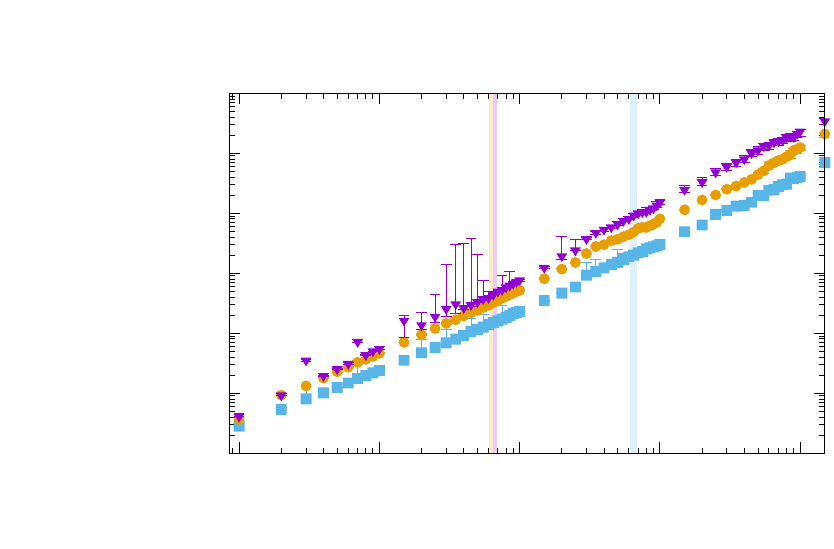}}\gplfronttext
  \end{picture}\endgroup
         \caption{Runtime performance as a function of particles.}
        \label{fig:perf-particles1}
    \end{minipage}
\end{figure*}

\begin{figure*}[p]
  \begin{center}
    \small\sf
    \pfbullet~PF $\quad$ \bdsbullet~BDS $\quad$ \pmdsbullet~SDS
  \end{center}
    \begin{minipage}[t]{0.48\textwidth}
      \small\sf
        \hspace*{-1.8em}
        \begingroup
  \makeatletter
  \providecommand\color[2][]{\GenericError{(gnuplot) \space\space\space\@spaces}{Package color not loaded in conjunction with
      terminal option `colourtext'}{See the gnuplot documentation for explanation.}{Either use 'blacktext' in gnuplot or load the package
      color.sty in LaTeX.}\renewcommand\color[2][]{}}\providecommand\includegraphics[2][]{\GenericError{(gnuplot) \space\space\space\@spaces}{Package graphicx or graphics not loaded}{See the gnuplot documentation for explanation.}{The gnuplot epslatex terminal needs graphicx.sty or graphics.sty.}\renewcommand\includegraphics[2][]{}}\providecommand\rotatebox[2]{#2}\@ifundefined{ifGPcolor}{\newif\ifGPcolor
    \GPcolorfalse
  }{}\@ifundefined{ifGPblacktext}{\newif\ifGPblacktext
    \GPblacktexttrue
  }{}\let\gplgaddtomacro\g@addto@macro
\gdef\gplbacktext{}\gdef\gplfronttext{}\makeatother
  \ifGPblacktext
\def\colorrgb#1{}\def\colorgray#1{}\else
\ifGPcolor
      \def\colorrgb#1{\color[rgb]{#1}}\def\colorgray#1{\color[gray]{#1}}\expandafter\def\csname LTw\endcsname{\color{white}}\expandafter\def\csname LTb\endcsname{\color{black}}\expandafter\def\csname LTa\endcsname{\color{black}}\expandafter\def\csname LT0\endcsname{\color[rgb]{1,0,0}}\expandafter\def\csname LT1\endcsname{\color[rgb]{0,1,0}}\expandafter\def\csname LT2\endcsname{\color[rgb]{0,0,1}}\expandafter\def\csname LT3\endcsname{\color[rgb]{1,0,1}}\expandafter\def\csname LT4\endcsname{\color[rgb]{0,1,1}}\expandafter\def\csname LT5\endcsname{\color[rgb]{1,1,0}}\expandafter\def\csname LT6\endcsname{\color[rgb]{0,0,0}}\expandafter\def\csname LT7\endcsname{\color[rgb]{1,0.3,0}}\expandafter\def\csname LT8\endcsname{\color[rgb]{0.5,0.5,0.5}}\else
\def\colorrgb#1{\color{black}}\def\colorgray#1{\color[gray]{#1}}\expandafter\def\csname LTw\endcsname{\color{white}}\expandafter\def\csname LTb\endcsname{\color{black}}\expandafter\def\csname LTa\endcsname{\color{black}}\expandafter\def\csname LT0\endcsname{\color{black}}\expandafter\def\csname LT1\endcsname{\color{black}}\expandafter\def\csname LT2\endcsname{\color{black}}\expandafter\def\csname LT3\endcsname{\color{black}}\expandafter\def\csname LT4\endcsname{\color{black}}\expandafter\def\csname LT5\endcsname{\color{black}}\expandafter\def\csname LT6\endcsname{\color{black}}\expandafter\def\csname LT7\endcsname{\color{black}}\expandafter\def\csname LT8\endcsname{\color{black}}\fi
  \fi
    \setlength{\unitlength}{0.0500bp}\ifx\gptboxheight\undefined \newlength{\gptboxheight}\newlength{\gptboxwidth}\newsavebox{\gptboxtext}\fi \setlength{\fboxrule}{0.5pt}\setlength{\fboxsep}{1pt}\begin{picture}(4818.00,3174.00)\gplgaddtomacro\gplbacktext{\csname LTb\endcsname \put(1188,550){\makebox(0,0)[r]{\strut{}$10^{1}$}}\put(1188,1241){\makebox(0,0)[r]{\strut{}$10^{2}$}}\put(1188,1932){\makebox(0,0)[r]{\strut{}$10^{3}$}}\put(1188,2623){\makebox(0,0)[r]{\strut{}$10^{4}$}}\put(1377,330){\makebox(0,0){\strut{}$1$}}\put(2185,330){\makebox(0,0){\strut{}$10$}}\put(2993,330){\makebox(0,0){\strut{}$100$}}\put(3801,330){\makebox(0,0){\strut{}$1000$}}\put(4609,330){\makebox(0,0){\strut{}$10000$}}}\gplgaddtomacro\gplfronttext{\csname LTb\endcsname \put(3035,2733){\makebox(0,0){\strut{}Robot Accuracy}}}\gplbacktext
    \put(0,0){\includegraphics{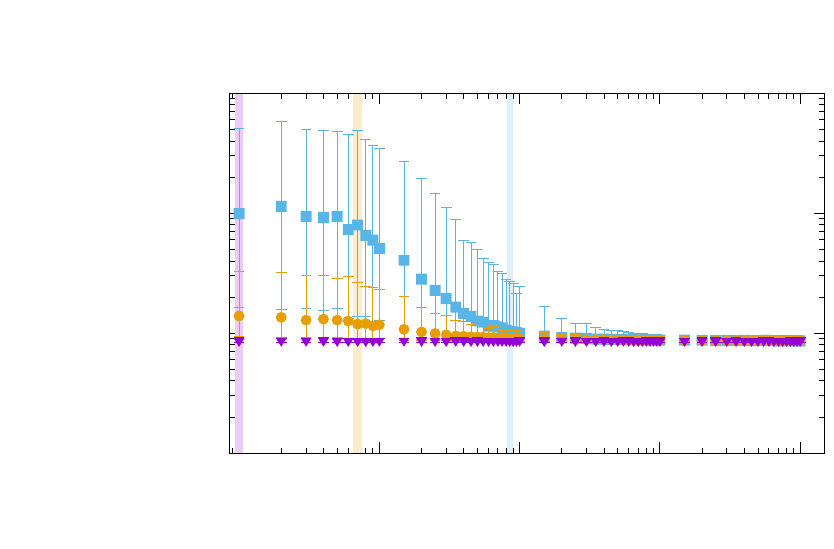}}\gplfronttext
  \end{picture}\endgroup
 \\[-0.3cm]
        \hspace*{-1.8em}
        \begingroup
  \makeatletter
  \providecommand\color[2][]{\GenericError{(gnuplot) \space\space\space\@spaces}{Package color not loaded in conjunction with
      terminal option `colourtext'}{See the gnuplot documentation for explanation.}{Either use 'blacktext' in gnuplot or load the package
      color.sty in LaTeX.}\renewcommand\color[2][]{}}\providecommand\includegraphics[2][]{\GenericError{(gnuplot) \space\space\space\@spaces}{Package graphicx or graphics not loaded}{See the gnuplot documentation for explanation.}{The gnuplot epslatex terminal needs graphicx.sty or graphics.sty.}\renewcommand\includegraphics[2][]{}}\providecommand\rotatebox[2]{#2}\@ifundefined{ifGPcolor}{\newif\ifGPcolor
    \GPcolorfalse
  }{}\@ifundefined{ifGPblacktext}{\newif\ifGPblacktext
    \GPblacktexttrue
  }{}\let\gplgaddtomacro\g@addto@macro
\gdef\gplbacktext{}\gdef\gplfronttext{}\makeatother
  \ifGPblacktext
\def\colorrgb#1{}\def\colorgray#1{}\else
\ifGPcolor
      \def\colorrgb#1{\color[rgb]{#1}}\def\colorgray#1{\color[gray]{#1}}\expandafter\def\csname LTw\endcsname{\color{white}}\expandafter\def\csname LTb\endcsname{\color{black}}\expandafter\def\csname LTa\endcsname{\color{black}}\expandafter\def\csname LT0\endcsname{\color[rgb]{1,0,0}}\expandafter\def\csname LT1\endcsname{\color[rgb]{0,1,0}}\expandafter\def\csname LT2\endcsname{\color[rgb]{0,0,1}}\expandafter\def\csname LT3\endcsname{\color[rgb]{1,0,1}}\expandafter\def\csname LT4\endcsname{\color[rgb]{0,1,1}}\expandafter\def\csname LT5\endcsname{\color[rgb]{1,1,0}}\expandafter\def\csname LT6\endcsname{\color[rgb]{0,0,0}}\expandafter\def\csname LT7\endcsname{\color[rgb]{1,0.3,0}}\expandafter\def\csname LT8\endcsname{\color[rgb]{0.5,0.5,0.5}}\else
\def\colorrgb#1{\color{black}}\def\colorgray#1{\color[gray]{#1}}\expandafter\def\csname LTw\endcsname{\color{white}}\expandafter\def\csname LTb\endcsname{\color{black}}\expandafter\def\csname LTa\endcsname{\color{black}}\expandafter\def\csname LT0\endcsname{\color{black}}\expandafter\def\csname LT1\endcsname{\color{black}}\expandafter\def\csname LT2\endcsname{\color{black}}\expandafter\def\csname LT3\endcsname{\color{black}}\expandafter\def\csname LT4\endcsname{\color{black}}\expandafter\def\csname LT5\endcsname{\color{black}}\expandafter\def\csname LT6\endcsname{\color{black}}\expandafter\def\csname LT7\endcsname{\color{black}}\expandafter\def\csname LT8\endcsname{\color{black}}\fi
  \fi
    \setlength{\unitlength}{0.0500bp}\ifx\gptboxheight\undefined \newlength{\gptboxheight}\newlength{\gptboxwidth}\newsavebox{\gptboxtext}\fi \setlength{\fboxrule}{0.5pt}\setlength{\fboxsep}{1pt}\begin{picture}(4818.00,3174.00)\gplgaddtomacro\gplbacktext{\csname LTb\endcsname \put(1188,550){\makebox(0,0)[r]{\strut{}$10^{-1}$}}\put(1188,1587){\makebox(0,0)[r]{\strut{}$10^{0}$}}\put(1188,2623){\makebox(0,0)[r]{\strut{}$10^{1}$}}\put(1377,330){\makebox(0,0){\strut{}$1$}}\put(2185,330){\makebox(0,0){\strut{}$10$}}\put(2993,330){\makebox(0,0){\strut{}$100$}}\put(3801,330){\makebox(0,0){\strut{}$1000$}}\put(4609,330){\makebox(0,0){\strut{}$10000$}}}\gplgaddtomacro\gplfronttext{\csname LTb\endcsname \put(583,1586){\rotatebox{-270}{\makebox(0,0){\strut{}Loss (log scale)}}}\put(3035,2733){\makebox(0,0){\strut{}SLAM Accuracy}}}\gplbacktext
    \put(0,0){\includegraphics{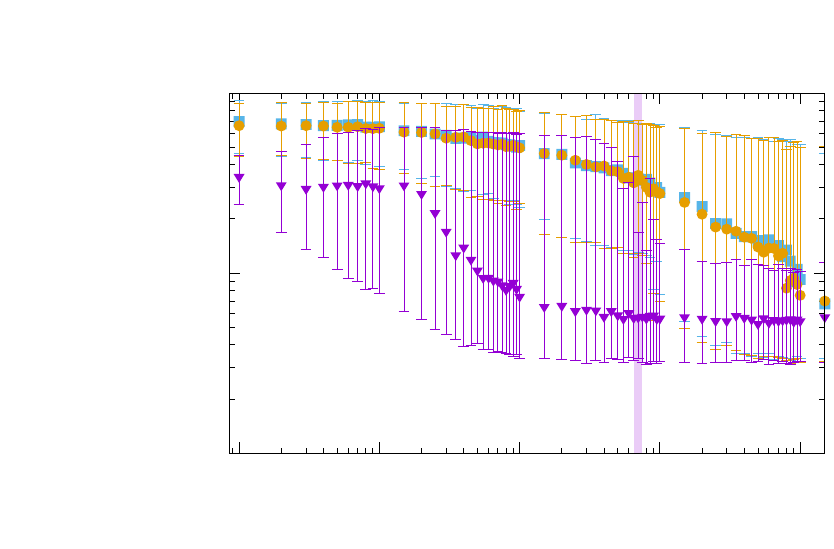}}\gplfronttext
  \end{picture}\endgroup
 \\[-0.3cm]
        \hspace*{-1.8em}
        \begingroup
  \makeatletter
  \providecommand\color[2][]{\GenericError{(gnuplot) \space\space\space\@spaces}{Package color not loaded in conjunction with
      terminal option `colourtext'}{See the gnuplot documentation for explanation.}{Either use 'blacktext' in gnuplot or load the package
      color.sty in LaTeX.}\renewcommand\color[2][]{}}\providecommand\includegraphics[2][]{\GenericError{(gnuplot) \space\space\space\@spaces}{Package graphicx or graphics not loaded}{See the gnuplot documentation for explanation.}{The gnuplot epslatex terminal needs graphicx.sty or graphics.sty.}\renewcommand\includegraphics[2][]{}}\providecommand\rotatebox[2]{#2}\@ifundefined{ifGPcolor}{\newif\ifGPcolor
    \GPcolorfalse
  }{}\@ifundefined{ifGPblacktext}{\newif\ifGPblacktext
    \GPblacktexttrue
  }{}\let\gplgaddtomacro\g@addto@macro
\gdef\gplbacktext{}\gdef\gplfronttext{}\makeatother
  \ifGPblacktext
\def\colorrgb#1{}\def\colorgray#1{}\else
\ifGPcolor
      \def\colorrgb#1{\color[rgb]{#1}}\def\colorgray#1{\color[gray]{#1}}\expandafter\def\csname LTw\endcsname{\color{white}}\expandafter\def\csname LTb\endcsname{\color{black}}\expandafter\def\csname LTa\endcsname{\color{black}}\expandafter\def\csname LT0\endcsname{\color[rgb]{1,0,0}}\expandafter\def\csname LT1\endcsname{\color[rgb]{0,1,0}}\expandafter\def\csname LT2\endcsname{\color[rgb]{0,0,1}}\expandafter\def\csname LT3\endcsname{\color[rgb]{1,0,1}}\expandafter\def\csname LT4\endcsname{\color[rgb]{0,1,1}}\expandafter\def\csname LT5\endcsname{\color[rgb]{1,1,0}}\expandafter\def\csname LT6\endcsname{\color[rgb]{0,0,0}}\expandafter\def\csname LT7\endcsname{\color[rgb]{1,0.3,0}}\expandafter\def\csname LT8\endcsname{\color[rgb]{0.5,0.5,0.5}}\else
\def\colorrgb#1{\color{black}}\def\colorgray#1{\color[gray]{#1}}\expandafter\def\csname LTw\endcsname{\color{white}}\expandafter\def\csname LTb\endcsname{\color{black}}\expandafter\def\csname LTa\endcsname{\color{black}}\expandafter\def\csname LT0\endcsname{\color{black}}\expandafter\def\csname LT1\endcsname{\color{black}}\expandafter\def\csname LT2\endcsname{\color{black}}\expandafter\def\csname LT3\endcsname{\color{black}}\expandafter\def\csname LT4\endcsname{\color{black}}\expandafter\def\csname LT5\endcsname{\color{black}}\expandafter\def\csname LT6\endcsname{\color{black}}\expandafter\def\csname LT7\endcsname{\color{black}}\expandafter\def\csname LT8\endcsname{\color{black}}\fi
  \fi
    \setlength{\unitlength}{0.0500bp}\ifx\gptboxheight\undefined \newlength{\gptboxheight}\newlength{\gptboxwidth}\newsavebox{\gptboxtext}\fi \setlength{\fboxrule}{0.5pt}\setlength{\fboxsep}{1pt}\begin{picture}(4818.00,3174.00)\gplgaddtomacro\gplbacktext{\csname LTb\endcsname \put(1188,550){\makebox(0,0)[r]{\strut{}$10^{0}$}}\put(1188,1241){\makebox(0,0)[r]{\strut{}$10^{1}$}}\put(1188,1932){\makebox(0,0)[r]{\strut{}$10^{2}$}}\put(1188,2623){\makebox(0,0)[r]{\strut{}$10^{3}$}}\put(1320,330){\makebox(0,0){\strut{}$1$}}\put(2142,330){\makebox(0,0){\strut{}$10$}}\put(2963,330){\makebox(0,0){\strut{}$100$}}\put(3785,330){\makebox(0,0){\strut{}$1000$}}\put(4606,330){\makebox(0,0){\strut{}$10000$}}}\gplgaddtomacro\gplfronttext{\csname LTb\endcsname \put(3035,0){\makebox(0,0){\strut{}Number of Particles (log scale)}}\put(3035,2733){\makebox(0,0){\strut{}MTT Accuracy}}}\gplbacktext
    \put(0,0){\includegraphics{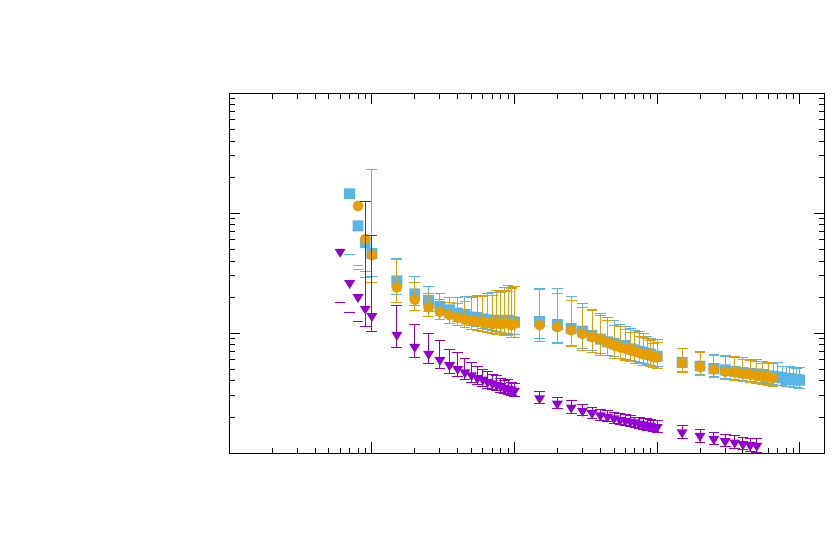}}\gplfronttext
  \end{picture}\endgroup
         \caption{Accuracy as a function of the number of particles.}
        \label{fig:accuracy2}
    \end{minipage}
    \hfill
    \begin{minipage}[t]{0.48\textwidth}
      \small\sf
        \begingroup
  \makeatletter
  \providecommand\color[2][]{\GenericError{(gnuplot) \space\space\space\@spaces}{Package color not loaded in conjunction with
      terminal option `colourtext'}{See the gnuplot documentation for explanation.}{Either use 'blacktext' in gnuplot or load the package
      color.sty in LaTeX.}\renewcommand\color[2][]{}}\providecommand\includegraphics[2][]{\GenericError{(gnuplot) \space\space\space\@spaces}{Package graphicx or graphics not loaded}{See the gnuplot documentation for explanation.}{The gnuplot epslatex terminal needs graphicx.sty or graphics.sty.}\renewcommand\includegraphics[2][]{}}\providecommand\rotatebox[2]{#2}\@ifundefined{ifGPcolor}{\newif\ifGPcolor
    \GPcolorfalse
  }{}\@ifundefined{ifGPblacktext}{\newif\ifGPblacktext
    \GPblacktexttrue
  }{}\let\gplgaddtomacro\g@addto@macro
\gdef\gplbacktext{}\gdef\gplfronttext{}\makeatother
  \ifGPblacktext
\def\colorrgb#1{}\def\colorgray#1{}\else
\ifGPcolor
      \def\colorrgb#1{\color[rgb]{#1}}\def\colorgray#1{\color[gray]{#1}}\expandafter\def\csname LTw\endcsname{\color{white}}\expandafter\def\csname LTb\endcsname{\color{black}}\expandafter\def\csname LTa\endcsname{\color{black}}\expandafter\def\csname LT0\endcsname{\color[rgb]{1,0,0}}\expandafter\def\csname LT1\endcsname{\color[rgb]{0,1,0}}\expandafter\def\csname LT2\endcsname{\color[rgb]{0,0,1}}\expandafter\def\csname LT3\endcsname{\color[rgb]{1,0,1}}\expandafter\def\csname LT4\endcsname{\color[rgb]{0,1,1}}\expandafter\def\csname LT5\endcsname{\color[rgb]{1,1,0}}\expandafter\def\csname LT6\endcsname{\color[rgb]{0,0,0}}\expandafter\def\csname LT7\endcsname{\color[rgb]{1,0.3,0}}\expandafter\def\csname LT8\endcsname{\color[rgb]{0.5,0.5,0.5}}\else
\def\colorrgb#1{\color{black}}\def\colorgray#1{\color[gray]{#1}}\expandafter\def\csname LTw\endcsname{\color{white}}\expandafter\def\csname LTb\endcsname{\color{black}}\expandafter\def\csname LTa\endcsname{\color{black}}\expandafter\def\csname LT0\endcsname{\color{black}}\expandafter\def\csname LT1\endcsname{\color{black}}\expandafter\def\csname LT2\endcsname{\color{black}}\expandafter\def\csname LT3\endcsname{\color{black}}\expandafter\def\csname LT4\endcsname{\color{black}}\expandafter\def\csname LT5\endcsname{\color{black}}\expandafter\def\csname LT6\endcsname{\color{black}}\expandafter\def\csname LT7\endcsname{\color{black}}\expandafter\def\csname LT8\endcsname{\color{black}}\fi
  \fi
    \setlength{\unitlength}{0.0500bp}\ifx\gptboxheight\undefined \newlength{\gptboxheight}\newlength{\gptboxwidth}\newsavebox{\gptboxtext}\fi \setlength{\fboxrule}{0.5pt}\setlength{\fboxsep}{1pt}\begin{picture}(4818.00,3174.00)\gplgaddtomacro\gplbacktext{\csname LTb\endcsname \put(1188,550){\makebox(0,0)[r]{\strut{}$10^{1}$}}\put(1188,965){\makebox(0,0)[r]{\strut{}$10^{2}$}}\put(1188,1379){\makebox(0,0)[r]{\strut{}$10^{3}$}}\put(1188,1794){\makebox(0,0)[r]{\strut{}$10^{4}$}}\put(1188,2208){\makebox(0,0)[r]{\strut{}$10^{5}$}}\put(1188,2623){\makebox(0,0)[r]{\strut{}$10^{6}$}}\put(1377,330){\makebox(0,0){\strut{}$1$}}\put(2185,330){\makebox(0,0){\strut{}$10$}}\put(2993,330){\makebox(0,0){\strut{}$100$}}\put(3801,330){\makebox(0,0){\strut{}$1000$}}\put(4609,330){\makebox(0,0){\strut{}$10000$}}}\gplgaddtomacro\gplfronttext{\csname LTb\endcsname \put(3035,2733){\makebox(0,0){\strut{}Robot Latency}}}\gplbacktext
    \put(0,0){\includegraphics{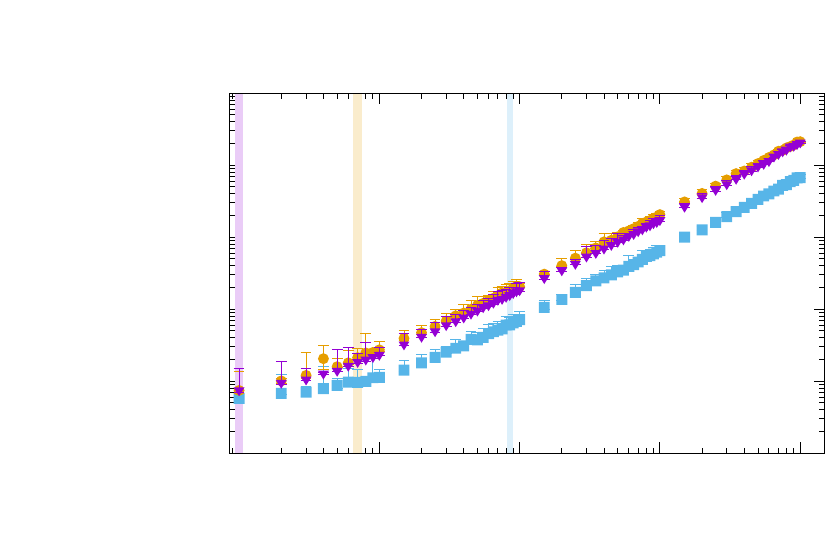}}\gplfronttext
  \end{picture}\endgroup
 \\[-0.3cm]
        \begingroup
  \makeatletter
  \providecommand\color[2][]{\GenericError{(gnuplot) \space\space\space\@spaces}{Package color not loaded in conjunction with
      terminal option `colourtext'}{See the gnuplot documentation for explanation.}{Either use 'blacktext' in gnuplot or load the package
      color.sty in LaTeX.}\renewcommand\color[2][]{}}\providecommand\includegraphics[2][]{\GenericError{(gnuplot) \space\space\space\@spaces}{Package graphicx or graphics not loaded}{See the gnuplot documentation for explanation.}{The gnuplot epslatex terminal needs graphicx.sty or graphics.sty.}\renewcommand\includegraphics[2][]{}}\providecommand\rotatebox[2]{#2}\@ifundefined{ifGPcolor}{\newif\ifGPcolor
    \GPcolorfalse
  }{}\@ifundefined{ifGPblacktext}{\newif\ifGPblacktext
    \GPblacktexttrue
  }{}\let\gplgaddtomacro\g@addto@macro
\gdef\gplbacktext{}\gdef\gplfronttext{}\makeatother
  \ifGPblacktext
\def\colorrgb#1{}\def\colorgray#1{}\else
\ifGPcolor
      \def\colorrgb#1{\color[rgb]{#1}}\def\colorgray#1{\color[gray]{#1}}\expandafter\def\csname LTw\endcsname{\color{white}}\expandafter\def\csname LTb\endcsname{\color{black}}\expandafter\def\csname LTa\endcsname{\color{black}}\expandafter\def\csname LT0\endcsname{\color[rgb]{1,0,0}}\expandafter\def\csname LT1\endcsname{\color[rgb]{0,1,0}}\expandafter\def\csname LT2\endcsname{\color[rgb]{0,0,1}}\expandafter\def\csname LT3\endcsname{\color[rgb]{1,0,1}}\expandafter\def\csname LT4\endcsname{\color[rgb]{0,1,1}}\expandafter\def\csname LT5\endcsname{\color[rgb]{1,1,0}}\expandafter\def\csname LT6\endcsname{\color[rgb]{0,0,0}}\expandafter\def\csname LT7\endcsname{\color[rgb]{1,0.3,0}}\expandafter\def\csname LT8\endcsname{\color[rgb]{0.5,0.5,0.5}}\else
\def\colorrgb#1{\color{black}}\def\colorgray#1{\color[gray]{#1}}\expandafter\def\csname LTw\endcsname{\color{white}}\expandafter\def\csname LTb\endcsname{\color{black}}\expandafter\def\csname LTa\endcsname{\color{black}}\expandafter\def\csname LT0\endcsname{\color{black}}\expandafter\def\csname LT1\endcsname{\color{black}}\expandafter\def\csname LT2\endcsname{\color{black}}\expandafter\def\csname LT3\endcsname{\color{black}}\expandafter\def\csname LT4\endcsname{\color{black}}\expandafter\def\csname LT5\endcsname{\color{black}}\expandafter\def\csname LT6\endcsname{\color{black}}\expandafter\def\csname LT7\endcsname{\color{black}}\expandafter\def\csname LT8\endcsname{\color{black}}\fi
  \fi
    \setlength{\unitlength}{0.0500bp}\ifx\gptboxheight\undefined \newlength{\gptboxheight}\newlength{\gptboxwidth}\newsavebox{\gptboxtext}\fi \setlength{\fboxrule}{0.5pt}\setlength{\fboxsep}{1pt}\begin{picture}(4818.00,3174.00)\gplgaddtomacro\gplbacktext{\csname LTb\endcsname \put(1188,550){\makebox(0,0)[r]{\strut{}$10^{-1}$}}\put(1188,846){\makebox(0,0)[r]{\strut{}$10^{0}$}}\put(1188,1142){\makebox(0,0)[r]{\strut{}$10^{1}$}}\put(1188,1438){\makebox(0,0)[r]{\strut{}$10^{2}$}}\put(1188,1735){\makebox(0,0)[r]{\strut{}$10^{3}$}}\put(1188,2031){\makebox(0,0)[r]{\strut{}$10^{4}$}}\put(1188,2327){\makebox(0,0)[r]{\strut{}$10^{5}$}}\put(1188,2623){\makebox(0,0)[r]{\strut{}$10^{6}$}}\put(1377,330){\makebox(0,0){\strut{}$1$}}\put(2185,330){\makebox(0,0){\strut{}$10$}}\put(2993,330){\makebox(0,0){\strut{}$100$}}\put(3801,330){\makebox(0,0){\strut{}$1000$}}\put(4609,330){\makebox(0,0){\strut{}$10000$}}}\gplgaddtomacro\gplfronttext{\csname LTb\endcsname \put(583,1586){\rotatebox{-270}{\makebox(0,0){\strut{}Execution time of 500 steps in ms (log scale)}}}\put(3035,2733){\makebox(0,0){\strut{}SLAM Latency}}}\gplbacktext
    \put(0,0){\includegraphics{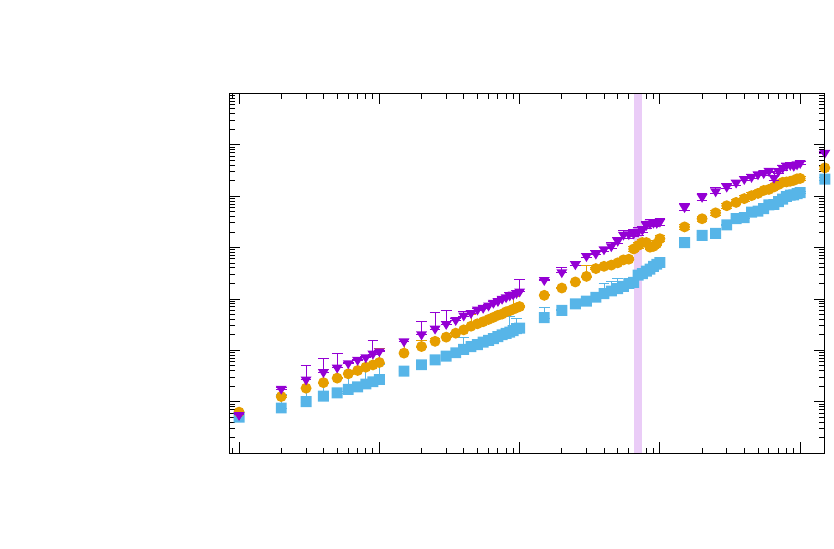}}\gplfronttext
  \end{picture}\endgroup
 \\[-0.3cm]
        \begingroup
  \makeatletter
  \providecommand\color[2][]{\GenericError{(gnuplot) \space\space\space\@spaces}{Package color not loaded in conjunction with
      terminal option `colourtext'}{See the gnuplot documentation for explanation.}{Either use 'blacktext' in gnuplot or load the package
      color.sty in LaTeX.}\renewcommand\color[2][]{}}\providecommand\includegraphics[2][]{\GenericError{(gnuplot) \space\space\space\@spaces}{Package graphicx or graphics not loaded}{See the gnuplot documentation for explanation.}{The gnuplot epslatex terminal needs graphicx.sty or graphics.sty.}\renewcommand\includegraphics[2][]{}}\providecommand\rotatebox[2]{#2}\@ifundefined{ifGPcolor}{\newif\ifGPcolor
    \GPcolorfalse
  }{}\@ifundefined{ifGPblacktext}{\newif\ifGPblacktext
    \GPblacktexttrue
  }{}\let\gplgaddtomacro\g@addto@macro
\gdef\gplbacktext{}\gdef\gplfronttext{}\makeatother
  \ifGPblacktext
\def\colorrgb#1{}\def\colorgray#1{}\else
\ifGPcolor
      \def\colorrgb#1{\color[rgb]{#1}}\def\colorgray#1{\color[gray]{#1}}\expandafter\def\csname LTw\endcsname{\color{white}}\expandafter\def\csname LTb\endcsname{\color{black}}\expandafter\def\csname LTa\endcsname{\color{black}}\expandafter\def\csname LT0\endcsname{\color[rgb]{1,0,0}}\expandafter\def\csname LT1\endcsname{\color[rgb]{0,1,0}}\expandafter\def\csname LT2\endcsname{\color[rgb]{0,0,1}}\expandafter\def\csname LT3\endcsname{\color[rgb]{1,0,1}}\expandafter\def\csname LT4\endcsname{\color[rgb]{0,1,1}}\expandafter\def\csname LT5\endcsname{\color[rgb]{1,1,0}}\expandafter\def\csname LT6\endcsname{\color[rgb]{0,0,0}}\expandafter\def\csname LT7\endcsname{\color[rgb]{1,0.3,0}}\expandafter\def\csname LT8\endcsname{\color[rgb]{0.5,0.5,0.5}}\else
\def\colorrgb#1{\color{black}}\def\colorgray#1{\color[gray]{#1}}\expandafter\def\csname LTw\endcsname{\color{white}}\expandafter\def\csname LTb\endcsname{\color{black}}\expandafter\def\csname LTa\endcsname{\color{black}}\expandafter\def\csname LT0\endcsname{\color{black}}\expandafter\def\csname LT1\endcsname{\color{black}}\expandafter\def\csname LT2\endcsname{\color{black}}\expandafter\def\csname LT3\endcsname{\color{black}}\expandafter\def\csname LT4\endcsname{\color{black}}\expandafter\def\csname LT5\endcsname{\color{black}}\expandafter\def\csname LT6\endcsname{\color{black}}\expandafter\def\csname LT7\endcsname{\color{black}}\expandafter\def\csname LT8\endcsname{\color{black}}\fi
  \fi
    \setlength{\unitlength}{0.0500bp}\ifx\gptboxheight\undefined \newlength{\gptboxheight}\newlength{\gptboxwidth}\newsavebox{\gptboxtext}\fi \setlength{\fboxrule}{0.5pt}\setlength{\fboxsep}{1pt}\begin{picture}(4818.00,3174.00)\gplgaddtomacro\gplbacktext{\csname LTb\endcsname \put(1188,550){\makebox(0,0)[r]{\strut{}$10^{1}$}}\put(1188,965){\makebox(0,0)[r]{\strut{}$10^{2}$}}\put(1188,1379){\makebox(0,0)[r]{\strut{}$10^{3}$}}\put(1188,1794){\makebox(0,0)[r]{\strut{}$10^{4}$}}\put(1188,2208){\makebox(0,0)[r]{\strut{}$10^{5}$}}\put(1188,2623){\makebox(0,0)[r]{\strut{}$10^{6}$}}\put(1320,330){\makebox(0,0){\strut{}$1$}}\put(2142,330){\makebox(0,0){\strut{}$10$}}\put(2963,330){\makebox(0,0){\strut{}$100$}}\put(3785,330){\makebox(0,0){\strut{}$1000$}}\put(4606,330){\makebox(0,0){\strut{}$10000$}}}\gplgaddtomacro\gplfronttext{\csname LTb\endcsname \put(3035,0){\makebox(0,0){\strut{}Number of Particles (log scale)}}\put(3035,2733){\makebox(0,0){\strut{}MTT Latency}}}\gplbacktext
    \put(0,0){\includegraphics{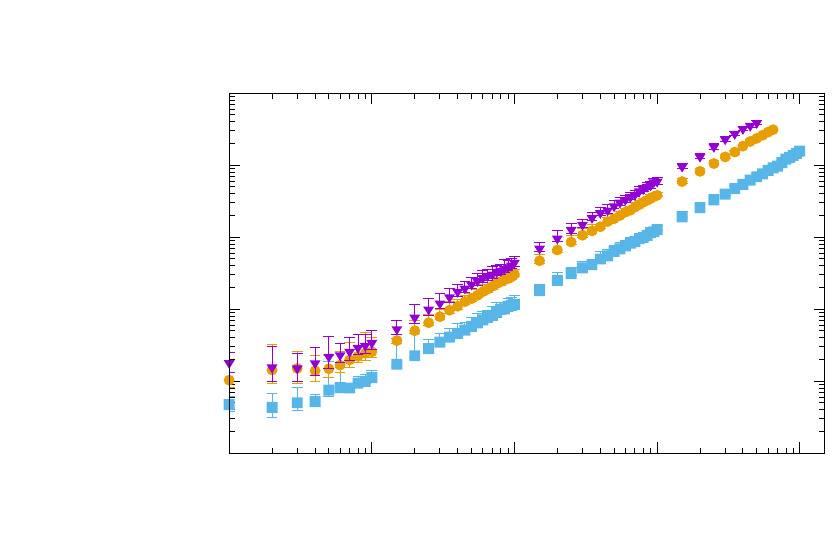}}\gplfronttext
  \end{picture}\endgroup
         \caption{Runtime performance as a function of particles.}
        \label{fig:perf-particles2}
    \end{minipage}
\end{figure*}

\begin{figure*}[p]
  \begin{center}
    \small\sf
    \pfbullet~PF $\quad$ \bdsbullet~BDS $\quad$ \pmdsbullet~SDS $\quad$ \ndsbullet~DS
  \end{center}
    \begin{minipage}[t]{0.48\textwidth}
      \small\sf
        \hspace*{-1.8em}
        \begingroup
  \makeatletter
  \providecommand\color[2][]{\GenericError{(gnuplot) \space\space\space\@spaces}{Package color not loaded in conjunction with
      terminal option `colourtext'}{See the gnuplot documentation for explanation.}{Either use 'blacktext' in gnuplot or load the package
      color.sty in LaTeX.}\renewcommand\color[2][]{}}\providecommand\includegraphics[2][]{\GenericError{(gnuplot) \space\space\space\@spaces}{Package graphicx or graphics not loaded}{See the gnuplot documentation for explanation.}{The gnuplot epslatex terminal needs graphicx.sty or graphics.sty.}\renewcommand\includegraphics[2][]{}}\providecommand\rotatebox[2]{#2}\@ifundefined{ifGPcolor}{\newif\ifGPcolor
    \GPcolorfalse
  }{}\@ifundefined{ifGPblacktext}{\newif\ifGPblacktext
    \GPblacktexttrue
  }{}\let\gplgaddtomacro\g@addto@macro
\gdef\gplbacktext{}\gdef\gplfronttext{}\makeatother
  \ifGPblacktext
\def\colorrgb#1{}\def\colorgray#1{}\else
\ifGPcolor
      \def\colorrgb#1{\color[rgb]{#1}}\def\colorgray#1{\color[gray]{#1}}\expandafter\def\csname LTw\endcsname{\color{white}}\expandafter\def\csname LTb\endcsname{\color{black}}\expandafter\def\csname LTa\endcsname{\color{black}}\expandafter\def\csname LT0\endcsname{\color[rgb]{1,0,0}}\expandafter\def\csname LT1\endcsname{\color[rgb]{0,1,0}}\expandafter\def\csname LT2\endcsname{\color[rgb]{0,0,1}}\expandafter\def\csname LT3\endcsname{\color[rgb]{1,0,1}}\expandafter\def\csname LT4\endcsname{\color[rgb]{0,1,1}}\expandafter\def\csname LT5\endcsname{\color[rgb]{1,1,0}}\expandafter\def\csname LT6\endcsname{\color[rgb]{0,0,0}}\expandafter\def\csname LT7\endcsname{\color[rgb]{1,0.3,0}}\expandafter\def\csname LT8\endcsname{\color[rgb]{0.5,0.5,0.5}}\else
\def\colorrgb#1{\color{black}}\def\colorgray#1{\color[gray]{#1}}\expandafter\def\csname LTw\endcsname{\color{white}}\expandafter\def\csname LTb\endcsname{\color{black}}\expandafter\def\csname LTa\endcsname{\color{black}}\expandafter\def\csname LT0\endcsname{\color{black}}\expandafter\def\csname LT1\endcsname{\color{black}}\expandafter\def\csname LT2\endcsname{\color{black}}\expandafter\def\csname LT3\endcsname{\color{black}}\expandafter\def\csname LT4\endcsname{\color{black}}\expandafter\def\csname LT5\endcsname{\color{black}}\expandafter\def\csname LT6\endcsname{\color{black}}\expandafter\def\csname LT7\endcsname{\color{black}}\expandafter\def\csname LT8\endcsname{\color{black}}\fi
  \fi
    \setlength{\unitlength}{0.0500bp}\ifx\gptboxheight\undefined \newlength{\gptboxheight}\newlength{\gptboxwidth}\newsavebox{\gptboxtext}\fi \setlength{\fboxrule}{0.5pt}\setlength{\fboxsep}{1pt}\begin{picture}(4818.00,3174.00)\gplgaddtomacro\gplbacktext{\csname LTb\endcsname \put(1188,550){\makebox(0,0)[r]{\strut{}$10^{-2}$}}\put(1188,1587){\makebox(0,0)[r]{\strut{}$10^{-1}$}}\put(1188,2623){\makebox(0,0)[r]{\strut{}$10^{0}$}}\put(1320,330){\makebox(0,0){\strut{}$0$}}\put(1663,330){\makebox(0,0){\strut{}$50$}}\put(2006,330){\makebox(0,0){\strut{}$100$}}\put(2349,330){\makebox(0,0){\strut{}$150$}}\put(2692,330){\makebox(0,0){\strut{}$200$}}\put(3036,330){\makebox(0,0){\strut{}$250$}}\put(3379,330){\makebox(0,0){\strut{}$300$}}\put(3722,330){\makebox(0,0){\strut{}$350$}}\put(4065,330){\makebox(0,0){\strut{}$400$}}\put(4408,330){\makebox(0,0){\strut{}$450$}}\put(4751,330){\makebox(0,0){\strut{}$500$}}}\gplgaddtomacro\gplfronttext{\csname LTb\endcsname \put(3035,2733){\makebox(0,0){\strut{}Beta-Bernoulli Latency}}}\gplbacktext
    \put(0,0){\includegraphics{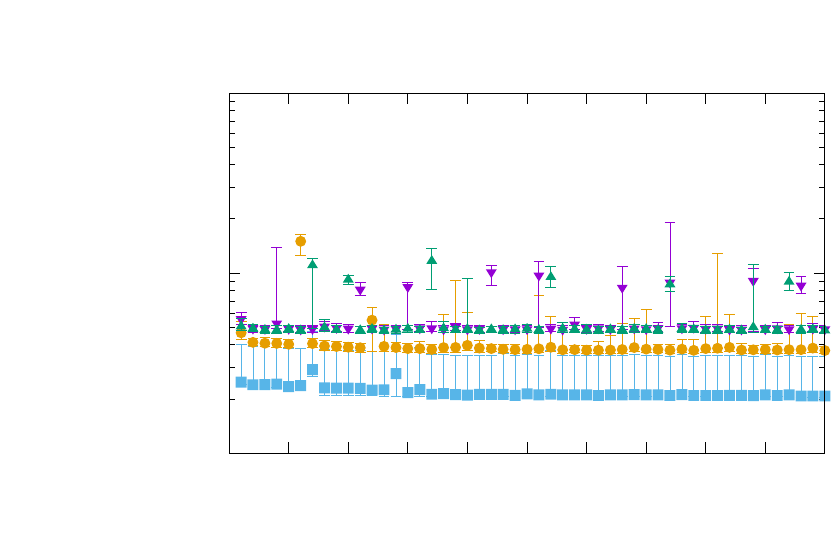}}\gplfronttext
  \end{picture}\endgroup
         \\[-0.3cm]
        \hspace*{-1.8em}
        \begingroup
  \makeatletter
  \providecommand\color[2][]{\GenericError{(gnuplot) \space\space\space\@spaces}{Package color not loaded in conjunction with
      terminal option `colourtext'}{See the gnuplot documentation for explanation.}{Either use 'blacktext' in gnuplot or load the package
      color.sty in LaTeX.}\renewcommand\color[2][]{}}\providecommand\includegraphics[2][]{\GenericError{(gnuplot) \space\space\space\@spaces}{Package graphicx or graphics not loaded}{See the gnuplot documentation for explanation.}{The gnuplot epslatex terminal needs graphicx.sty or graphics.sty.}\renewcommand\includegraphics[2][]{}}\providecommand\rotatebox[2]{#2}\@ifundefined{ifGPcolor}{\newif\ifGPcolor
    \GPcolorfalse
  }{}\@ifundefined{ifGPblacktext}{\newif\ifGPblacktext
    \GPblacktexttrue
  }{}\let\gplgaddtomacro\g@addto@macro
\gdef\gplbacktext{}\gdef\gplfronttext{}\makeatother
  \ifGPblacktext
\def\colorrgb#1{}\def\colorgray#1{}\else
\ifGPcolor
      \def\colorrgb#1{\color[rgb]{#1}}\def\colorgray#1{\color[gray]{#1}}\expandafter\def\csname LTw\endcsname{\color{white}}\expandafter\def\csname LTb\endcsname{\color{black}}\expandafter\def\csname LTa\endcsname{\color{black}}\expandafter\def\csname LT0\endcsname{\color[rgb]{1,0,0}}\expandafter\def\csname LT1\endcsname{\color[rgb]{0,1,0}}\expandafter\def\csname LT2\endcsname{\color[rgb]{0,0,1}}\expandafter\def\csname LT3\endcsname{\color[rgb]{1,0,1}}\expandafter\def\csname LT4\endcsname{\color[rgb]{0,1,1}}\expandafter\def\csname LT5\endcsname{\color[rgb]{1,1,0}}\expandafter\def\csname LT6\endcsname{\color[rgb]{0,0,0}}\expandafter\def\csname LT7\endcsname{\color[rgb]{1,0.3,0}}\expandafter\def\csname LT8\endcsname{\color[rgb]{0.5,0.5,0.5}}\else
\def\colorrgb#1{\color{black}}\def\colorgray#1{\color[gray]{#1}}\expandafter\def\csname LTw\endcsname{\color{white}}\expandafter\def\csname LTb\endcsname{\color{black}}\expandafter\def\csname LTa\endcsname{\color{black}}\expandafter\def\csname LT0\endcsname{\color{black}}\expandafter\def\csname LT1\endcsname{\color{black}}\expandafter\def\csname LT2\endcsname{\color{black}}\expandafter\def\csname LT3\endcsname{\color{black}}\expandafter\def\csname LT4\endcsname{\color{black}}\expandafter\def\csname LT5\endcsname{\color{black}}\expandafter\def\csname LT6\endcsname{\color{black}}\expandafter\def\csname LT7\endcsname{\color{black}}\expandafter\def\csname LT8\endcsname{\color{black}}\fi
  \fi
    \setlength{\unitlength}{0.0500bp}\ifx\gptboxheight\undefined \newlength{\gptboxheight}\newlength{\gptboxwidth}\newsavebox{\gptboxtext}\fi \setlength{\fboxrule}{0.5pt}\setlength{\fboxsep}{1pt}\begin{picture}(4818.00,3174.00)\gplgaddtomacro\gplbacktext{\csname LTb\endcsname \put(1188,550){\makebox(0,0)[r]{\strut{}$10^{-2}$}}\put(1188,1587){\makebox(0,0)[r]{\strut{}$10^{-1}$}}\put(1188,2623){\makebox(0,0)[r]{\strut{}$10^{0}$}}\put(1320,330){\makebox(0,0){\strut{}$0$}}\put(1663,330){\makebox(0,0){\strut{}$50$}}\put(2006,330){\makebox(0,0){\strut{}$100$}}\put(2349,330){\makebox(0,0){\strut{}$150$}}\put(2692,330){\makebox(0,0){\strut{}$200$}}\put(3036,330){\makebox(0,0){\strut{}$250$}}\put(3379,330){\makebox(0,0){\strut{}$300$}}\put(3722,330){\makebox(0,0){\strut{}$350$}}\put(4065,330){\makebox(0,0){\strut{}$400$}}\put(4408,330){\makebox(0,0){\strut{}$450$}}\put(4751,330){\makebox(0,0){\strut{}$500$}}}\gplgaddtomacro\gplfronttext{\csname LTb\endcsname \put(583,1586){\rotatebox{-270}{\makebox(0,0){\strut{}Step latency in ms (log scale)}}}\put(3035,2733){\makebox(0,0){\strut{}Gaussian-Gaussian Latency}}}\gplbacktext
    \put(0,0){\includegraphics{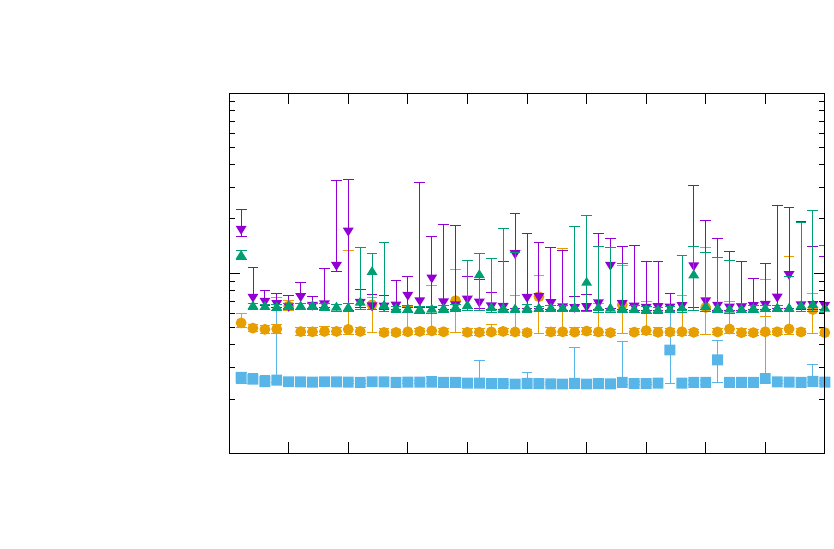}}\gplfronttext
  \end{picture}\endgroup
         \\[-0.3cm]
        \hspace*{-1.8em}
        \begingroup
  \makeatletter
  \providecommand\color[2][]{\GenericError{(gnuplot) \space\space\space\@spaces}{Package color not loaded in conjunction with
      terminal option `colourtext'}{See the gnuplot documentation for explanation.}{Either use 'blacktext' in gnuplot or load the package
      color.sty in LaTeX.}\renewcommand\color[2][]{}}\providecommand\includegraphics[2][]{\GenericError{(gnuplot) \space\space\space\@spaces}{Package graphicx or graphics not loaded}{See the gnuplot documentation for explanation.}{The gnuplot epslatex terminal needs graphicx.sty or graphics.sty.}\renewcommand\includegraphics[2][]{}}\providecommand\rotatebox[2]{#2}\@ifundefined{ifGPcolor}{\newif\ifGPcolor
    \GPcolorfalse
  }{}\@ifundefined{ifGPblacktext}{\newif\ifGPblacktext
    \GPblacktexttrue
  }{}\let\gplgaddtomacro\g@addto@macro
\gdef\gplbacktext{}\gdef\gplfronttext{}\makeatother
  \ifGPblacktext
\def\colorrgb#1{}\def\colorgray#1{}\else
\ifGPcolor
      \def\colorrgb#1{\color[rgb]{#1}}\def\colorgray#1{\color[gray]{#1}}\expandafter\def\csname LTw\endcsname{\color{white}}\expandafter\def\csname LTb\endcsname{\color{black}}\expandafter\def\csname LTa\endcsname{\color{black}}\expandafter\def\csname LT0\endcsname{\color[rgb]{1,0,0}}\expandafter\def\csname LT1\endcsname{\color[rgb]{0,1,0}}\expandafter\def\csname LT2\endcsname{\color[rgb]{0,0,1}}\expandafter\def\csname LT3\endcsname{\color[rgb]{1,0,1}}\expandafter\def\csname LT4\endcsname{\color[rgb]{0,1,1}}\expandafter\def\csname LT5\endcsname{\color[rgb]{1,1,0}}\expandafter\def\csname LT6\endcsname{\color[rgb]{0,0,0}}\expandafter\def\csname LT7\endcsname{\color[rgb]{1,0.3,0}}\expandafter\def\csname LT8\endcsname{\color[rgb]{0.5,0.5,0.5}}\else
\def\colorrgb#1{\color{black}}\def\colorgray#1{\color[gray]{#1}}\expandafter\def\csname LTw\endcsname{\color{white}}\expandafter\def\csname LTb\endcsname{\color{black}}\expandafter\def\csname LTa\endcsname{\color{black}}\expandafter\def\csname LT0\endcsname{\color{black}}\expandafter\def\csname LT1\endcsname{\color{black}}\expandafter\def\csname LT2\endcsname{\color{black}}\expandafter\def\csname LT3\endcsname{\color{black}}\expandafter\def\csname LT4\endcsname{\color{black}}\expandafter\def\csname LT5\endcsname{\color{black}}\expandafter\def\csname LT6\endcsname{\color{black}}\expandafter\def\csname LT7\endcsname{\color{black}}\expandafter\def\csname LT8\endcsname{\color{black}}\fi
  \fi
    \setlength{\unitlength}{0.0500bp}\ifx\gptboxheight\undefined \newlength{\gptboxheight}\newlength{\gptboxwidth}\newsavebox{\gptboxtext}\fi \setlength{\fboxrule}{0.5pt}\setlength{\fboxsep}{1pt}\begin{picture}(4818.00,3174.00)\gplgaddtomacro\gplbacktext{\csname LTb\endcsname \put(1188,550){\makebox(0,0)[r]{\strut{}$10^{-2}$}}\put(1188,1068){\makebox(0,0)[r]{\strut{}$10^{-1}$}}\put(1188,1587){\makebox(0,0)[r]{\strut{}$10^{0}$}}\put(1188,2105){\makebox(0,0)[r]{\strut{}$10^{1}$}}\put(1188,2623){\makebox(0,0)[r]{\strut{}$10^{2}$}}\put(1320,330){\makebox(0,0){\strut{}$0$}}\put(1663,330){\makebox(0,0){\strut{}$50$}}\put(2006,330){\makebox(0,0){\strut{}$100$}}\put(2349,330){\makebox(0,0){\strut{}$150$}}\put(2692,330){\makebox(0,0){\strut{}$200$}}\put(3036,330){\makebox(0,0){\strut{}$250$}}\put(3379,330){\makebox(0,0){\strut{}$300$}}\put(3722,330){\makebox(0,0){\strut{}$350$}}\put(4065,330){\makebox(0,0){\strut{}$400$}}\put(4408,330){\makebox(0,0){\strut{}$450$}}\put(4751,330){\makebox(0,0){\strut{}$500$}}}\gplgaddtomacro\gplfronttext{\csname LTb\endcsname \put(3035,2733){\makebox(0,0){\strut{}Kalman-1D Latency}}}\gplbacktext
    \put(0,0){\includegraphics{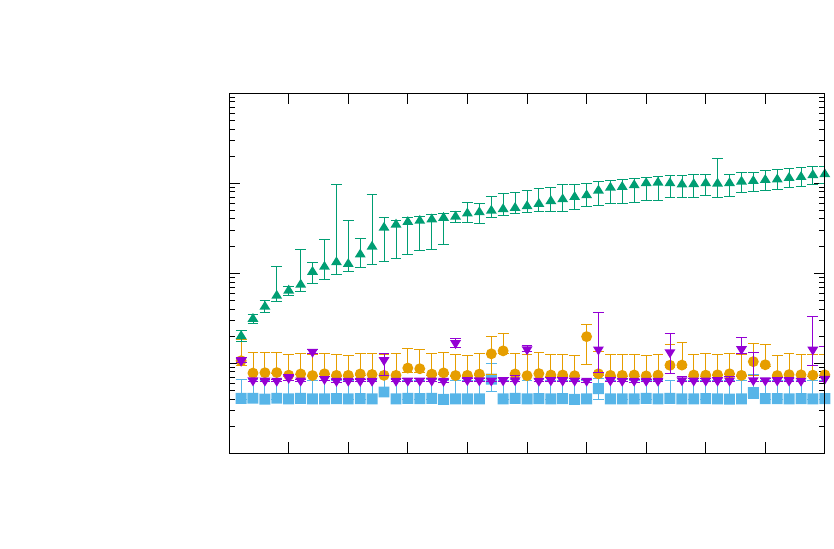}}\gplfronttext
  \end{picture}\endgroup
         \\[-0.3cm]
        \hspace*{-1.8em}
        \begingroup
  \makeatletter
  \providecommand\color[2][]{\GenericError{(gnuplot) \space\space\space\@spaces}{Package color not loaded in conjunction with
      terminal option `colourtext'}{See the gnuplot documentation for explanation.}{Either use 'blacktext' in gnuplot or load the package
      color.sty in LaTeX.}\renewcommand\color[2][]{}}\providecommand\includegraphics[2][]{\GenericError{(gnuplot) \space\space\space\@spaces}{Package graphicx or graphics not loaded}{See the gnuplot documentation for explanation.}{The gnuplot epslatex terminal needs graphicx.sty or graphics.sty.}\renewcommand\includegraphics[2][]{}}\providecommand\rotatebox[2]{#2}\@ifundefined{ifGPcolor}{\newif\ifGPcolor
    \GPcolorfalse
  }{}\@ifundefined{ifGPblacktext}{\newif\ifGPblacktext
    \GPblacktexttrue
  }{}\let\gplgaddtomacro\g@addto@macro
\gdef\gplbacktext{}\gdef\gplfronttext{}\makeatother
  \ifGPblacktext
\def\colorrgb#1{}\def\colorgray#1{}\else
\ifGPcolor
      \def\colorrgb#1{\color[rgb]{#1}}\def\colorgray#1{\color[gray]{#1}}\expandafter\def\csname LTw\endcsname{\color{white}}\expandafter\def\csname LTb\endcsname{\color{black}}\expandafter\def\csname LTa\endcsname{\color{black}}\expandafter\def\csname LT0\endcsname{\color[rgb]{1,0,0}}\expandafter\def\csname LT1\endcsname{\color[rgb]{0,1,0}}\expandafter\def\csname LT2\endcsname{\color[rgb]{0,0,1}}\expandafter\def\csname LT3\endcsname{\color[rgb]{1,0,1}}\expandafter\def\csname LT4\endcsname{\color[rgb]{0,1,1}}\expandafter\def\csname LT5\endcsname{\color[rgb]{1,1,0}}\expandafter\def\csname LT6\endcsname{\color[rgb]{0,0,0}}\expandafter\def\csname LT7\endcsname{\color[rgb]{1,0.3,0}}\expandafter\def\csname LT8\endcsname{\color[rgb]{0.5,0.5,0.5}}\else
\def\colorrgb#1{\color{black}}\def\colorgray#1{\color[gray]{#1}}\expandafter\def\csname LTw\endcsname{\color{white}}\expandafter\def\csname LTb\endcsname{\color{black}}\expandafter\def\csname LTa\endcsname{\color{black}}\expandafter\def\csname LT0\endcsname{\color{black}}\expandafter\def\csname LT1\endcsname{\color{black}}\expandafter\def\csname LT2\endcsname{\color{black}}\expandafter\def\csname LT3\endcsname{\color{black}}\expandafter\def\csname LT4\endcsname{\color{black}}\expandafter\def\csname LT5\endcsname{\color{black}}\expandafter\def\csname LT6\endcsname{\color{black}}\expandafter\def\csname LT7\endcsname{\color{black}}\expandafter\def\csname LT8\endcsname{\color{black}}\fi
  \fi
    \setlength{\unitlength}{0.0500bp}\ifx\gptboxheight\undefined \newlength{\gptboxheight}\newlength{\gptboxwidth}\newsavebox{\gptboxtext}\fi \setlength{\fboxrule}{0.5pt}\setlength{\fboxsep}{1pt}\begin{picture}(4818.00,3174.00)\gplgaddtomacro\gplbacktext{\csname LTb\endcsname \put(1188,550){\makebox(0,0)[r]{\strut{}$10^{-2}$}}\put(1188,1068){\makebox(0,0)[r]{\strut{}$10^{-1}$}}\put(1188,1587){\makebox(0,0)[r]{\strut{}$10^{0}$}}\put(1188,2105){\makebox(0,0)[r]{\strut{}$10^{1}$}}\put(1188,2623){\makebox(0,0)[r]{\strut{}$10^{2}$}}\put(1320,330){\makebox(0,0){\strut{}$0$}}\put(1663,330){\makebox(0,0){\strut{}$50$}}\put(2006,330){\makebox(0,0){\strut{}$100$}}\put(2349,330){\makebox(0,0){\strut{}$150$}}\put(2692,330){\makebox(0,0){\strut{}$200$}}\put(3036,330){\makebox(0,0){\strut{}$250$}}\put(3379,330){\makebox(0,0){\strut{}$300$}}\put(3722,330){\makebox(0,0){\strut{}$350$}}\put(4065,330){\makebox(0,0){\strut{}$400$}}\put(4408,330){\makebox(0,0){\strut{}$450$}}\put(4751,330){\makebox(0,0){\strut{}$500$}}}\gplgaddtomacro\gplfronttext{\csname LTb\endcsname \put(3035,0){\makebox(0,0){\strut{}Step}}\put(3035,2733){\makebox(0,0){\strut{}Outlier Latency}}}\gplbacktext
    \put(0,0){\includegraphics{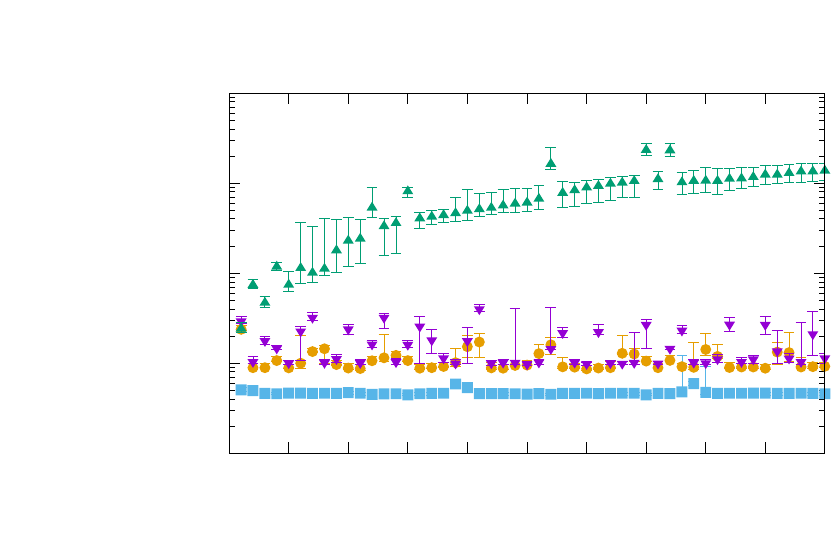}}\gplfronttext
  \end{picture}\endgroup
         \caption{Runtime performance at each step of a run.}
        \label{fig:perf-step1}
    \end{minipage}
    \hfill
    \begin{minipage}[t]{0.48\textwidth}
      \small\sf
        \hspace*{-1.8em}
        \begingroup
  \makeatletter
  \providecommand\color[2][]{\GenericError{(gnuplot) \space\space\space\@spaces}{Package color not loaded in conjunction with
      terminal option `colourtext'}{See the gnuplot documentation for explanation.}{Either use 'blacktext' in gnuplot or load the package
      color.sty in LaTeX.}\renewcommand\color[2][]{}}\providecommand\includegraphics[2][]{\GenericError{(gnuplot) \space\space\space\@spaces}{Package graphicx or graphics not loaded}{See the gnuplot documentation for explanation.}{The gnuplot epslatex terminal needs graphicx.sty or graphics.sty.}\renewcommand\includegraphics[2][]{}}\providecommand\rotatebox[2]{#2}\@ifundefined{ifGPcolor}{\newif\ifGPcolor
    \GPcolorfalse
  }{}\@ifundefined{ifGPblacktext}{\newif\ifGPblacktext
    \GPblacktexttrue
  }{}\let\gplgaddtomacro\g@addto@macro
\gdef\gplbacktext{}\gdef\gplfronttext{}\makeatother
  \ifGPblacktext
\def\colorrgb#1{}\def\colorgray#1{}\else
\ifGPcolor
      \def\colorrgb#1{\color[rgb]{#1}}\def\colorgray#1{\color[gray]{#1}}\expandafter\def\csname LTw\endcsname{\color{white}}\expandafter\def\csname LTb\endcsname{\color{black}}\expandafter\def\csname LTa\endcsname{\color{black}}\expandafter\def\csname LT0\endcsname{\color[rgb]{1,0,0}}\expandafter\def\csname LT1\endcsname{\color[rgb]{0,1,0}}\expandafter\def\csname LT2\endcsname{\color[rgb]{0,0,1}}\expandafter\def\csname LT3\endcsname{\color[rgb]{1,0,1}}\expandafter\def\csname LT4\endcsname{\color[rgb]{0,1,1}}\expandafter\def\csname LT5\endcsname{\color[rgb]{1,1,0}}\expandafter\def\csname LT6\endcsname{\color[rgb]{0,0,0}}\expandafter\def\csname LT7\endcsname{\color[rgb]{1,0.3,0}}\expandafter\def\csname LT8\endcsname{\color[rgb]{0.5,0.5,0.5}}\else
\def\colorrgb#1{\color{black}}\def\colorgray#1{\color[gray]{#1}}\expandafter\def\csname LTw\endcsname{\color{white}}\expandafter\def\csname LTb\endcsname{\color{black}}\expandafter\def\csname LTa\endcsname{\color{black}}\expandafter\def\csname LT0\endcsname{\color{black}}\expandafter\def\csname LT1\endcsname{\color{black}}\expandafter\def\csname LT2\endcsname{\color{black}}\expandafter\def\csname LT3\endcsname{\color{black}}\expandafter\def\csname LT4\endcsname{\color{black}}\expandafter\def\csname LT5\endcsname{\color{black}}\expandafter\def\csname LT6\endcsname{\color{black}}\expandafter\def\csname LT7\endcsname{\color{black}}\expandafter\def\csname LT8\endcsname{\color{black}}\fi
  \fi
    \setlength{\unitlength}{0.0500bp}\ifx\gptboxheight\undefined \newlength{\gptboxheight}\newlength{\gptboxwidth}\newsavebox{\gptboxtext}\fi \setlength{\fboxrule}{0.5pt}\setlength{\fboxsep}{1pt}\begin{picture}(4818.00,3174.00)\gplgaddtomacro\gplbacktext{\csname LTb\endcsname \put(1188,550){\makebox(0,0)[r]{\strut{}$10^{1}$}}\put(1188,2623){\makebox(0,0)[r]{\strut{}$10^{2}$}}\put(1320,330){\makebox(0,0){\strut{}$0$}}\put(1663,330){\makebox(0,0){\strut{}$50$}}\put(2006,330){\makebox(0,0){\strut{}$100$}}\put(2349,330){\makebox(0,0){\strut{}$150$}}\put(2692,330){\makebox(0,0){\strut{}$200$}}\put(3036,330){\makebox(0,0){\strut{}$250$}}\put(3379,330){\makebox(0,0){\strut{}$300$}}\put(3722,330){\makebox(0,0){\strut{}$350$}}\put(4065,330){\makebox(0,0){\strut{}$400$}}\put(4408,330){\makebox(0,0){\strut{}$450$}}\put(4751,330){\makebox(0,0){\strut{}$500$}}}\gplgaddtomacro\gplfronttext{\csname LTb\endcsname \put(3035,2733){\makebox(0,0){\strut{}Beta-Bernoulli Ideal Memory}}}\gplbacktext
    \put(0,0){\includegraphics{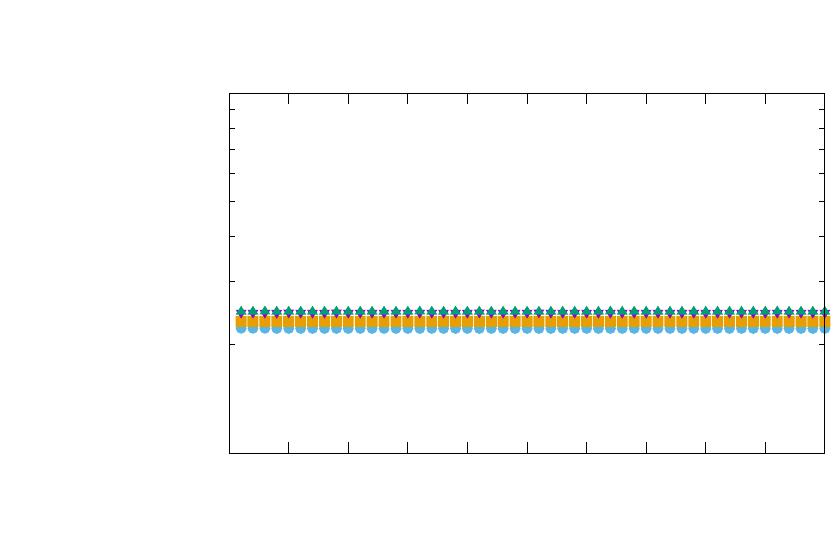}}\gplfronttext
  \end{picture}\endgroup
         \\[-.3cm]
        \hspace*{-1.8em}
        \begingroup
  \makeatletter
  \providecommand\color[2][]{\GenericError{(gnuplot) \space\space\space\@spaces}{Package color not loaded in conjunction with
      terminal option `colourtext'}{See the gnuplot documentation for explanation.}{Either use 'blacktext' in gnuplot or load the package
      color.sty in LaTeX.}\renewcommand\color[2][]{}}\providecommand\includegraphics[2][]{\GenericError{(gnuplot) \space\space\space\@spaces}{Package graphicx or graphics not loaded}{See the gnuplot documentation for explanation.}{The gnuplot epslatex terminal needs graphicx.sty or graphics.sty.}\renewcommand\includegraphics[2][]{}}\providecommand\rotatebox[2]{#2}\@ifundefined{ifGPcolor}{\newif\ifGPcolor
    \GPcolorfalse
  }{}\@ifundefined{ifGPblacktext}{\newif\ifGPblacktext
    \GPblacktexttrue
  }{}\let\gplgaddtomacro\g@addto@macro
\gdef\gplbacktext{}\gdef\gplfronttext{}\makeatother
  \ifGPblacktext
\def\colorrgb#1{}\def\colorgray#1{}\else
\ifGPcolor
      \def\colorrgb#1{\color[rgb]{#1}}\def\colorgray#1{\color[gray]{#1}}\expandafter\def\csname LTw\endcsname{\color{white}}\expandafter\def\csname LTb\endcsname{\color{black}}\expandafter\def\csname LTa\endcsname{\color{black}}\expandafter\def\csname LT0\endcsname{\color[rgb]{1,0,0}}\expandafter\def\csname LT1\endcsname{\color[rgb]{0,1,0}}\expandafter\def\csname LT2\endcsname{\color[rgb]{0,0,1}}\expandafter\def\csname LT3\endcsname{\color[rgb]{1,0,1}}\expandafter\def\csname LT4\endcsname{\color[rgb]{0,1,1}}\expandafter\def\csname LT5\endcsname{\color[rgb]{1,1,0}}\expandafter\def\csname LT6\endcsname{\color[rgb]{0,0,0}}\expandafter\def\csname LT7\endcsname{\color[rgb]{1,0.3,0}}\expandafter\def\csname LT8\endcsname{\color[rgb]{0.5,0.5,0.5}}\else
\def\colorrgb#1{\color{black}}\def\colorgray#1{\color[gray]{#1}}\expandafter\def\csname LTw\endcsname{\color{white}}\expandafter\def\csname LTb\endcsname{\color{black}}\expandafter\def\csname LTa\endcsname{\color{black}}\expandafter\def\csname LT0\endcsname{\color{black}}\expandafter\def\csname LT1\endcsname{\color{black}}\expandafter\def\csname LT2\endcsname{\color{black}}\expandafter\def\csname LT3\endcsname{\color{black}}\expandafter\def\csname LT4\endcsname{\color{black}}\expandafter\def\csname LT5\endcsname{\color{black}}\expandafter\def\csname LT6\endcsname{\color{black}}\expandafter\def\csname LT7\endcsname{\color{black}}\expandafter\def\csname LT8\endcsname{\color{black}}\fi
  \fi
    \setlength{\unitlength}{0.0500bp}\ifx\gptboxheight\undefined \newlength{\gptboxheight}\newlength{\gptboxwidth}\newsavebox{\gptboxtext}\fi \setlength{\fboxrule}{0.5pt}\setlength{\fboxsep}{1pt}\begin{picture}(4818.00,3174.00)\gplgaddtomacro\gplbacktext{\csname LTb\endcsname \put(1188,550){\makebox(0,0)[r]{\strut{}$10^{1}$}}\put(1188,2623){\makebox(0,0)[r]{\strut{}$10^{2}$}}\put(1320,330){\makebox(0,0){\strut{}$0$}}\put(1663,330){\makebox(0,0){\strut{}$50$}}\put(2006,330){\makebox(0,0){\strut{}$100$}}\put(2349,330){\makebox(0,0){\strut{}$150$}}\put(2692,330){\makebox(0,0){\strut{}$200$}}\put(3036,330){\makebox(0,0){\strut{}$250$}}\put(3379,330){\makebox(0,0){\strut{}$300$}}\put(3722,330){\makebox(0,0){\strut{}$350$}}\put(4065,330){\makebox(0,0){\strut{}$400$}}\put(4408,330){\makebox(0,0){\strut{}$450$}}\put(4751,330){\makebox(0,0){\strut{}$500$}}}\gplgaddtomacro\gplfronttext{\csname LTb\endcsname \put(715,1586){\rotatebox{-270}{\makebox(0,0){\strut{}Thousands of words in heap (log scale)}}}\put(3035,2733){\makebox(0,0){\strut{}Gaussian-Gaussian Ideal Memory}}}\gplbacktext
    \put(0,0){\includegraphics{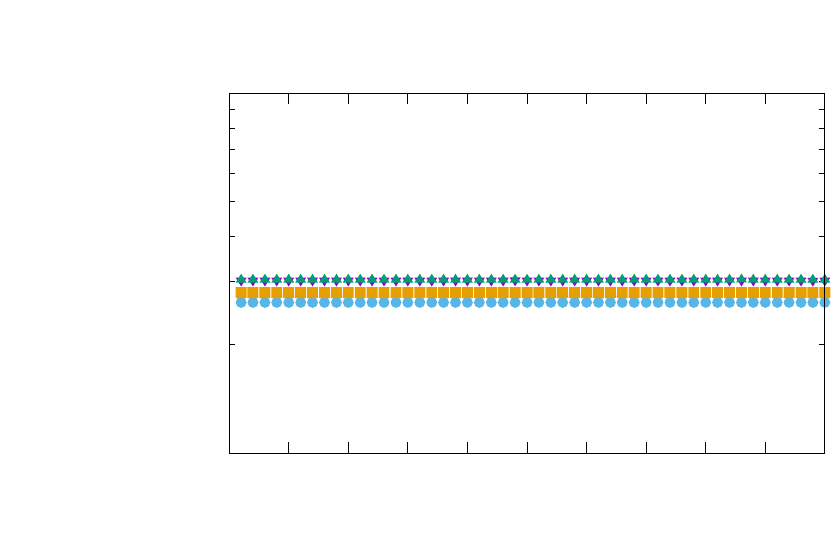}}\gplfronttext
  \end{picture}\endgroup
         \\[-.3cm]
        \hspace*{-1.8em}
        \begingroup
  \makeatletter
  \providecommand\color[2][]{\GenericError{(gnuplot) \space\space\space\@spaces}{Package color not loaded in conjunction with
      terminal option `colourtext'}{See the gnuplot documentation for explanation.}{Either use 'blacktext' in gnuplot or load the package
      color.sty in LaTeX.}\renewcommand\color[2][]{}}\providecommand\includegraphics[2][]{\GenericError{(gnuplot) \space\space\space\@spaces}{Package graphicx or graphics not loaded}{See the gnuplot documentation for explanation.}{The gnuplot epslatex terminal needs graphicx.sty or graphics.sty.}\renewcommand\includegraphics[2][]{}}\providecommand\rotatebox[2]{#2}\@ifundefined{ifGPcolor}{\newif\ifGPcolor
    \GPcolorfalse
  }{}\@ifundefined{ifGPblacktext}{\newif\ifGPblacktext
    \GPblacktexttrue
  }{}\let\gplgaddtomacro\g@addto@macro
\gdef\gplbacktext{}\gdef\gplfronttext{}\makeatother
  \ifGPblacktext
\def\colorrgb#1{}\def\colorgray#1{}\else
\ifGPcolor
      \def\colorrgb#1{\color[rgb]{#1}}\def\colorgray#1{\color[gray]{#1}}\expandafter\def\csname LTw\endcsname{\color{white}}\expandafter\def\csname LTb\endcsname{\color{black}}\expandafter\def\csname LTa\endcsname{\color{black}}\expandafter\def\csname LT0\endcsname{\color[rgb]{1,0,0}}\expandafter\def\csname LT1\endcsname{\color[rgb]{0,1,0}}\expandafter\def\csname LT2\endcsname{\color[rgb]{0,0,1}}\expandafter\def\csname LT3\endcsname{\color[rgb]{1,0,1}}\expandafter\def\csname LT4\endcsname{\color[rgb]{0,1,1}}\expandafter\def\csname LT5\endcsname{\color[rgb]{1,1,0}}\expandafter\def\csname LT6\endcsname{\color[rgb]{0,0,0}}\expandafter\def\csname LT7\endcsname{\color[rgb]{1,0.3,0}}\expandafter\def\csname LT8\endcsname{\color[rgb]{0.5,0.5,0.5}}\else
\def\colorrgb#1{\color{black}}\def\colorgray#1{\color[gray]{#1}}\expandafter\def\csname LTw\endcsname{\color{white}}\expandafter\def\csname LTb\endcsname{\color{black}}\expandafter\def\csname LTa\endcsname{\color{black}}\expandafter\def\csname LT0\endcsname{\color{black}}\expandafter\def\csname LT1\endcsname{\color{black}}\expandafter\def\csname LT2\endcsname{\color{black}}\expandafter\def\csname LT3\endcsname{\color{black}}\expandafter\def\csname LT4\endcsname{\color{black}}\expandafter\def\csname LT5\endcsname{\color{black}}\expandafter\def\csname LT6\endcsname{\color{black}}\expandafter\def\csname LT7\endcsname{\color{black}}\expandafter\def\csname LT8\endcsname{\color{black}}\fi
  \fi
    \setlength{\unitlength}{0.0500bp}\ifx\gptboxheight\undefined \newlength{\gptboxheight}\newlength{\gptboxwidth}\newsavebox{\gptboxtext}\fi \setlength{\fboxrule}{0.5pt}\setlength{\fboxsep}{1pt}\begin{picture}(4818.00,3174.00)\gplgaddtomacro\gplbacktext{\csname LTb\endcsname \put(1188,550){\makebox(0,0)[r]{\strut{}$10^{1}$}}\put(1188,1241){\makebox(0,0)[r]{\strut{}$10^{2}$}}\put(1188,1932){\makebox(0,0)[r]{\strut{}$10^{3}$}}\put(1188,2623){\makebox(0,0)[r]{\strut{}$10^{4}$}}\put(1320,330){\makebox(0,0){\strut{}$0$}}\put(1663,330){\makebox(0,0){\strut{}$50$}}\put(2006,330){\makebox(0,0){\strut{}$100$}}\put(2349,330){\makebox(0,0){\strut{}$150$}}\put(2692,330){\makebox(0,0){\strut{}$200$}}\put(3036,330){\makebox(0,0){\strut{}$250$}}\put(3379,330){\makebox(0,0){\strut{}$300$}}\put(3722,330){\makebox(0,0){\strut{}$350$}}\put(4065,330){\makebox(0,0){\strut{}$400$}}\put(4408,330){\makebox(0,0){\strut{}$450$}}\put(4751,330){\makebox(0,0){\strut{}$500$}}}\gplgaddtomacro\gplfronttext{\csname LTb\endcsname \put(3035,2733){\makebox(0,0){\strut{}Kalman-1D Ideal Memory}}}\gplbacktext
    \put(0,0){\includegraphics{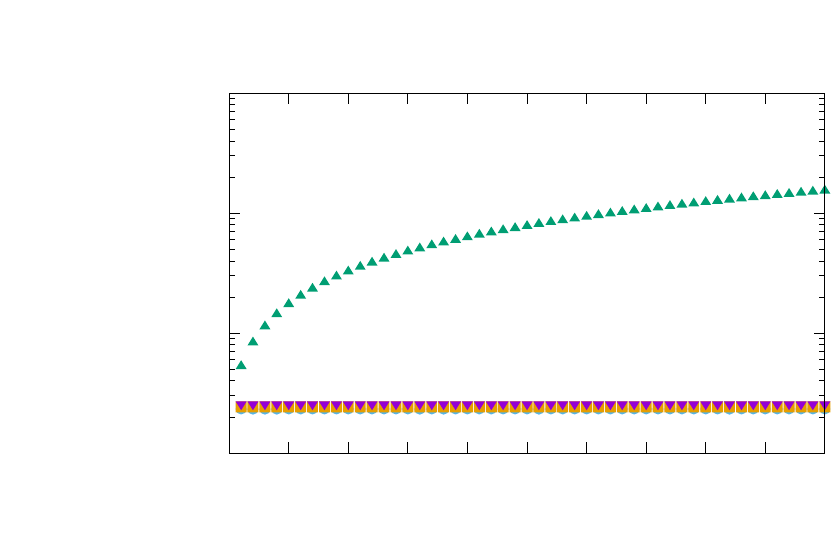}}\gplfronttext
  \end{picture}\endgroup
         \\[-.3cm]
        \hspace*{-1.8em}
        \begingroup
  \makeatletter
  \providecommand\color[2][]{\GenericError{(gnuplot) \space\space\space\@spaces}{Package color not loaded in conjunction with
      terminal option `colourtext'}{See the gnuplot documentation for explanation.}{Either use 'blacktext' in gnuplot or load the package
      color.sty in LaTeX.}\renewcommand\color[2][]{}}\providecommand\includegraphics[2][]{\GenericError{(gnuplot) \space\space\space\@spaces}{Package graphicx or graphics not loaded}{See the gnuplot documentation for explanation.}{The gnuplot epslatex terminal needs graphicx.sty or graphics.sty.}\renewcommand\includegraphics[2][]{}}\providecommand\rotatebox[2]{#2}\@ifundefined{ifGPcolor}{\newif\ifGPcolor
    \GPcolorfalse
  }{}\@ifundefined{ifGPblacktext}{\newif\ifGPblacktext
    \GPblacktexttrue
  }{}\let\gplgaddtomacro\g@addto@macro
\gdef\gplbacktext{}\gdef\gplfronttext{}\makeatother
  \ifGPblacktext
\def\colorrgb#1{}\def\colorgray#1{}\else
\ifGPcolor
      \def\colorrgb#1{\color[rgb]{#1}}\def\colorgray#1{\color[gray]{#1}}\expandafter\def\csname LTw\endcsname{\color{white}}\expandafter\def\csname LTb\endcsname{\color{black}}\expandafter\def\csname LTa\endcsname{\color{black}}\expandafter\def\csname LT0\endcsname{\color[rgb]{1,0,0}}\expandafter\def\csname LT1\endcsname{\color[rgb]{0,1,0}}\expandafter\def\csname LT2\endcsname{\color[rgb]{0,0,1}}\expandafter\def\csname LT3\endcsname{\color[rgb]{1,0,1}}\expandafter\def\csname LT4\endcsname{\color[rgb]{0,1,1}}\expandafter\def\csname LT5\endcsname{\color[rgb]{1,1,0}}\expandafter\def\csname LT6\endcsname{\color[rgb]{0,0,0}}\expandafter\def\csname LT7\endcsname{\color[rgb]{1,0.3,0}}\expandafter\def\csname LT8\endcsname{\color[rgb]{0.5,0.5,0.5}}\else
\def\colorrgb#1{\color{black}}\def\colorgray#1{\color[gray]{#1}}\expandafter\def\csname LTw\endcsname{\color{white}}\expandafter\def\csname LTb\endcsname{\color{black}}\expandafter\def\csname LTa\endcsname{\color{black}}\expandafter\def\csname LT0\endcsname{\color{black}}\expandafter\def\csname LT1\endcsname{\color{black}}\expandafter\def\csname LT2\endcsname{\color{black}}\expandafter\def\csname LT3\endcsname{\color{black}}\expandafter\def\csname LT4\endcsname{\color{black}}\expandafter\def\csname LT5\endcsname{\color{black}}\expandafter\def\csname LT6\endcsname{\color{black}}\expandafter\def\csname LT7\endcsname{\color{black}}\expandafter\def\csname LT8\endcsname{\color{black}}\fi
  \fi
    \setlength{\unitlength}{0.0500bp}\ifx\gptboxheight\undefined \newlength{\gptboxheight}\newlength{\gptboxwidth}\newsavebox{\gptboxtext}\fi \setlength{\fboxrule}{0.5pt}\setlength{\fboxsep}{1pt}\begin{picture}(4818.00,3174.00)\gplgaddtomacro\gplbacktext{\csname LTb\endcsname \put(1188,550){\makebox(0,0)[r]{\strut{}$10^{1}$}}\put(1188,1241){\makebox(0,0)[r]{\strut{}$10^{2}$}}\put(1188,1932){\makebox(0,0)[r]{\strut{}$10^{3}$}}\put(1188,2623){\makebox(0,0)[r]{\strut{}$10^{4}$}}\put(1320,330){\makebox(0,0){\strut{}$0$}}\put(1663,330){\makebox(0,0){\strut{}$50$}}\put(2006,330){\makebox(0,0){\strut{}$100$}}\put(2349,330){\makebox(0,0){\strut{}$150$}}\put(2692,330){\makebox(0,0){\strut{}$200$}}\put(3036,330){\makebox(0,0){\strut{}$250$}}\put(3379,330){\makebox(0,0){\strut{}$300$}}\put(3722,330){\makebox(0,0){\strut{}$350$}}\put(4065,330){\makebox(0,0){\strut{}$400$}}\put(4408,330){\makebox(0,0){\strut{}$450$}}\put(4751,330){\makebox(0,0){\strut{}$500$}}}\gplgaddtomacro\gplfronttext{\csname LTb\endcsname \put(3035,0){\makebox(0,0){\strut{}Step}}\put(3035,2733){\makebox(0,0){\strut{}Outlier Ideal Memory}}}\gplbacktext
    \put(0,0){\includegraphics{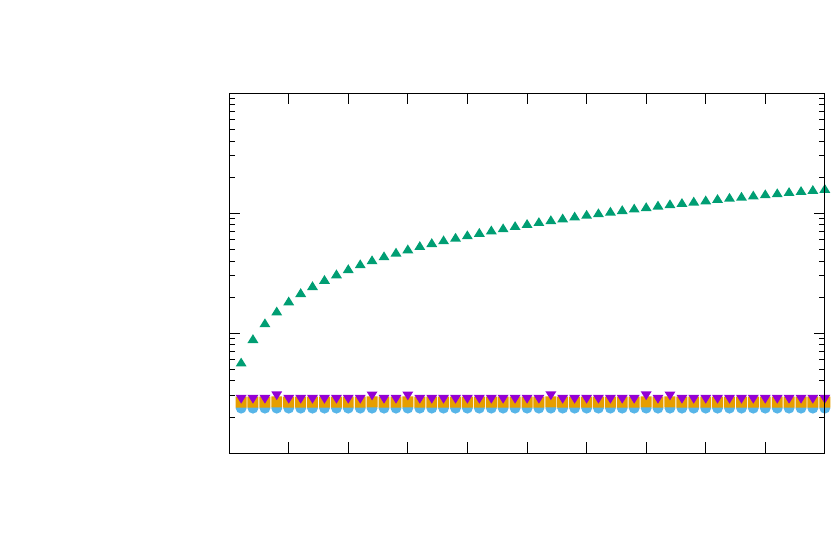}}\gplfronttext
  \end{picture}\endgroup
         \caption{Memory consumption at each step of a run.}
        \label{fig:mem1}
    \end{minipage}
\end{figure*}

\begin{figure*}
  \begin{center}
    \small\sf
    \pfbullet~PF $\quad$ \bdsbullet~BDS $\quad$ \pmdsbullet~SDS $\quad$ \ndsbullet~DS
  \end{center}
    \begin{minipage}[t]{0.48\textwidth}
      \small\sf
        \hspace*{-1.8em}
        \begingroup
  \makeatletter
  \providecommand\color[2][]{\GenericError{(gnuplot) \space\space\space\@spaces}{Package color not loaded in conjunction with
      terminal option `colourtext'}{See the gnuplot documentation for explanation.}{Either use 'blacktext' in gnuplot or load the package
      color.sty in LaTeX.}\renewcommand\color[2][]{}}\providecommand\includegraphics[2][]{\GenericError{(gnuplot) \space\space\space\@spaces}{Package graphicx or graphics not loaded}{See the gnuplot documentation for explanation.}{The gnuplot epslatex terminal needs graphicx.sty or graphics.sty.}\renewcommand\includegraphics[2][]{}}\providecommand\rotatebox[2]{#2}\@ifundefined{ifGPcolor}{\newif\ifGPcolor
    \GPcolorfalse
  }{}\@ifundefined{ifGPblacktext}{\newif\ifGPblacktext
    \GPblacktexttrue
  }{}\let\gplgaddtomacro\g@addto@macro
\gdef\gplbacktext{}\gdef\gplfronttext{}\makeatother
  \ifGPblacktext
\def\colorrgb#1{}\def\colorgray#1{}\else
\ifGPcolor
      \def\colorrgb#1{\color[rgb]{#1}}\def\colorgray#1{\color[gray]{#1}}\expandafter\def\csname LTw\endcsname{\color{white}}\expandafter\def\csname LTb\endcsname{\color{black}}\expandafter\def\csname LTa\endcsname{\color{black}}\expandafter\def\csname LT0\endcsname{\color[rgb]{1,0,0}}\expandafter\def\csname LT1\endcsname{\color[rgb]{0,1,0}}\expandafter\def\csname LT2\endcsname{\color[rgb]{0,0,1}}\expandafter\def\csname LT3\endcsname{\color[rgb]{1,0,1}}\expandafter\def\csname LT4\endcsname{\color[rgb]{0,1,1}}\expandafter\def\csname LT5\endcsname{\color[rgb]{1,1,0}}\expandafter\def\csname LT6\endcsname{\color[rgb]{0,0,0}}\expandafter\def\csname LT7\endcsname{\color[rgb]{1,0.3,0}}\expandafter\def\csname LT8\endcsname{\color[rgb]{0.5,0.5,0.5}}\else
\def\colorrgb#1{\color{black}}\def\colorgray#1{\color[gray]{#1}}\expandafter\def\csname LTw\endcsname{\color{white}}\expandafter\def\csname LTb\endcsname{\color{black}}\expandafter\def\csname LTa\endcsname{\color{black}}\expandafter\def\csname LT0\endcsname{\color{black}}\expandafter\def\csname LT1\endcsname{\color{black}}\expandafter\def\csname LT2\endcsname{\color{black}}\expandafter\def\csname LT3\endcsname{\color{black}}\expandafter\def\csname LT4\endcsname{\color{black}}\expandafter\def\csname LT5\endcsname{\color{black}}\expandafter\def\csname LT6\endcsname{\color{black}}\expandafter\def\csname LT7\endcsname{\color{black}}\expandafter\def\csname LT8\endcsname{\color{black}}\fi
  \fi
    \setlength{\unitlength}{0.0500bp}\ifx\gptboxheight\undefined \newlength{\gptboxheight}\newlength{\gptboxwidth}\newsavebox{\gptboxtext}\fi \setlength{\fboxrule}{0.5pt}\setlength{\fboxsep}{1pt}\begin{picture}(4818.00,3174.00)\gplgaddtomacro\gplbacktext{\csname LTb\endcsname \put(1188,550){\makebox(0,0)[r]{\strut{}$10^{0}$}}\put(1188,1587){\makebox(0,0)[r]{\strut{}$10^{1}$}}\put(1188,2623){\makebox(0,0)[r]{\strut{}$10^{2}$}}\put(1320,330){\makebox(0,0){\strut{}$0$}}\put(1663,330){\makebox(0,0){\strut{}$50$}}\put(2006,330){\makebox(0,0){\strut{}$100$}}\put(2349,330){\makebox(0,0){\strut{}$150$}}\put(2692,330){\makebox(0,0){\strut{}$200$}}\put(3036,330){\makebox(0,0){\strut{}$250$}}\put(3379,330){\makebox(0,0){\strut{}$300$}}\put(3722,330){\makebox(0,0){\strut{}$350$}}\put(4065,330){\makebox(0,0){\strut{}$400$}}\put(4408,330){\makebox(0,0){\strut{}$450$}}\put(4751,330){\makebox(0,0){\strut{}$500$}}}\gplgaddtomacro\gplfronttext{\csname LTb\endcsname \put(3035,2733){\makebox(0,0){\strut{}Robot Latency}}}\gplbacktext
    \put(0,0){\includegraphics{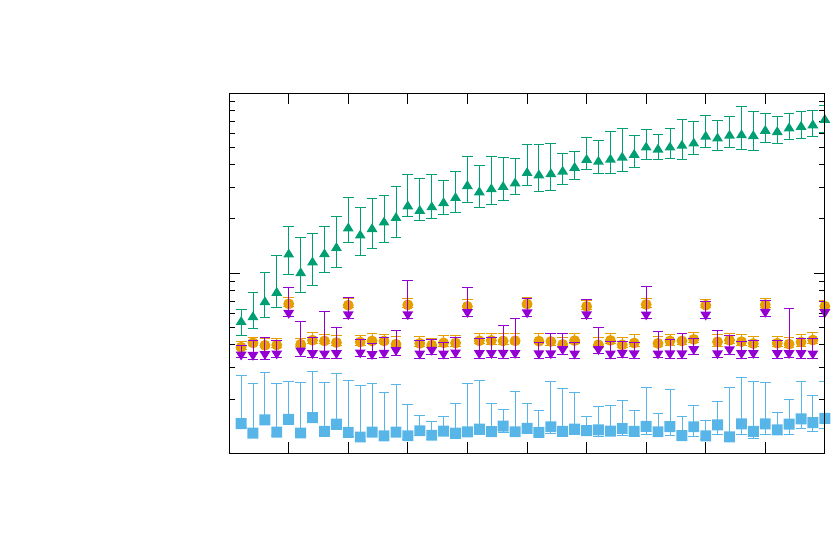}}\gplfronttext
  \end{picture}\endgroup
         \\[-0.3cm]
        \hspace*{-1.8em}
        \begingroup
  \makeatletter
  \providecommand\color[2][]{\GenericError{(gnuplot) \space\space\space\@spaces}{Package color not loaded in conjunction with
      terminal option `colourtext'}{See the gnuplot documentation for explanation.}{Either use 'blacktext' in gnuplot or load the package
      color.sty in LaTeX.}\renewcommand\color[2][]{}}\providecommand\includegraphics[2][]{\GenericError{(gnuplot) \space\space\space\@spaces}{Package graphicx or graphics not loaded}{See the gnuplot documentation for explanation.}{The gnuplot epslatex terminal needs graphicx.sty or graphics.sty.}\renewcommand\includegraphics[2][]{}}\providecommand\rotatebox[2]{#2}\@ifundefined{ifGPcolor}{\newif\ifGPcolor
    \GPcolorfalse
  }{}\@ifundefined{ifGPblacktext}{\newif\ifGPblacktext
    \GPblacktexttrue
  }{}\let\gplgaddtomacro\g@addto@macro
\gdef\gplbacktext{}\gdef\gplfronttext{}\makeatother
  \ifGPblacktext
\def\colorrgb#1{}\def\colorgray#1{}\else
\ifGPcolor
      \def\colorrgb#1{\color[rgb]{#1}}\def\colorgray#1{\color[gray]{#1}}\expandafter\def\csname LTw\endcsname{\color{white}}\expandafter\def\csname LTb\endcsname{\color{black}}\expandafter\def\csname LTa\endcsname{\color{black}}\expandafter\def\csname LT0\endcsname{\color[rgb]{1,0,0}}\expandafter\def\csname LT1\endcsname{\color[rgb]{0,1,0}}\expandafter\def\csname LT2\endcsname{\color[rgb]{0,0,1}}\expandafter\def\csname LT3\endcsname{\color[rgb]{1,0,1}}\expandafter\def\csname LT4\endcsname{\color[rgb]{0,1,1}}\expandafter\def\csname LT5\endcsname{\color[rgb]{1,1,0}}\expandafter\def\csname LT6\endcsname{\color[rgb]{0,0,0}}\expandafter\def\csname LT7\endcsname{\color[rgb]{1,0.3,0}}\expandafter\def\csname LT8\endcsname{\color[rgb]{0.5,0.5,0.5}}\else
\def\colorrgb#1{\color{black}}\def\colorgray#1{\color[gray]{#1}}\expandafter\def\csname LTw\endcsname{\color{white}}\expandafter\def\csname LTb\endcsname{\color{black}}\expandafter\def\csname LTa\endcsname{\color{black}}\expandafter\def\csname LT0\endcsname{\color{black}}\expandafter\def\csname LT1\endcsname{\color{black}}\expandafter\def\csname LT2\endcsname{\color{black}}\expandafter\def\csname LT3\endcsname{\color{black}}\expandafter\def\csname LT4\endcsname{\color{black}}\expandafter\def\csname LT5\endcsname{\color{black}}\expandafter\def\csname LT6\endcsname{\color{black}}\expandafter\def\csname LT7\endcsname{\color{black}}\expandafter\def\csname LT8\endcsname{\color{black}}\fi
  \fi
    \setlength{\unitlength}{0.0500bp}\ifx\gptboxheight\undefined \newlength{\gptboxheight}\newlength{\gptboxwidth}\newsavebox{\gptboxtext}\fi \setlength{\fboxrule}{0.5pt}\setlength{\fboxsep}{1pt}\begin{picture}(4818.00,3174.00)\gplgaddtomacro\gplbacktext{\csname LTb\endcsname \put(1188,550){\makebox(0,0)[r]{\strut{}$10^{-2}$}}\put(1188,1241){\makebox(0,0)[r]{\strut{}$10^{-1}$}}\put(1188,1932){\makebox(0,0)[r]{\strut{}$10^{0}$}}\put(1188,2623){\makebox(0,0)[r]{\strut{}$10^{1}$}}\put(1320,330){\makebox(0,0){\strut{}$0$}}\put(1663,330){\makebox(0,0){\strut{}$50$}}\put(2006,330){\makebox(0,0){\strut{}$100$}}\put(2349,330){\makebox(0,0){\strut{}$150$}}\put(2692,330){\makebox(0,0){\strut{}$200$}}\put(3036,330){\makebox(0,0){\strut{}$250$}}\put(3379,330){\makebox(0,0){\strut{}$300$}}\put(3722,330){\makebox(0,0){\strut{}$350$}}\put(4065,330){\makebox(0,0){\strut{}$400$}}\put(4408,330){\makebox(0,0){\strut{}$450$}}\put(4751,330){\makebox(0,0){\strut{}$500$}}}\gplgaddtomacro\gplfronttext{\csname LTb\endcsname \put(583,1586){\rotatebox{-270}{\makebox(0,0){\strut{}Step latency in ms (log scale)}}}\put(3035,2733){\makebox(0,0){\strut{}SLAM Latency}}}\gplbacktext
    \put(0,0){\includegraphics{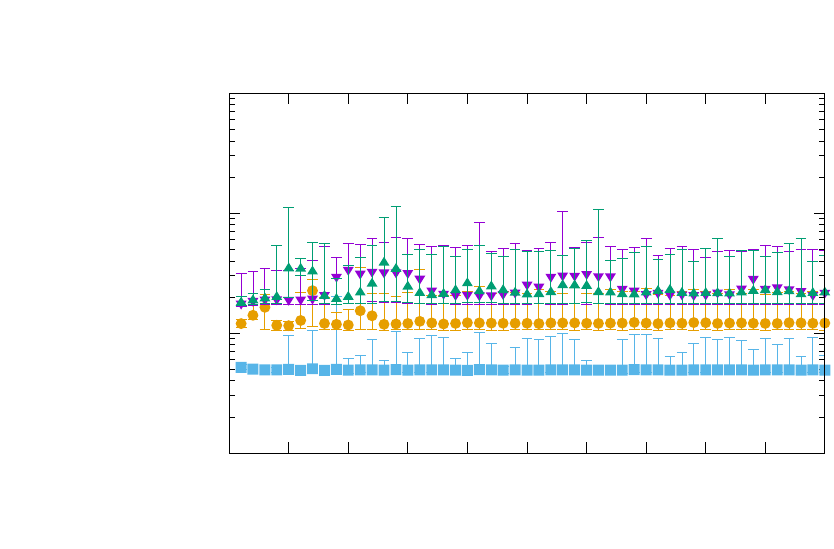}}\gplfronttext
  \end{picture}\endgroup
         \\[-0.3cm]
        \hspace*{-1.8em}
        \begingroup
  \makeatletter
  \providecommand\color[2][]{\GenericError{(gnuplot) \space\space\space\@spaces}{Package color not loaded in conjunction with
      terminal option `colourtext'}{See the gnuplot documentation for explanation.}{Either use 'blacktext' in gnuplot or load the package
      color.sty in LaTeX.}\renewcommand\color[2][]{}}\providecommand\includegraphics[2][]{\GenericError{(gnuplot) \space\space\space\@spaces}{Package graphicx or graphics not loaded}{See the gnuplot documentation for explanation.}{The gnuplot epslatex terminal needs graphicx.sty or graphics.sty.}\renewcommand\includegraphics[2][]{}}\providecommand\rotatebox[2]{#2}\@ifundefined{ifGPcolor}{\newif\ifGPcolor
    \GPcolorfalse
  }{}\@ifundefined{ifGPblacktext}{\newif\ifGPblacktext
    \GPblacktexttrue
  }{}\let\gplgaddtomacro\g@addto@macro
\gdef\gplbacktext{}\gdef\gplfronttext{}\makeatother
  \ifGPblacktext
\def\colorrgb#1{}\def\colorgray#1{}\else
\ifGPcolor
      \def\colorrgb#1{\color[rgb]{#1}}\def\colorgray#1{\color[gray]{#1}}\expandafter\def\csname LTw\endcsname{\color{white}}\expandafter\def\csname LTb\endcsname{\color{black}}\expandafter\def\csname LTa\endcsname{\color{black}}\expandafter\def\csname LT0\endcsname{\color[rgb]{1,0,0}}\expandafter\def\csname LT1\endcsname{\color[rgb]{0,1,0}}\expandafter\def\csname LT2\endcsname{\color[rgb]{0,0,1}}\expandafter\def\csname LT3\endcsname{\color[rgb]{1,0,1}}\expandafter\def\csname LT4\endcsname{\color[rgb]{0,1,1}}\expandafter\def\csname LT5\endcsname{\color[rgb]{1,1,0}}\expandafter\def\csname LT6\endcsname{\color[rgb]{0,0,0}}\expandafter\def\csname LT7\endcsname{\color[rgb]{1,0.3,0}}\expandafter\def\csname LT8\endcsname{\color[rgb]{0.5,0.5,0.5}}\else
\def\colorrgb#1{\color{black}}\def\colorgray#1{\color[gray]{#1}}\expandafter\def\csname LTw\endcsname{\color{white}}\expandafter\def\csname LTb\endcsname{\color{black}}\expandafter\def\csname LTa\endcsname{\color{black}}\expandafter\def\csname LT0\endcsname{\color{black}}\expandafter\def\csname LT1\endcsname{\color{black}}\expandafter\def\csname LT2\endcsname{\color{black}}\expandafter\def\csname LT3\endcsname{\color{black}}\expandafter\def\csname LT4\endcsname{\color{black}}\expandafter\def\csname LT5\endcsname{\color{black}}\expandafter\def\csname LT6\endcsname{\color{black}}\expandafter\def\csname LT7\endcsname{\color{black}}\expandafter\def\csname LT8\endcsname{\color{black}}\fi
  \fi
    \setlength{\unitlength}{0.0500bp}\ifx\gptboxheight\undefined \newlength{\gptboxheight}\newlength{\gptboxwidth}\newsavebox{\gptboxtext}\fi \setlength{\fboxrule}{0.5pt}\setlength{\fboxsep}{1pt}\begin{picture}(4818.00,3174.00)\gplgaddtomacro\gplbacktext{\csname LTb\endcsname \put(1188,550){\makebox(0,0)[r]{\strut{}$10^{-1}$}}\put(1188,1068){\makebox(0,0)[r]{\strut{}$10^{0}$}}\put(1188,1587){\makebox(0,0)[r]{\strut{}$10^{1}$}}\put(1188,2105){\makebox(0,0)[r]{\strut{}$10^{2}$}}\put(1188,2623){\makebox(0,0)[r]{\strut{}$10^{3}$}}\put(1320,330){\makebox(0,0){\strut{}$0$}}\put(1663,330){\makebox(0,0){\strut{}$50$}}\put(2006,330){\makebox(0,0){\strut{}$100$}}\put(2349,330){\makebox(0,0){\strut{}$150$}}\put(2692,330){\makebox(0,0){\strut{}$200$}}\put(3036,330){\makebox(0,0){\strut{}$250$}}\put(3379,330){\makebox(0,0){\strut{}$300$}}\put(3722,330){\makebox(0,0){\strut{}$350$}}\put(4065,330){\makebox(0,0){\strut{}$400$}}\put(4408,330){\makebox(0,0){\strut{}$450$}}\put(4751,330){\makebox(0,0){\strut{}$500$}}}\gplgaddtomacro\gplfronttext{\csname LTb\endcsname \put(3035,0){\makebox(0,0){\strut{}Step}}\put(3035,2733){\makebox(0,0){\strut{}MTT Latency}}}\gplbacktext
    \put(0,0){\includegraphics{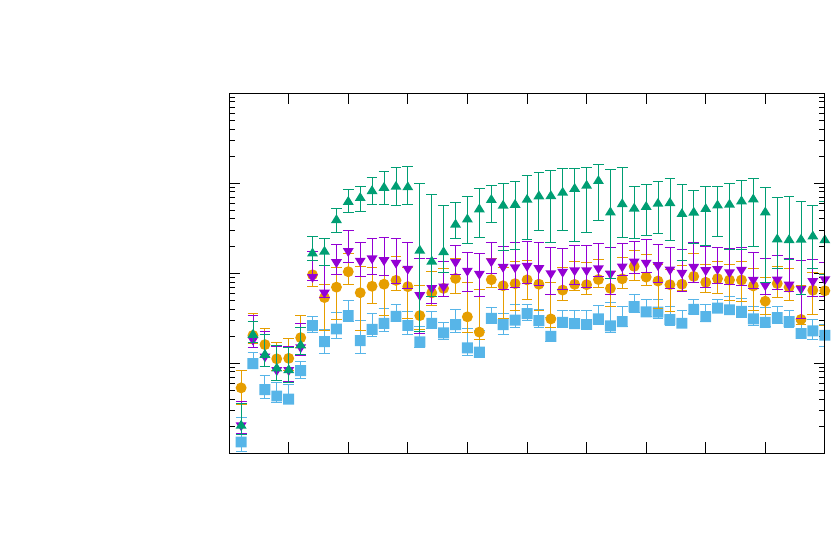}}\gplfronttext
  \end{picture}\endgroup
         \caption{Runtime performance at each step of a run.}
        \label{fig:perf-step2}
    \end{minipage}
    \hfill
    \begin{minipage}[t]{0.48\textwidth}
      \small\sf
        \hspace*{-1.8em}
        \begingroup
  \makeatletter
  \providecommand\color[2][]{\GenericError{(gnuplot) \space\space\space\@spaces}{Package color not loaded in conjunction with
      terminal option `colourtext'}{See the gnuplot documentation for explanation.}{Either use 'blacktext' in gnuplot or load the package
      color.sty in LaTeX.}\renewcommand\color[2][]{}}\providecommand\includegraphics[2][]{\GenericError{(gnuplot) \space\space\space\@spaces}{Package graphicx or graphics not loaded}{See the gnuplot documentation for explanation.}{The gnuplot epslatex terminal needs graphicx.sty or graphics.sty.}\renewcommand\includegraphics[2][]{}}\providecommand\rotatebox[2]{#2}\@ifundefined{ifGPcolor}{\newif\ifGPcolor
    \GPcolorfalse
  }{}\@ifundefined{ifGPblacktext}{\newif\ifGPblacktext
    \GPblacktexttrue
  }{}\let\gplgaddtomacro\g@addto@macro
\gdef\gplbacktext{}\gdef\gplfronttext{}\makeatother
  \ifGPblacktext
\def\colorrgb#1{}\def\colorgray#1{}\else
\ifGPcolor
      \def\colorrgb#1{\color[rgb]{#1}}\def\colorgray#1{\color[gray]{#1}}\expandafter\def\csname LTw\endcsname{\color{white}}\expandafter\def\csname LTb\endcsname{\color{black}}\expandafter\def\csname LTa\endcsname{\color{black}}\expandafter\def\csname LT0\endcsname{\color[rgb]{1,0,0}}\expandafter\def\csname LT1\endcsname{\color[rgb]{0,1,0}}\expandafter\def\csname LT2\endcsname{\color[rgb]{0,0,1}}\expandafter\def\csname LT3\endcsname{\color[rgb]{1,0,1}}\expandafter\def\csname LT4\endcsname{\color[rgb]{0,1,1}}\expandafter\def\csname LT5\endcsname{\color[rgb]{1,1,0}}\expandafter\def\csname LT6\endcsname{\color[rgb]{0,0,0}}\expandafter\def\csname LT7\endcsname{\color[rgb]{1,0.3,0}}\expandafter\def\csname LT8\endcsname{\color[rgb]{0.5,0.5,0.5}}\else
\def\colorrgb#1{\color{black}}\def\colorgray#1{\color[gray]{#1}}\expandafter\def\csname LTw\endcsname{\color{white}}\expandafter\def\csname LTb\endcsname{\color{black}}\expandafter\def\csname LTa\endcsname{\color{black}}\expandafter\def\csname LT0\endcsname{\color{black}}\expandafter\def\csname LT1\endcsname{\color{black}}\expandafter\def\csname LT2\endcsname{\color{black}}\expandafter\def\csname LT3\endcsname{\color{black}}\expandafter\def\csname LT4\endcsname{\color{black}}\expandafter\def\csname LT5\endcsname{\color{black}}\expandafter\def\csname LT6\endcsname{\color{black}}\expandafter\def\csname LT7\endcsname{\color{black}}\expandafter\def\csname LT8\endcsname{\color{black}}\fi
  \fi
    \setlength{\unitlength}{0.0500bp}\ifx\gptboxheight\undefined \newlength{\gptboxheight}\newlength{\gptboxwidth}\newsavebox{\gptboxtext}\fi \setlength{\fboxrule}{0.5pt}\setlength{\fboxsep}{1pt}\begin{picture}(4818.00,3174.00)\gplgaddtomacro\gplbacktext{\csname LTb\endcsname \put(1188,550){\makebox(0,0)[r]{\strut{}$10^{1}$}}\put(1188,1241){\makebox(0,0)[r]{\strut{}$10^{2}$}}\put(1188,1932){\makebox(0,0)[r]{\strut{}$10^{3}$}}\put(1188,2623){\makebox(0,0)[r]{\strut{}$10^{4}$}}\put(1320,330){\makebox(0,0){\strut{}$0$}}\put(1663,330){\makebox(0,0){\strut{}$50$}}\put(2006,330){\makebox(0,0){\strut{}$100$}}\put(2349,330){\makebox(0,0){\strut{}$150$}}\put(2692,330){\makebox(0,0){\strut{}$200$}}\put(3036,330){\makebox(0,0){\strut{}$250$}}\put(3379,330){\makebox(0,0){\strut{}$300$}}\put(3722,330){\makebox(0,0){\strut{}$350$}}\put(4065,330){\makebox(0,0){\strut{}$400$}}\put(4408,330){\makebox(0,0){\strut{}$450$}}\put(4751,330){\makebox(0,0){\strut{}$500$}}}\gplgaddtomacro\gplfronttext{\csname LTb\endcsname \put(3035,2733){\makebox(0,0){\strut{}Robot Ideal Memory}}}\gplbacktext
    \put(0,0){\includegraphics{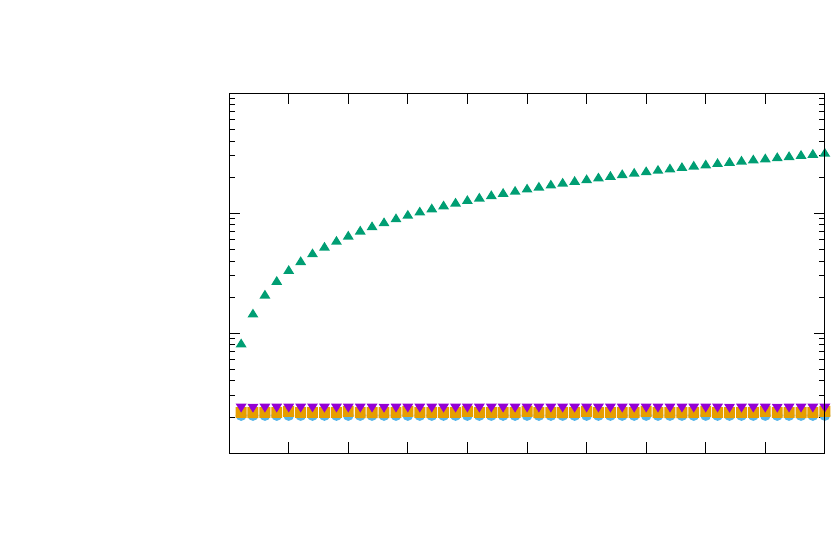}}\gplfronttext
  \end{picture}\endgroup
         \\[-.3cm]
        \hspace*{-1.8em}
        \begingroup
  \makeatletter
  \providecommand\color[2][]{\GenericError{(gnuplot) \space\space\space\@spaces}{Package color not loaded in conjunction with
      terminal option `colourtext'}{See the gnuplot documentation for explanation.}{Either use 'blacktext' in gnuplot or load the package
      color.sty in LaTeX.}\renewcommand\color[2][]{}}\providecommand\includegraphics[2][]{\GenericError{(gnuplot) \space\space\space\@spaces}{Package graphicx or graphics not loaded}{See the gnuplot documentation for explanation.}{The gnuplot epslatex terminal needs graphicx.sty or graphics.sty.}\renewcommand\includegraphics[2][]{}}\providecommand\rotatebox[2]{#2}\@ifundefined{ifGPcolor}{\newif\ifGPcolor
    \GPcolorfalse
  }{}\@ifundefined{ifGPblacktext}{\newif\ifGPblacktext
    \GPblacktexttrue
  }{}\let\gplgaddtomacro\g@addto@macro
\gdef\gplbacktext{}\gdef\gplfronttext{}\makeatother
  \ifGPblacktext
\def\colorrgb#1{}\def\colorgray#1{}\else
\ifGPcolor
      \def\colorrgb#1{\color[rgb]{#1}}\def\colorgray#1{\color[gray]{#1}}\expandafter\def\csname LTw\endcsname{\color{white}}\expandafter\def\csname LTb\endcsname{\color{black}}\expandafter\def\csname LTa\endcsname{\color{black}}\expandafter\def\csname LT0\endcsname{\color[rgb]{1,0,0}}\expandafter\def\csname LT1\endcsname{\color[rgb]{0,1,0}}\expandafter\def\csname LT2\endcsname{\color[rgb]{0,0,1}}\expandafter\def\csname LT3\endcsname{\color[rgb]{1,0,1}}\expandafter\def\csname LT4\endcsname{\color[rgb]{0,1,1}}\expandafter\def\csname LT5\endcsname{\color[rgb]{1,1,0}}\expandafter\def\csname LT6\endcsname{\color[rgb]{0,0,0}}\expandafter\def\csname LT7\endcsname{\color[rgb]{1,0.3,0}}\expandafter\def\csname LT8\endcsname{\color[rgb]{0.5,0.5,0.5}}\else
\def\colorrgb#1{\color{black}}\def\colorgray#1{\color[gray]{#1}}\expandafter\def\csname LTw\endcsname{\color{white}}\expandafter\def\csname LTb\endcsname{\color{black}}\expandafter\def\csname LTa\endcsname{\color{black}}\expandafter\def\csname LT0\endcsname{\color{black}}\expandafter\def\csname LT1\endcsname{\color{black}}\expandafter\def\csname LT2\endcsname{\color{black}}\expandafter\def\csname LT3\endcsname{\color{black}}\expandafter\def\csname LT4\endcsname{\color{black}}\expandafter\def\csname LT5\endcsname{\color{black}}\expandafter\def\csname LT6\endcsname{\color{black}}\expandafter\def\csname LT7\endcsname{\color{black}}\expandafter\def\csname LT8\endcsname{\color{black}}\fi
  \fi
    \setlength{\unitlength}{0.0500bp}\ifx\gptboxheight\undefined \newlength{\gptboxheight}\newlength{\gptboxwidth}\newsavebox{\gptboxtext}\fi \setlength{\fboxrule}{0.5pt}\setlength{\fboxsep}{1pt}\begin{picture}(4818.00,3174.00)\gplgaddtomacro\gplbacktext{\csname LTb\endcsname \put(1188,550){\makebox(0,0)[r]{\strut{}$10^{1}$}}\put(1188,2623){\makebox(0,0)[r]{\strut{}$10^{2}$}}\put(1320,330){\makebox(0,0){\strut{}$0$}}\put(1663,330){\makebox(0,0){\strut{}$50$}}\put(2006,330){\makebox(0,0){\strut{}$100$}}\put(2349,330){\makebox(0,0){\strut{}$150$}}\put(2692,330){\makebox(0,0){\strut{}$200$}}\put(3036,330){\makebox(0,0){\strut{}$250$}}\put(3379,330){\makebox(0,0){\strut{}$300$}}\put(3722,330){\makebox(0,0){\strut{}$350$}}\put(4065,330){\makebox(0,0){\strut{}$400$}}\put(4408,330){\makebox(0,0){\strut{}$450$}}\put(4751,330){\makebox(0,0){\strut{}$500$}}}\gplgaddtomacro\gplfronttext{\csname LTb\endcsname \put(715,1586){\rotatebox{-270}{\makebox(0,0){\strut{}Thousands of words in heap (log scale)}}}\put(3035,2733){\makebox(0,0){\strut{}SLAM Ideal Memory}}}\gplbacktext
    \put(0,0){\includegraphics{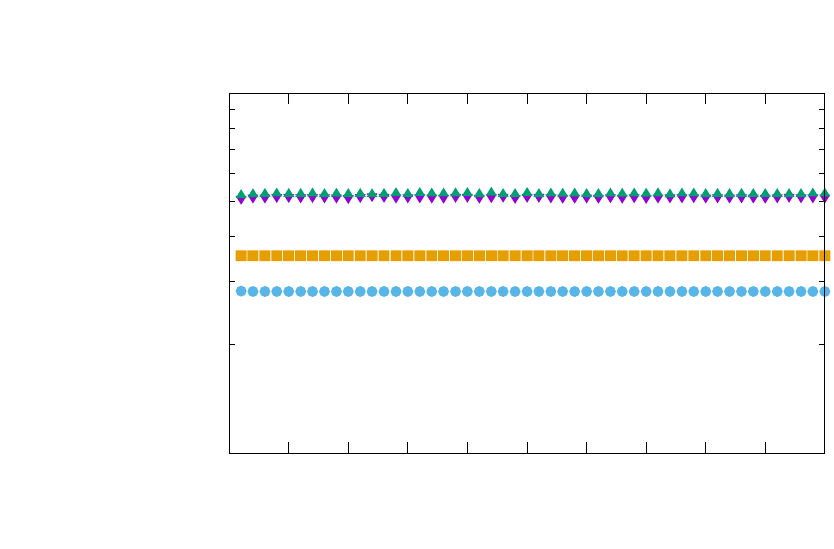}}\gplfronttext
  \end{picture}\endgroup
         \\[-.3cm]
        \hspace*{-1.8em}
        \begingroup
  \makeatletter
  \providecommand\color[2][]{\GenericError{(gnuplot) \space\space\space\@spaces}{Package color not loaded in conjunction with
      terminal option `colourtext'}{See the gnuplot documentation for explanation.}{Either use 'blacktext' in gnuplot or load the package
      color.sty in LaTeX.}\renewcommand\color[2][]{}}\providecommand\includegraphics[2][]{\GenericError{(gnuplot) \space\space\space\@spaces}{Package graphicx or graphics not loaded}{See the gnuplot documentation for explanation.}{The gnuplot epslatex terminal needs graphicx.sty or graphics.sty.}\renewcommand\includegraphics[2][]{}}\providecommand\rotatebox[2]{#2}\@ifundefined{ifGPcolor}{\newif\ifGPcolor
    \GPcolorfalse
  }{}\@ifundefined{ifGPblacktext}{\newif\ifGPblacktext
    \GPblacktexttrue
  }{}\let\gplgaddtomacro\g@addto@macro
\gdef\gplbacktext{}\gdef\gplfronttext{}\makeatother
  \ifGPblacktext
\def\colorrgb#1{}\def\colorgray#1{}\else
\ifGPcolor
      \def\colorrgb#1{\color[rgb]{#1}}\def\colorgray#1{\color[gray]{#1}}\expandafter\def\csname LTw\endcsname{\color{white}}\expandafter\def\csname LTb\endcsname{\color{black}}\expandafter\def\csname LTa\endcsname{\color{black}}\expandafter\def\csname LT0\endcsname{\color[rgb]{1,0,0}}\expandafter\def\csname LT1\endcsname{\color[rgb]{0,1,0}}\expandafter\def\csname LT2\endcsname{\color[rgb]{0,0,1}}\expandafter\def\csname LT3\endcsname{\color[rgb]{1,0,1}}\expandafter\def\csname LT4\endcsname{\color[rgb]{0,1,1}}\expandafter\def\csname LT5\endcsname{\color[rgb]{1,1,0}}\expandafter\def\csname LT6\endcsname{\color[rgb]{0,0,0}}\expandafter\def\csname LT7\endcsname{\color[rgb]{1,0.3,0}}\expandafter\def\csname LT8\endcsname{\color[rgb]{0.5,0.5,0.5}}\else
\def\colorrgb#1{\color{black}}\def\colorgray#1{\color[gray]{#1}}\expandafter\def\csname LTw\endcsname{\color{white}}\expandafter\def\csname LTb\endcsname{\color{black}}\expandafter\def\csname LTa\endcsname{\color{black}}\expandafter\def\csname LT0\endcsname{\color{black}}\expandafter\def\csname LT1\endcsname{\color{black}}\expandafter\def\csname LT2\endcsname{\color{black}}\expandafter\def\csname LT3\endcsname{\color{black}}\expandafter\def\csname LT4\endcsname{\color{black}}\expandafter\def\csname LT5\endcsname{\color{black}}\expandafter\def\csname LT6\endcsname{\color{black}}\expandafter\def\csname LT7\endcsname{\color{black}}\expandafter\def\csname LT8\endcsname{\color{black}}\fi
  \fi
    \setlength{\unitlength}{0.0500bp}\ifx\gptboxheight\undefined \newlength{\gptboxheight}\newlength{\gptboxwidth}\newsavebox{\gptboxtext}\fi \setlength{\fboxrule}{0.5pt}\setlength{\fboxsep}{1pt}\begin{picture}(4818.00,3174.00)\gplgaddtomacro\gplbacktext{\csname LTb\endcsname \put(1188,550){\makebox(0,0)[r]{\strut{}$10^{1}$}}\put(1188,1241){\makebox(0,0)[r]{\strut{}$10^{2}$}}\put(1188,1932){\makebox(0,0)[r]{\strut{}$10^{3}$}}\put(1188,2623){\makebox(0,0)[r]{\strut{}$10^{4}$}}\put(1320,330){\makebox(0,0){\strut{}$0$}}\put(1663,330){\makebox(0,0){\strut{}$50$}}\put(2006,330){\makebox(0,0){\strut{}$100$}}\put(2349,330){\makebox(0,0){\strut{}$150$}}\put(2692,330){\makebox(0,0){\strut{}$200$}}\put(3036,330){\makebox(0,0){\strut{}$250$}}\put(3379,330){\makebox(0,0){\strut{}$300$}}\put(3722,330){\makebox(0,0){\strut{}$350$}}\put(4065,330){\makebox(0,0){\strut{}$400$}}\put(4408,330){\makebox(0,0){\strut{}$450$}}\put(4751,330){\makebox(0,0){\strut{}$500$}}}\gplgaddtomacro\gplfronttext{\csname LTb\endcsname \put(3035,0){\makebox(0,0){\strut{}Step}}\put(3035,2733){\makebox(0,0){\strut{}MTT Ideal Memory}}}\gplbacktext
    \put(0,0){\includegraphics{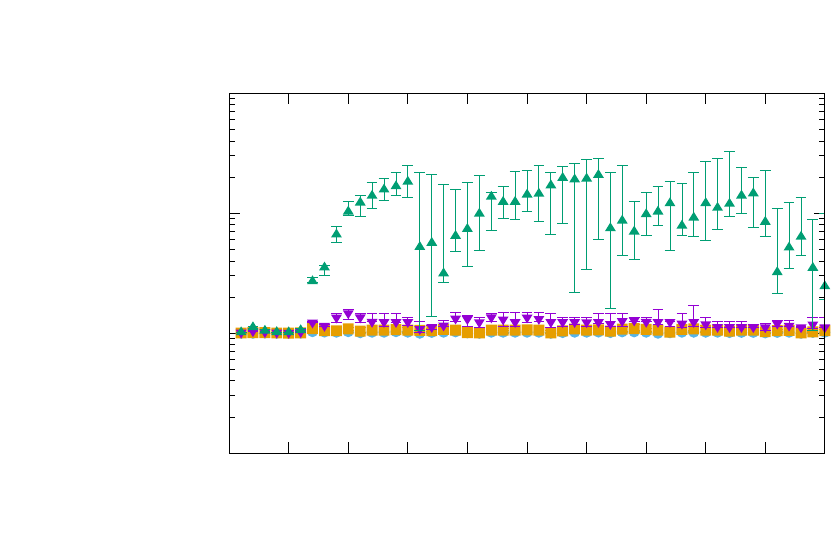}}\gplfronttext
  \end{picture}\endgroup
         \caption{Memory consumption at each step of a run.}
        \label{fig:mem2}
    \end{minipage}
\end{figure*}

 \else
\fi


\begin{thebibliography}{42}


\ifx \showCODEN    \undefined \def \showCODEN     #1{\unskip}     \fi
\ifx \showDOI      \undefined \def \showDOI       #1{#1}\fi
\ifx \showISBNx    \undefined \def \showISBNx     #1{\unskip}     \fi
\ifx \showISBNxiii \undefined \def \showISBNxiii  #1{\unskip}     \fi
\ifx \showISSN     \undefined \def \showISSN      #1{\unskip}     \fi
\ifx \showLCCN     \undefined \def \showLCCN      #1{\unskip}     \fi
\ifx \shownote     \undefined \def \shownote      #1{#1}          \fi
\ifx \showarticletitle \undefined \def \showarticletitle #1{#1}   \fi
\ifx \showURL      \undefined \def \showURL       {\relax}        \fi
\providecommand\bibfield[2]{#2}
\providecommand\bibinfo[2]{#2}
\providecommand\natexlab[1]{#1}
\providecommand\showeprint[2][]{arXiv:#2}

\bibitem[\protect\citeauthoryear{Baudart, Mandel, Atkinson, Sherman, Pouzet,
  and Carbin}{Baudart et~al\mbox{.}}{2020}]%
        {rppl-short}
\bibfield{author}{\bibinfo{person}{Guillaume Baudart}, \bibinfo{person}{Louis
  Mandel}, \bibinfo{person}{Eric Atkinson}, \bibinfo{person}{Benjamin Sherman},
  \bibinfo{person}{Marc Pouzet}, {and} \bibinfo{person}{Michael Carbin}.}
  \bibinfo{year}{2020}\natexlab{}.
\newblock \showarticletitle{Reactive Probabilistic Programming}. In
  \bibinfo{booktitle}{\emph{{PLDI}}}. {ACM}.
\newblock


\bibitem[\protect\citeauthoryear{Benveniste, Caspi, Edwards, Halbwachs, {Le
  Guernic}, and de~Simone}{Benveniste et~al\mbox{.}}{2003}]%
        {synchronous-twelve-years-later}
\bibfield{author}{\bibinfo{person}{Albert Benveniste}, \bibinfo{person}{Paul
  Caspi}, \bibinfo{person}{Stephen~A. Edwards}, \bibinfo{person}{Nicolas
  Halbwachs}, \bibinfo{person}{Paul {Le Guernic}}, {and}
  \bibinfo{person}{Robert de Simone}.} \bibinfo{year}{2003}\natexlab{}.
\newblock \showarticletitle{The synchronous languages 12 years later}.
\newblock \bibinfo{journal}{\emph{Proc. IEEE}} \bibinfo{volume}{91},
  \bibinfo{number}{1} (\bibinfo{year}{2003}), \bibinfo{pages}{64--83}.
\newblock


\bibitem[\protect\citeauthoryear{Bernardin and Stiefelhagen}{Bernardin and
  Stiefelhagen}{2008}]%
        {MOTA08}
\bibfield{author}{\bibinfo{person}{Keni Bernardin} {and}
  \bibinfo{person}{Rainer Stiefelhagen}.} \bibinfo{year}{2008}\natexlab{}.
\newblock \showarticletitle{Evaluating Multiple Object Tracking Performance:
  The {CLEAR} {MOT} Metrics}.
\newblock \bibinfo{journal}{\emph{{EURASIP} J. Image and Video Processing}}
  \bibinfo{volume}{2008} (\bibinfo{year}{2008}).
\newblock


\bibitem[\protect\citeauthoryear{Berry}{Berry}{1989}]%
        {esterel:ifip89}
\bibfield{author}{\bibinfo{person}{G{\'{e}}rard Berry}.}
  \bibinfo{year}{1989}\natexlab{}.
\newblock \showarticletitle{Real Time Programming: Special Purpose or General
  Purpose Languages}. In \bibinfo{booktitle}{\emph{{IFIP} Congress}}.
  \bibinfo{publisher}{North-Holland/IFIP}, \bibinfo{pages}{11--17}.
\newblock


\bibitem[\protect\citeauthoryear{Biernacki, Cola{\c{c}}o, Hamon, and
  Pouzet}{Biernacki et~al\mbox{.}}{2008}]%
        {clock-lctes08}
\bibfield{author}{\bibinfo{person}{Dariusz Biernacki},
  \bibinfo{person}{Jean{-}Louis Cola{\c{c}}o}, \bibinfo{person}{Gr{\'{e}}goire
  Hamon}, {and} \bibinfo{person}{Marc Pouzet}.}
  \bibinfo{year}{2008}\natexlab{}.
\newblock \showarticletitle{Clock-directed modular code generation for
  synchronous data-flow languages}. In \bibinfo{booktitle}{\emph{{LCTES}}}.
  \bibinfo{publisher}{{ACM}}, \bibinfo{pages}{121--130}.
\newblock


\bibitem[\protect\citeauthoryear{Bingham, Chen, Jankowiak, Obermeyer, Pradhan,
  Karaletsos, Singh, Szerlip, Horsfall, and Goodman}{Bingham
  et~al\mbox{.}}{2019}]%
        {BinghamCJOPKSSH19}
\bibfield{author}{\bibinfo{person}{Eli Bingham}, \bibinfo{person}{Jonathan~P.
  Chen}, \bibinfo{person}{Martin Jankowiak}, \bibinfo{person}{Fritz Obermeyer},
  \bibinfo{person}{Neeraj Pradhan}, \bibinfo{person}{Theofanis Karaletsos},
  \bibinfo{person}{Rohit Singh}, \bibinfo{person}{Paul~A. Szerlip},
  \bibinfo{person}{Paul Horsfall}, {and} \bibinfo{person}{Noah~D. Goodman}.}
  \bibinfo{year}{2019}\natexlab{}.
\newblock \showarticletitle{Pyro: Deep Universal Probabilistic Programming}.
\newblock \bibinfo{journal}{\emph{J. Mach. Learn. Res.}}  \bibinfo{volume}{20}
  (\bibinfo{year}{2019}), \bibinfo{pages}{28:1--28:6}.
\newblock


\bibitem[\protect\citeauthoryear{Bourke and Pouzet}{Bourke and Pouzet}{2013a}]%
        {zelus-manual}
\bibfield{author}{\bibinfo{person}{Timothy Bourke} {and} \bibinfo{person}{Marc
  Pouzet}.} \bibinfo{year}{2013}\natexlab{a}.
\newblock \bibinfo{booktitle}{\emph{Z\'elus, a Hybrid Synchronous Language}}.
\newblock \'Ecole normale sup\'erieure.
\newblock
\urldef\tempurl%
\url{http://zelus.di.ens.fr}
\showURL{%
\tempurl}
\newblock
\shownote{Distribution at: {\tt zelus.di.ens.fr}.}


\bibitem[\protect\citeauthoryear{Bourke and Pouzet}{Bourke and Pouzet}{2013b}]%
        {lucy:hscc13}
\bibfield{author}{\bibinfo{person}{Timothy Bourke} {and} \bibinfo{person}{Marc
  Pouzet}.} \bibinfo{year}{2013}\natexlab{b}.
\newblock \showarticletitle{Z{\'{e}}lus: a synchronous language with {ODEs}}.
  In \bibinfo{booktitle}{\emph{{HSCC}}}. \bibinfo{publisher}{{ACM}},
  \bibinfo{pages}{113--118}.
\newblock


\bibitem[\protect\citeauthoryear{Broderick, Boyd, Wibisono, Wilson, and
  Jordan}{Broderick et~al\mbox{.}}{2013}]%
        {BroderickBWWJ13}
\bibfield{author}{\bibinfo{person}{Tamara Broderick}, \bibinfo{person}{Nicholas
  Boyd}, \bibinfo{person}{Andre Wibisono}, \bibinfo{person}{Ashia~C. Wilson},
  {and} \bibinfo{person}{Michael~I. Jordan}.} \bibinfo{year}{2013}\natexlab{}.
\newblock \showarticletitle{Streaming Variational Bayes}. In
  \bibinfo{booktitle}{\emph{{NIPS}}}. \bibinfo{pages}{1727--1735}.
\newblock


\bibitem[\protect\citeauthoryear{Carpenter, Gelman, Hoffman, Lee, Goodrich,
  Betancourt, Brubaker, Guo, Li, and Riddell}{Carpenter et~al\mbox{.}}{2017}]%
        {carpenter2017stan}
\bibfield{author}{\bibinfo{person}{Bob Carpenter}, \bibinfo{person}{Andrew
  Gelman}, \bibinfo{person}{Matthew~D Hoffman}, \bibinfo{person}{Daniel Lee},
  \bibinfo{person}{Ben Goodrich}, \bibinfo{person}{Michael Betancourt},
  \bibinfo{person}{Marcus Brubaker}, \bibinfo{person}{Jiqiang Guo},
  \bibinfo{person}{Peter Li}, {and} \bibinfo{person}{Allen Riddell}.}
  \bibinfo{year}{2017}\natexlab{}.
\newblock \showarticletitle{Stan: A probabilistic programming language}.
\newblock \bibinfo{journal}{\emph{J. Statistical Software}}
  \bibinfo{volume}{76}, \bibinfo{number}{1} (\bibinfo{year}{2017}),
  \bibinfo{pages}{1--37}.
\newblock


\bibitem[\protect\citeauthoryear{Caspi}{Caspi}{1992}]%
        {caspi92}
\bibfield{author}{\bibinfo{person}{Paul Caspi}.}
  \bibinfo{year}{1992}\natexlab{}.
\newblock \showarticletitle{Clocks in Dataflow Languages}.
\newblock \bibinfo{journal}{\emph{Theor. Comput. Sci.}} \bibinfo{volume}{94},
  \bibinfo{number}{1} (\bibinfo{year}{1992}), \bibinfo{pages}{125--140}.
\newblock


\bibitem[\protect\citeauthoryear{Caspi and Pouzet}{Caspi and Pouzet}{1998}]%
        {pouzet-cmcs98}
\bibfield{author}{\bibinfo{person}{Paul Caspi} {and} \bibinfo{person}{Marc
  Pouzet}.} \bibinfo{year}{1998}\natexlab{}.
\newblock \showarticletitle{A Co-iterative Characterization of Synchronous
  Stream Functions}. In \bibinfo{booktitle}{\emph{{CMCS}}}
  \emph{(\bibinfo{series}{Electronic Notes in Theoretical Computer Science})},
  Vol.~\bibinfo{volume}{11}. \bibinfo{publisher}{Elsevier},
  \bibinfo{pages}{1--21}.
\newblock


\bibitem[\protect\citeauthoryear{Cola{\c{c}}o, Hamon, and Pouzet}{Cola{\c{c}}o
  et~al\mbox{.}}{2006}]%
        {automaton-emsoft06}
\bibfield{author}{\bibinfo{person}{Jean{-}Louis Cola{\c{c}}o},
  \bibinfo{person}{Gr{\'{e}}goire Hamon}, {and} \bibinfo{person}{Marc Pouzet}.}
  \bibinfo{year}{2006}\natexlab{}.
\newblock \showarticletitle{Mixing signals and modes in synchronous data-flow
  systems}. In \bibinfo{booktitle}{\emph{{EMSOFT}}}.
  \bibinfo{publisher}{{ACM}}, \bibinfo{pages}{73--82}.
\newblock


\bibitem[\protect\citeauthoryear{Cola{\c{c}}o, Pagano, and Pouzet}{Cola{\c{c}}o
  et~al\mbox{.}}{2017}]%
        {lucy:tase17}
\bibfield{author}{\bibinfo{person}{Jean{-}Louis Cola{\c{c}}o},
  \bibinfo{person}{Bruno Pagano}, {and} \bibinfo{person}{Marc Pouzet}.}
  \bibinfo{year}{2017}\natexlab{}.
\newblock \showarticletitle{{SCADE} 6: {A} formal language for embedded
  critical software development (invited paper)}. In
  \bibinfo{booktitle}{\emph{{TASE}}}. \bibinfo{publisher}{{IEEE} Computer
  Society}, \bibinfo{pages}{1--11}.
\newblock


\bibitem[\protect\citeauthoryear{Cola{\c{c}}o and Pouzet}{Cola{\c{c}}o and
  Pouzet}{2004}]%
        {pouzet-initialization}
\bibfield{author}{\bibinfo{person}{Jean{-}Louis Cola{\c{c}}o} {and}
  \bibinfo{person}{Marc Pouzet}.} \bibinfo{year}{2004}\natexlab{}.
\newblock \showarticletitle{Type-based initialization analysis of a synchronous
  dataflow language}.
\newblock \bibinfo{journal}{\emph{Int. J. Softw. Tools Technol. Transf.}}
  \bibinfo{volume}{6}, \bibinfo{number}{3} (\bibinfo{year}{2004}),
  \bibinfo{pages}{245--255}.
\newblock


\bibitem[\protect\citeauthoryear{Del~Moral, Doucet, and Jasra}{Del~Moral
  et~al\mbox{.}}{2006}]%
        {doucet-smc-2006}
\bibfield{author}{\bibinfo{person}{Pierre Del~Moral}, \bibinfo{person}{Arnaud
  Doucet}, {and} \bibinfo{person}{Ajay Jasra}.}
  \bibinfo{year}{2006}\natexlab{}.
\newblock \showarticletitle{Sequential {Monte} {Carlo} samplers}.
\newblock \bibinfo{journal}{\emph{J. Royal Statistical Society: Series B
  (Statistical Methodology)}} \bibinfo{volume}{68}, \bibinfo{number}{3}
  (\bibinfo{year}{2006}), \bibinfo{pages}{411--436}.
\newblock


\bibitem[\protect\citeauthoryear{Doucet, de~Freitas, Murphy, and
  Russell}{Doucet et~al\mbox{.}}{2000}]%
        {rbpf}
\bibfield{author}{\bibinfo{person}{Arnaud Doucet}, \bibinfo{person}{Nando de
  Freitas}, \bibinfo{person}{Kevin~P. Murphy}, {and} \bibinfo{person}{Stuart~J.
  Russell}.} \bibinfo{year}{2000}\natexlab{}.
\newblock \showarticletitle{{Rao-Blackwellised} Particle Filtering for Dynamic
  Bayesian Networks}. In \bibinfo{booktitle}{\emph{{UAI}}}.
  \bibinfo{publisher}{Morgan Kaufmann}, \bibinfo{pages}{176--183}.
\newblock


\bibitem[\protect\citeauthoryear{Fink}{Fink}{1997}]%
        {conjprior}
\bibfield{author}{\bibinfo{person}{Daniel Fink}.}
  \bibinfo{year}{1997}\natexlab{}.
\newblock \bibinfo{title}{A Compendium of Conjugate Priors}.
\newblock
\newblock


\bibitem[\protect\citeauthoryear{Gehr, Misailovic, and Vechev}{Gehr
  et~al\mbox{.}}{2016}]%
        {psi}
\bibfield{author}{\bibinfo{person}{Timon Gehr}, \bibinfo{person}{Sasa
  Misailovic}, {and} \bibinfo{person}{Martin~T. Vechev}.}
  \bibinfo{year}{2016}\natexlab{}.
\newblock \showarticletitle{{PSI:} Exact Symbolic Inference for Probabilistic
  Programs}. In \bibinfo{booktitle}{\emph{{CAV} {(1)}}}
  \emph{(\bibinfo{series}{Lecture Notes in Computer Science})},
  Vol.~\bibinfo{volume}{9779}. \bibinfo{publisher}{Springer},
  \bibinfo{pages}{62--83}.
\newblock


\bibitem[\protect\citeauthoryear{Goodman and Stuhlm{\"u}ller}{Goodman and
  Stuhlm{\"u}ller}{2014}]%
        {goodman_stuhlmuller_2014}
\bibfield{author}{\bibinfo{person}{Noah~D. Goodman} {and}
  \bibinfo{person}{Andreas Stuhlm{\"u}ller}.} \bibinfo{year}{2014}\natexlab{}.
\newblock \bibinfo{title}{The Design and Implementation of Probabilistic
  Programming Languages}.
\newblock
\newblock
\urldef\tempurl%
\url{http://dippl.org}
\showURL{%
\tempurl}
\newblock
\shownote{Accessed April 2020.}


\bibitem[\protect\citeauthoryear{{Gordon}, {Salmond}, and {Smith}}{{Gordon}
  et~al\mbox{.}}{1993}]%
        {particlefilter}
\bibfield{author}{\bibinfo{person}{N.~J. {Gordon}}, \bibinfo{person}{D.~J.
  {Salmond}}, {and} \bibinfo{person}{A.~F.~M. {Smith}}.}
  \bibinfo{year}{1993}\natexlab{}.
\newblock \showarticletitle{Novel approach to nonlinear/non-Gaussian Bayesian
  state estimation}.
\newblock \bibinfo{journal}{\emph{IEE Proceedings F - Radar and Signal
  Processing}} \bibinfo{volume}{140}, \bibinfo{number}{2},
  \bibinfo{pages}{107--113}.
\newblock


\bibitem[\protect\citeauthoryear{Halbwachs, Caspi, Raymond, and
  Pilaud}{Halbwachs et~al\mbox{.}}{1991}]%
        {lustre:ieee91}
\bibfield{author}{\bibinfo{person}{N. Halbwachs}, \bibinfo{person}{P. Caspi},
  \bibinfo{person}{P. Raymond}, {and} \bibinfo{person}{D. Pilaud}.}
  \bibinfo{year}{1991}\natexlab{}.
\newblock \showarticletitle{The Synchronous Dataflow Programming Language {\sc
  Lustre}}.
\newblock \bibinfo{journal}{\emph{Proc. IEEE}} \bibinfo{volume}{79},
  \bibinfo{number}{9} (\bibinfo{date}{September} \bibinfo{year}{1991}),
  \bibinfo{pages}{1305--1320}.
\newblock


\bibitem[\protect\citeauthoryear{Huang, Tristan, and Morrisett}{Huang
  et~al\mbox{.}}{2017}]%
        {HuangTM17}
\bibfield{author}{\bibinfo{person}{Daniel Huang},
  \bibinfo{person}{Jean{-}Baptiste Tristan}, {and} \bibinfo{person}{Greg
  Morrisett}.} \bibinfo{year}{2017}\natexlab{}.
\newblock \showarticletitle{Compiling Markov chain Monte Carlo algorithms for
  probabilistic modeling}. In \bibinfo{booktitle}{\emph{{PLDI}}}.
  \bibinfo{publisher}{{ACM}}, \bibinfo{pages}{111--125}.
\newblock


\bibitem[\protect\citeauthoryear{Kahn}{Kahn}{1974}]%
        {gilles1974semantics}
\bibfield{author}{\bibinfo{person}{Gilles Kahn}.}
  \bibinfo{year}{1974}\natexlab{}.
\newblock \showarticletitle{The Semantics of a Simple Language for Parallel
  Programming}. In \bibinfo{booktitle}{\emph{{IFIP} Congress}}.
  \bibinfo{publisher}{North-Holland}, \bibinfo{pages}{471--475}.
\newblock


\bibitem[\protect\citeauthoryear{Kalman}{Kalman}{1960}]%
        {kalman}
\bibfield{author}{\bibinfo{person}{Rudolph~Emil Kalman}.}
  \bibinfo{year}{1960}\natexlab{}.
\newblock \showarticletitle{A New Approach to Linear Filtering and Prediction
  Problems}.
\newblock \bibinfo{journal}{\emph{Journal of Basic Engineering}}
  \bibinfo{volume}{82}, \bibinfo{number}{1} (\bibinfo{date}{03}
  \bibinfo{year}{1960}), \bibinfo{pages}{35--45}.
\newblock
\showISSN{0021-9223}


\bibitem[\protect\citeauthoryear{Lund{\'{e}}n}{Lund{\'{e}}n}{2017}]%
        {lunden17}
\bibfield{author}{\bibinfo{person}{Daniel Lund{\'{e}}n}.}
  \bibinfo{year}{2017}\natexlab{}.
\newblock \emph{\bibinfo{title}{Delayed sampling in the probabilistic
  programming language {Anglican}}}.
\newblock \bibinfo{thesistype}{Master's\ thesis}. \bibinfo{school}{KTH Royal
  Institute of Technology}.
\newblock


\bibitem[\protect\citeauthoryear{Lund{\'{e}}n, Broman, Ronquist, and
  Murray}{Lund{\'{e}}n et~al\mbox{.}}{2018}]%
        {lunden19alignment}
\bibfield{author}{\bibinfo{person}{Daniel Lund{\'{e}}n}, \bibinfo{person}{David
  Broman}, \bibinfo{person}{Fredrik Ronquist}, {and}
  \bibinfo{person}{Lawrence~M. Murray}.} \bibinfo{year}{2018}\natexlab{}.
\newblock \showarticletitle{Automatic Alignment of Sequential {Monte} {Carlo}
  Inference in Higher-Order Probabilistic Programs}.
\newblock \bibinfo{journal}{\emph{CoRR}}  \bibinfo{volume}{abs/1812.07439}
  (\bibinfo{year}{2018}).
\newblock


\bibitem[\protect\citeauthoryear{Lunn, Spiegelhalter, Thomas, and Best}{Lunn
  et~al\mbox{.}}{2009}]%
        {lunn2009bugs}
\bibfield{author}{\bibinfo{person}{David Lunn}, \bibinfo{person}{David
  Spiegelhalter}, \bibinfo{person}{Andrew Thomas}, {and} \bibinfo{person}{Nicky
  Best}.} \bibinfo{year}{2009}\natexlab{}.
\newblock \showarticletitle{The {BUGS} project: Evolution, critique and future
  directions}.
\newblock \bibinfo{journal}{\emph{Statistics in medicine}}
  \bibinfo{volume}{28}, \bibinfo{number}{25} (\bibinfo{year}{2009}),
  \bibinfo{pages}{3049--3067}.
\newblock


\bibitem[\protect\citeauthoryear{Minka}{Minka}{2001}]%
        {ep}
\bibfield{author}{\bibinfo{person}{Thomas~P. Minka}.}
  \bibinfo{year}{2001}\natexlab{}.
\newblock \showarticletitle{Expectation Propagation for approximate Bayesian
  inference}. In \bibinfo{booktitle}{\emph{{UAI}}}. \bibinfo{publisher}{Morgan
  Kaufmann}, \bibinfo{pages}{362--369}.
\newblock


\bibitem[\protect\citeauthoryear{Montemerlo, Thrun, Koller, and
  Wegbreit}{Montemerlo et~al\mbox{.}}{2002}]%
        {Montemerlo02-fastslam}
\bibfield{author}{\bibinfo{person}{Michael Montemerlo},
  \bibinfo{person}{Sebastian Thrun}, \bibinfo{person}{Daphne Koller}, {and}
  \bibinfo{person}{Ben Wegbreit}.} \bibinfo{year}{2002}\natexlab{}.
\newblock \showarticletitle{FastSLAM: {A} Factored Solution to the Simultaneous
  Localization and Mapping Problem}. In
  \bibinfo{booktitle}{\emph{{AAAI/IAAI}}}. \bibinfo{publisher}{{AAAI} Press /
  The {MIT} Press}, \bibinfo{pages}{593--598}.
\newblock


\bibitem[\protect\citeauthoryear{Murray, Lund{\'{e}}n, Kudlicka, Broman, and
  Sch{\"{o}}n}{Murray et~al\mbox{.}}{2018}]%
        {murray18delayed_sampling}
\bibfield{author}{\bibinfo{person}{Lawrence~M. Murray}, \bibinfo{person}{Daniel
  Lund{\'{e}}n}, \bibinfo{person}{Jan Kudlicka}, \bibinfo{person}{David
  Broman}, {and} \bibinfo{person}{Thomas~B. Sch{\"{o}}n}.}
  \bibinfo{year}{2018}\natexlab{}.
\newblock \showarticletitle{Delayed Sampling and Automatic Rao-Blackwellization
  of Probabilistic Programs}. In \bibinfo{booktitle}{\emph{{AISTATS}}}
  \emph{(\bibinfo{series}{Proceedings of Machine Learning Research})},
  Vol.~\bibinfo{volume}{84}. \bibinfo{publisher}{{PMLR}},
  \bibinfo{pages}{1037--1046}.
\newblock


\bibitem[\protect\citeauthoryear{Murray and Sch{\"{o}}n}{Murray and
  Sch{\"{o}}n}{2018}]%
        {MurrayS18}
\bibfield{author}{\bibinfo{person}{Lawrence~M. Murray} {and}
  \bibinfo{person}{Thomas~B. Sch{\"{o}}n}.} \bibinfo{year}{2018}\natexlab{}.
\newblock \showarticletitle{Automated learning with a probabilistic programming
  language: Birch}.
\newblock \bibinfo{journal}{\emph{Annual Reviews in Control}}
  \bibinfo{volume}{46} (\bibinfo{year}{2018}), \bibinfo{pages}{29--43}.
\newblock


\bibitem[\protect\citeauthoryear{Narayanan, Carette, Romano, Shan, and
  Zinkov}{Narayanan et~al\mbox{.}}{2016}]%
        {hakaru}
\bibfield{author}{\bibinfo{person}{Praveen Narayanan}, \bibinfo{person}{Jacques
  Carette}, \bibinfo{person}{Wren Romano}, \bibinfo{person}{Chung{-}chieh
  Shan}, {and} \bibinfo{person}{Robert Zinkov}.}
  \bibinfo{year}{2016}\natexlab{}.
\newblock \showarticletitle{Probabilistic Inference by Program Transformation
  in {Hakaru} (System Description)}. In \bibinfo{booktitle}{\emph{{FLOPS}}}
  \emph{(\bibinfo{series}{Lecture Notes in Computer Science})},
  Vol.~\bibinfo{volume}{9613}. \bibinfo{publisher}{Springer},
  \bibinfo{pages}{62--79}.
\newblock


\bibitem[\protect\citeauthoryear{Pfeffer}{Pfeffer}{2005}]%
        {Pfeffer05}
\bibfield{author}{\bibinfo{person}{Avi Pfeffer}.}
  \bibinfo{year}{2005}\natexlab{}.
\newblock \showarticletitle{Functional Specification of Probabilistic Process
  Models}. In \bibinfo{booktitle}{\emph{{AAAI}}}. \bibinfo{publisher}{{AAAI}
  Press / The {MIT} Press}, \bibinfo{pages}{663--669}.
\newblock


\bibitem[\protect\citeauthoryear{Pfeffer}{Pfeffer}{2009}]%
        {Pfeffer09}
\bibfield{author}{\bibinfo{person}{Avi Pfeffer}.}
  \bibinfo{year}{2009}\natexlab{}.
\newblock \showarticletitle{{CTPPL:} {A} Continuous Time Probabilistic
  Programming Language}. In \bibinfo{booktitle}{\emph{{IJCAI}}}.
  \bibinfo{pages}{1943--1950}.
\newblock


\bibitem[\protect\citeauthoryear{Raymond, Roux, and Jahier}{Raymond
  et~al\mbox{.}}{2008}]%
        {raymond-lutin-2008}
\bibfield{author}{\bibinfo{person}{Pascal Raymond}, \bibinfo{person}{Yvan
  Roux}, {and} \bibinfo{person}{Erwan Jahier}.}
  \bibinfo{year}{2008}\natexlab{}.
\newblock \showarticletitle{{Lutin}: {A} Language for Specifying and Executing
  Reactive Scenarios}.
\newblock \bibinfo{journal}{\emph{{EURASIP} Journal of Embedded Sytems}}
  \bibinfo{volume}{2008} (\bibinfo{year}{2008}).
\newblock


\bibitem[\protect\citeauthoryear{Ritchie, Stuhlm{\"{u}}ller, and
  Goodman}{Ritchie et~al\mbox{.}}{2016}]%
        {ritchie_c3_2016}
\bibfield{author}{\bibinfo{person}{Daniel Ritchie}, \bibinfo{person}{Andreas
  Stuhlm{\"{u}}ller}, {and} \bibinfo{person}{Noah~D. Goodman}.}
  \bibinfo{year}{2016}\natexlab{}.
\newblock \showarticletitle{{C3:} Lightweight Incrementalized {MCMC} for
  Probabilistic Programs using Continuations and Callsite Caching}. In
  \bibinfo{booktitle}{\emph{{AISTATS}}} \emph{(\bibinfo{series}{{JMLR} Workshop
  and Conference Proceedings})}, Vol.~\bibinfo{volume}{51}.
  \bibinfo{publisher}{JMLR.org}, \bibinfo{pages}{28--37}.
\newblock


\bibitem[\protect\citeauthoryear{Sontag}{Sontag}{2013}]%
        {sontag13control}
\bibfield{author}{\bibinfo{person}{Eduardo~D Sontag}.}
  \bibinfo{year}{2013}\natexlab{}.
\newblock \bibinfo{booktitle}{\emph{Mathematical control theory: deterministic
  finite dimensional systems}}. Vol.~\bibinfo{volume}{6}.
\newblock \bibinfo{publisher}{Springer Science \& Business Media}.
\newblock


\bibitem[\protect\citeauthoryear{Staton}{Staton}{2017}]%
        {staton17}
\bibfield{author}{\bibinfo{person}{Sam Staton}.}
  \bibinfo{year}{2017}\natexlab{}.
\newblock \showarticletitle{Commutative Semantics for Probabilistic
  Programming}. In \bibinfo{booktitle}{\emph{{ESOP}}}
  \emph{(\bibinfo{series}{Lecture Notes in Computer Science})},
  Vol.~\bibinfo{volume}{10201}. \bibinfo{publisher}{Springer},
  \bibinfo{pages}{855--879}.
\newblock


\bibitem[\protect\citeauthoryear{Tolpin, van~de Meent, Yang, and Wood}{Tolpin
  et~al\mbox{.}}{2016}]%
        {tolpin_et_al_2016}
\bibfield{author}{\bibinfo{person}{David Tolpin}, \bibinfo{person}{Jan{-}Willem
  van~de Meent}, \bibinfo{person}{Hongseok Yang}, {and}
  \bibinfo{person}{Frank~D. Wood}.} \bibinfo{year}{2016}\natexlab{}.
\newblock \showarticletitle{Design and Implementation of Probabilistic
  Programming Language {Anglican}}. In \bibinfo{booktitle}{\emph{{IFL}}}.
  \bibinfo{publisher}{{ACM}}, \bibinfo{pages}{6:1--6:12}.
\newblock


\bibitem[\protect\citeauthoryear{Tran, Hoffman, Saurous, Brevdo, Murphy, and
  Blei}{Tran et~al\mbox{.}}{2017}]%
        {TranHSB0B17}
\bibfield{author}{\bibinfo{person}{Dustin Tran}, \bibinfo{person}{Matthew~D.
  Hoffman}, \bibinfo{person}{Rif~A. Saurous}, \bibinfo{person}{Eugene Brevdo},
  \bibinfo{person}{Kevin Murphy}, {and} \bibinfo{person}{David~M. Blei}.}
  \bibinfo{year}{2017}\natexlab{}.
\newblock \showarticletitle{Deep Probabilistic Programming}. In
  \bibinfo{booktitle}{\emph{{ICLR} (Poster)}}.
  \bibinfo{publisher}{OpenReview.net}.
\newblock


\bibitem[\protect\citeauthoryear{Wu, Li, Russell, and Bod{\'{\i}}k}{Wu
  et~al\mbox{.}}{2016}]%
        {wu2016SwiftCI}
\bibfield{author}{\bibinfo{person}{Yi Wu}, \bibinfo{person}{Lei Li},
  \bibinfo{person}{Stuart~J. Russell}, {and} \bibinfo{person}{Rastislav
  Bod{\'{\i}}k}.} \bibinfo{year}{2016}\natexlab{}.
\newblock \showarticletitle{Swift: Compiled Inference for Probabilistic
  Programming Languages}. In \bibinfo{booktitle}{\emph{{IJCAI}}}.
  \bibinfo{publisher}{{IJCAI/AAAI} Press}, \bibinfo{pages}{3637--3645}.
\newblock


\end{thebibliography}
\end{document}